\documentclass[11pt]{article}

\usepackage[left=2cm, right=2cm, top=2cm, bottom=3.5cm]{geometry}

\usepackage[table]{xcolor}
\definecolor{wineRed}{RGB}{114, 17, 34}
\definecolor{vertceladon}{RGB}{74, 124, 89}

\usepackage{graphicx} 
\usepackage{subcaption}
\usepackage{booktabs}
\usepackage{tikz}
\usepackage{quantikz}
\usepackage{enumerate}
\usepackage{enumitem}

\usepackage{titlesec}
\renewcommand{\thesection}{\Roman{section}}
\renewcommand{\thesubsection}{\Roman{section}.\arabic{subsection}}
\titleformat{\section}
  {\color{wineRed}\normalfont\Large\bfseries}  
{\thesection}{1em}{}

\usepackage{tocloft}
\usepackage{amsthm, mathrsfs, color, mathtools, amsmath, amssymb, mathtools, float, verbatim, xcolor, placeins, needspace, blkarray}
\theoremstyle{definition}
\newtheorem{theorem}{Theorem}
\newtheorem{lemma}{Lemma}
\newtheorem{remark}{Remark}

\newtheorem{definition}{Definition}

\usepackage{algorithm}
\usepackage{algpseudocode}
\usepackage{verbatim}

\usepackage{hyperref}
\hypersetup{
    colorlinks=true,       
    allcolors=wineRed        
}
\usepackage[style=alphabetic,maxalphanames=4,minalphanames=3,minbibnames=5,maxbibnames=5]{biblatex}
\usepackage{url}

\usepackage{authblk}
  
\date{} 

\newcommand{\diro}{Department of Computer Science and Operations Research, Universit{\' e} de Montr{\' e}al, Montr{\' e}al, Qu{\' e}bec, Canada}

\newcommand{\courtois}{Institut Courtois, Universit{\' e} de Montr{\' e}al, Montr{\' e}al, Qu{\' e}bec, Canada}

\newcommand{\mila}{Mila -- Qu{\' e}bec AI Institute, Montr{\' e}al, Qu{\' e}bec, Canada}

\usepackage{etoolbox}
\apptocmd{\endabstract}{\vspace{1.5em}}{}{}

\renewcommand{\theequation}{\arabic{equation}}

\makeatletter
\def\tagform@#1{\maketag@@@{(\ignorespaces#1\unskip\@@italiccorr)}}
\makeatother
\title{\textbf{Linearised quantum signal processing}}

\author[1,2,3]{Marek Arsenault}

\author[1,2,3]{Hl\'er Kristj\'ansson}

\affil[1]{\diro}
\affil[2]{\courtois}
\affil[3]{\mila}

\begin{document}
\maketitle

\begin{abstract}
Quantum functional programming has been developed through two distinct paradigms in the last few years: Quantum Signal Processing (QSP)-based methods, including the Quantum Singular Value Transformation (QSVT) \cite{gilyen2019quantum}, and methods based on higher-order quantum transformations, such as the Universal Hamiltonian Eigenvalue Transformation (UHET) \cite{odake2025universal}. While UHET performs functional transformations of Hamiltonian dynamics, its relationship to QSP-based techniques has remained unclear despite evident structural similarities. In this work, we resolve this gap by establishing a connection between UHET and QSP-based frameworks; specifically, we show that UHET can be interpreted as a (randomised) linearisation of Generalised QSP (GQSP) \cite{motlagh2024generalized}. Building on this result, we introduce a linearised variant of (Hamiltonian-based) QSVT, which we call Universal Hamiltonian Singular Value Transformation (UHSVT), that enables the efficient transformation of the singular values of any arbitrary matrix $A$ encoded in a block of a Hamiltonian, whose dynamics is accessible as a black box, by any sufficiently differentiable complex-valued function $f$. Our algorithm requires the sole condition that $f$ vanishes at the origin, in contrast to previous QSVT-based approaches that assumed either a lower bound on the singular values of $A$ or the ability to perform $X$-rotation gates on the induced two-dimensional `qubitised' subspace.

\end{abstract}

\tableofcontents

\section{Introduction}

Understanding the core structure underlying quantum algorithms has long been a central objective of the field of quantum computation. Such structural insights help clarify the origin of quantum advantage and often lead to the development of new and improved algorithms. One of the earliest unifying paradigms was the hidden subgroup framework \cite{hallgren2003hidden}, rooted in the quantum Fourier transform, which provided a common structure for several foundational quantum algorithms. 
In recent years, a new unification encompassing most of the known quantum algorithms has emerged through the framework of \textit{Quantum Signal Processing} (QSP) \cite{low2017optimal,gilyen2019quantum,martyn2021grand}, which enables the manipulation of information encoded in arbitrary matrices, thus initiating a functional programming scheme for quantum algorithms \cite{rossi2025modular}. 

A broad class of QSP-type algorithms follows a common blueprint: (1) encode a matrix $A$ into a block of a larger unitary operator $U[A]$, and (2) interleave this block-encoding repeatedly with simple rotation unitaries known as processing operators. This procedure induces a polynomial transformation of $A$, which appears in a designated sub-block of the resulting unitary operator. This framework encompasses a variety of algorithms, differing in their choice of block-encodings and processing operators, yet with a remarkably unified underlying structure. For example, when $A$ is non-unitary, the Quantum Singular Value Transformation (QSVT) \cite{gilyen2019quantum} enables polynomial transformations of its singular values. When $A$ is unitary, Generalised QSP (GQSP) \cite{motlagh2024generalized} or QET-U \cite{dong2022ground} 
allows polynomial transformations of its eigenvalues. 

At a structural level, these methods are understood through \textit{qubitisation}: 
a block-encoding can be decomposed into invariant two-dimensional subspaces, each behaving like a qubit. Applying QSP independently within each such subspace yields a polynomial transformation of the associated eigenvalues or singular values, thereby implementing a functional transformation of the original operator. Recent progress has been made in understanding the relationship between all the different QSP variants \cite{laneve2025generalized} as well in finding a connection between GQSP and the nonlinear Fourier transform (NLFT) \cite{alexis2024quantum,laneve2025generalized}.

In parallel, a different approach to quantum functional programming has been built on \textit{higher-order quantum transformations} \cite{chiribella2008transforming,chiribella2008quantum,chiribella2009theoretical}. Such algorithms enable the transformations of unitary operations given as a black box, for example transforming a unitary $U$ to its inverse, transpose, complex conjugate or controlled version \cite{miyazaki2019complex,quintino2019probabilistic,sedlak2019optimal,quintino2022deterministic,quintino2019reversing,yoshida2023reversing,mo2024parameterized,chen2024quantum,dong2019controlled,chiribella2019quantum}, as well as transformations of the underlying Hamiltonian $H$ \cite{odake2024higher,odake2025universal}.

In particular, the Universal Hamiltonian Eigenvalue Transformation (UHET) algorithm \cite{odake2025universal} operates in a setting where one has access to the (control-free) dynamics of a Hamiltonian $H$ through the evolution operator $e^{-iH\tau}$ for arbitrary time $\tau$. UHET achieves a functional transformation of the eigenvalues of the original Hamiltonian, by any sufficiently differentiable function $f$, mapping the evolution operator $e^{-iH\tau}$ to $e^{-if(H)t}$. The structure of UHET appears immediately familiar to those accustomed to QSP-like frameworks: controlled evolution operators are interleaved with SU(2) rotations on an auxiliary qubit. However, while QSP carefully selects phase angles to coherently engineer a target polynomial within a sub-block of the resulting unitary, UHET relies instead on a conceptually simpler, yet powerful, randomised Trotter-Suzuki-based construction (analogous to qDRIFT \cite{campbell2018random}). 

Due to this randomised Trotter-Suzuki-based structure, the circuit depth of UHET typically scales as 
$\mathcal{O}(1/\epsilon)$ with the target error $\epsilon$, whereas standard QSPapproaches achieve logarithmic scaling in $1/\mathcal{\epsilon}$ for smooth functions and $\frac{1}{\epsilon^{1/(J-1)} }$ for $J$-smooth functions. Nevertheless UHET can outperform QSP-based algorithms  in settings where only the control-free Hamiltonian evolution is available, such that the controlled evolution itself must also be implemented via a randomised Trotter-Suzuki subroutine \cite{dong2019controlled}. This advantage arises from a \textit{compilation procedure} that leverages correlated randomness to suppress the overall error of concatenated randomised subroutines \cite{odake2025universal}.

An important question that naturally arises when considering UHET is how it fits into the broader ``map" of quantum functional algorithms. What is its precise relationship to QSP and QSVT? Understanding the structural connections between these different approaches is crucial to developing a complete theory of quantum functional programming that can lead to the realisation of new algorithmic primitives. 

In this work, we deliver on both of these desiderata. First, we establish a relationship between 
UHET and Generalised QSP (GQSP) \cite{motlagh2024generalized} showing that UHET can be seen as a \textit{linearisation} of GQSP. Secondly, we apply the same linearisation procedure to (the Hamiltonian version of) QSVT \cite{lloyd2104hamiltonian}, to develop a new algorithm for transforming the singular values of arbitrary matrices block-encoded in a Hamiltonian, whose dynamics is available as a black box, which we call Universal Hamiltonian Singular Value Transformation (UHSVT). We compare our algorithm to existing QSVT-based methods for the same task \cite{lloyd2104hamiltonian,kang2025quantum}, as well as to a new QSVT-based algorithm, which we construct by combining the advantages of the existing QSVT-based methods. Unlike the QSVT-based methods, UHSVT does not require a lower bound on the singular values of $A$ nor the ability to perform $X$-gates on the qubitised subspace, enabling the possibility of a truly black-box transformation. 

The remainder of this paper is organised as follows.  Section~\ref{sec: Summary of main results}  summarises the main results. Section~\ref{sec: Overview of higher-order quantum algorithms for Hamiltonian dynamics} provides a brief overview of Trotter-Suzuki Hamiltonian simulation and the UHET algorithm. Section~\ref{sec: Overview of QSP-based functional transformations} gives a pedagogical introduction to QSP across different settings, and demonstrates how it can be lifted to transform block-encoded matrices via qubitisation. This includes Generalised QSP (GQSP) and the Hamiltonian-based Quantum Singular Value Transformation (H-QSVT). In Section \ref{sec: Universal Hamiltonian Eigenvalue Transformation as a linearisation of Generalised QSP}, we show how UHET can arise as a linearisation of GQSP, while in Section \ref{sec: Universal Hamiltonian Singular Value Transformation} we derive an analogous linearisation of H-QSVT, yielding the new UHSVT algorithm, which we compare with existing (and new) QSVT-based methods. Section \ref{sec:conclusion} concludes.

\begin{figure}[htb]
    \centering
    \includegraphics[width=0.80\textwidth]{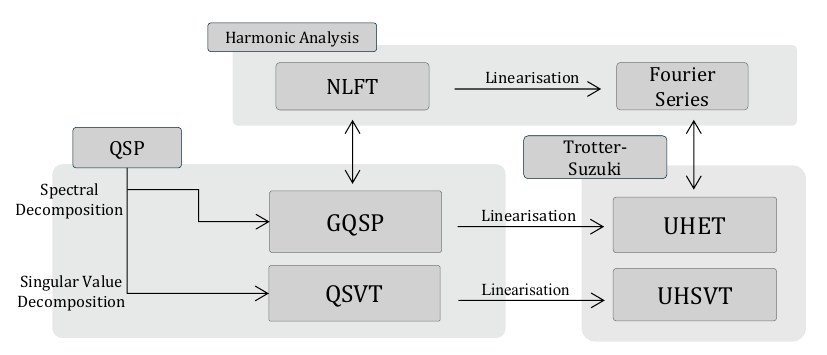}
    
    \caption{Conceptual ``map'' of the different functional programming schemes and their relations to harmonic analysis. We show that the Universal Hamiltonian Eigenvalue Transformation (UHET) \cite{odake2025universal} is a (randomised) linearised version of Generalised Quantum Signal Processing (GQSP) \cite{motlagh2024generalized}. Using this understanding, we lift the intuition to construct a new quantum functional programming algorithm called Universal Hamiltonian Singular Value Transformation (UHSVT), which can be conceptualised as a linearisation of (the Hamiltonian formulation \cite{lloyd2104hamiltonian,kang2025quantum} of) the Quantum Singular Value Transformation (QSVT) \cite{gilyen2019quantum}. This linearisation mapping is analogous to that between the Fourier series and the non-linear Fourier transform (NLFT), which in turn has been shown to be equivalent to GQSP \cite{laneve2025generalized}.}
    \label{fig:QFPmap}
\end{figure}

\vspace{1ex}

\section{Summary of main results}
\label{sec: Summary of main results}

Our results consist of three main parts:
\begin{enumerate}[label={\textcolor{wineRed}{(\Roman*)}}]
    \item We place UHET on the ``map'' (see Figure \ref{fig:QFPmap}) of quantum functional algorithms by establishing its structural connection to GQSP. We construct a variant of GQSP whose architecture closely mirrors that of UHET. This construction allows us to interpret UHET as a randomised linearisation of GQSP, in direct analogy with the relationship between the Fourier transform and its nonlinear version. In this sense, this conceptual result can be summarised as follows: UHET is to GQSP what the Fourier transform is to the nonlinear Fourier transform. 
    \item  With the connection that we have developed, we apply the same linearisation procedure to QSVT, resulting in the formulation of a new algorithm which we call Universal Hamiltonian Singular Value Transformation (UHSVT), adding one more node into the map of quantum functional algorithms (Figure \ref{fig:QFPmap}). We consider the setting where we want to transform the singular values of an arbitrary unknown matrix $A$, and have access to the evolution of a Hamiltonian $H[A]$, which is a \textit{Hamiltonian block-encoding} \cite{kang2025quantum,lloyd2104hamiltonian} of $A$:
    \begin{equation*}
        H[A]=\begin{bmatrix}
            0 & A^\dagger \\
            A & 0
        \end{bmatrix} ,
    \end{equation*}
    but is otherwise given as a black box.

    We then employ a randomised Trotter-Suzuki protocol querying $e^{-i H[A] \tau}$ for variable $\tau$, to construct the evolution $e^{-i H[f^{\rm SV}(A)]t}$ corresponding to a Hamiltonian block-encoding of the target matrix $f^{\rm SV}(A)$, where $f^{\rm SV}$ acts on the singular values of $A$. For any sufficiently differentiable function (specifically, 2-smooth with the sole additional condition that $f$ vanishes at the origin), we show that this algorithm achieves a time complexity of $\mathcal{O}(t^2/\epsilon)$ elementary gates (other than the Hamiltonian dynamics) and total evolution time $\mathcal{O}(t^2 \epsilon^{-1} \log (t \epsilon^{-1}) )$ of the Hamiltonian dynamics, inherited from the underlying randomised Trotter-Suzuki approximation. As a direct application of our new algorithm, we develop an inverse block encoding algorithm to transform a Hamiltonian block-encoding $H[A]$ into a unitary block-encoding $U[A]$, with no dependency on the lower bound of $A$, which is a significant improvement upon the algorithm of \cite{lloyd2104hamiltonian}. 
    
    \item We compare our UHSVT approach to transforming the singular values of a matrix encoded within a Hamiltonian with the QSVT-based methods of \cite{kang2025quantum,lloyd2104hamiltonian}. The QSVT-based approach in \cite{kang2025quantum} achieves a better asymptotic scaling of $\mathcal{O}(1/\epsilon^{1/(J-1)})$ for $J$-smooth functions, but requires the ability to perform $X$-rotations on the \emph{qubitised subspace} encoding $A$. This requirement goes beyond standard QSVT and is not compatible with a black-box setting, where only the location of $A$ within $H$ is known. As such, we also construct a new QSVT-based algorithm similar to that in \cite{kang2025quantum}, but which avoids the need for $X$-rotations on the qubitised subspace. Instead, the new QSVT-based algorithm requires knowledge of a lower bound $\delta$ on the singular values of $A$, inspired by the slightly different setting of \cite{lloyd2104hamiltonian}, with which the time complexity turns out to scale inverse polynomially. We find that although the QSVT-based approaches can outperform UHSVT in terms of scaling with error $\epsilon$ in some cases, UHSVT is applicable in a true black-box setting where neither $X$-gates on the qubitised subspace nor a lower bound on the singular value of $A$ are available, in contrast to the QSVT-based approaches to date.
    Finally, we remark that our UHSVT algorithm is able of implement \emph{complex-valued} functions of the matrix $A$, in contrast to existing QSVT-based approaches. 

    \end{enumerate}

\section{Overview of higher-order quantum algorithms for Hamiltonian dynamics}
\label{sec: Overview of higher-order quantum algorithms for Hamiltonian dynamics}
Quantum algorithms have typically been constructed by applying a sequence of quantum channels (i.e.\ completely positive trace-preserving maps) $\mathcal{C}_1, \mathcal{C}_2, \dots, \mathcal{C}_N$ to an input quantum state $\rho \in \mathcal{H}$. In some scenarios however, the information to be manipulated is encoded not in quantum states, but in quantum channels themselves. This has led to the study of higher-order quantum transformations, which map quantum channels to quantum channels in an analogous  way to which quantum channels map quantum states to quantum states \cite{chiribella2008quantum,chiribella2008transforming}. As such, a higher-order quantum transformation $\mathcal{S}$ is a completely complete-positivity preserving and trace-preservation preserving map from $N$-tuples of quantum channels $(\mathcal{C}_1, \mathcal{C}_2, \dots, \mathcal{C}_N)$ to a new quantum channel $\mathcal{S}(\mathcal{C}_1, \mathcal{C}_2, \dots, \mathcal{C}_N)$. 

Quantum algorithms constructed using higher-order quantum transformations are particularly useful for the task of transforming quantum channels available as a black box by a given  \textit{function} of the channels, for example, transforming $N$ copies of an unknown unitary channel $\mathcal{U}$ to its inverse channel $\mathcal{U}^\dagger$ \cite{quintino2019reversing,yoshida2023reversing,mo2024parameterized,chen2024quantum}, transpose \cite{quintino2019probabilistic} or complex conjuate \cite{miyazaki2019complex,ebler2023optimal}. The work of \cite{odake2024higher,odake2025universal} extended this to the case where the input quantum channels are black-box unitary channels given as the time evolution $e^{-iH\tau}$ of an unknown Hamiltonian $H$, over some variable  time interval $\tau$. In this case, one can construct a higher-order quantum algorithm (i.e.\ a higher-order transformation calling the time evolution $e^{-iH\tau}$, $N$ times in sequence) to perform a functional transformation $f$ of the Hamiltonian $H$, resulting in a new (black-box) unitary simulating the time evolution  $e^{-if(H)t}$ for any desired time $t$, without needing a classical description of $H$. 

Of particular interest in this current work is the Universal Hamiltonian Eigenvalue Transformation (UHET) algorithm \cite{odake2025universal}, which can perform the above simulation for any sufficiently differentiable function $f$ defined on the eigenvalues of a Hamiltonian. Before explaining the technical workings of this higher-order algorithm for transforming \textit{unknown} Hamiltonian dynamics, we will briefly overview a standard family of methods for simulating \textit{known} Hamiltonians, known as  Trotter-Suzuki formulas \cite{Suzuki:1991jtk}, including a randomised version called qDRIFT \cite{campbell2018random}, whose structure forms a cornerstone of UHET.

First, let us  clarify some terminology. Throughout this section and the rest of this paper, we compare algorithms that make use of queries to Hamiltonian dynamics $e^{\pm iH\tau}$. To this end, we employ two distinct metrics. The first is the \textbf{time complexity}, also called the \textit{runtime}, which is the total \textit{depth} of the circuit in terms of elementary gates, not counting calls to the dynamics $e^{\pm iH\tau}$. The second is the \textbf{total evolution time}, defined as the sum $\sum_j |\tau_j|$ of the absolute values of the individual evolution times of all the calls to the Hamiltonian dynamics $e^{\pm iH\tau_j}$.

\subsection{Hamiltonian simulation via Trotter-Suzuki methods}

\label{subsec:Hamiltonian simulation via randomised Trotterisation}

The task of implementing the time evolution generated by a (known) Hamiltonian 
$H$ admitting a known decomposition into easily implementable terms, $H=\sum_{j=1}^{L}\alpha_j H_j$ (with each $||H_j ||<1$), has been extensively studied via Trotter--Suzuki methods. The central idea \cite{lloyd1996universal} is that the evolution $e^{-iHt}=e^{-i\sum_j \alpha_j H_j t}$
can be approximated by sequentially applying short time-evolutions of each Hamiltonian term $H_j$ over time steps of size $t/N$, using the (first-order) Trotter formula
\begin{equation}
e^{-iHt}
=
\left(
\prod_{j=1}^{L}
e^{-i\alpha_j H_j t/N}
\right)^N
+
\mathcal{O}\!\left(\frac{\alpha^2 t^2}{N}\right),
\label{eq: trotter}
\end{equation}
where $\alpha=\|H \|_1=\sum_{j=1}^L\|\alpha_jH_j\|=\sum_{j=1}^L|\alpha_j|$. The error of this bound is taken from Corollary 2 in \cite{PhysRevX.11.011020}, which also provides a sharper bound in terms of commutators; for simplicity, here we stick to the $1$-norm bound. 
To achieve an approximation error $\epsilon$, one must therefore choose $N=\mathcal{O}\!\left(\alpha^2 t^2/\epsilon\right)$ and thus the total time complexity of the algorithm is $L\times N$ or $\mathcal{O}\!\left(L\alpha^2 t^2/\epsilon\right)$. In cases where $\alpha\in \mathcal{O}(L)$, the total time complexity has a cubic dependence on the number of Hamiltonian terms $L$, which can become prohibitive for large decompositions. More elaborate higher-order product formulas (not to be confused with higher-order quantum transformation) can approximate a Hamiltonian more accurately in many cases. In particular, the higher-order Suzuki formulas \cite{Suzuki:1991jtk} are defined recursively as:
\begin{align}
    S_2(t)&\equiv \left(\prod_{j=1}^{L}
e^{-i\alpha_j H_j \frac{t}{2}}\right)\left(\prod_{j=L}^{1}
e^{-i\alpha_j H_j \frac{t}{2}}\right), \label{eq: second-order-trotter}\\
S_{2p}(t)&\equiv S_{2p-2}^2(u_p t )S_{2p-2}((1-4u_p)t)S_{2p-2}^2(u_pt),
\label{eq: Higher-Order Trotter-Suzuki Formula}
\end{align}
where $u_p=1/(4-4^{1/(2p-1)})$.  We can then approximate $e^{-iHt}$ by, once more, cutting the evolution time into $N$ parts and applying the product formula $N$ times resulting in $e^{-iHt}=[S_{2p}(t/N)]^N+\mathcal{O}(\epsilon)$. The scaling in terms of number of elementary gates is given by $\Gamma \times L \times N$, where $\Gamma$ is the number of \textit{stages} in the formula (for instance in the first-order Trotter formula of Eq.\ \eqref{eq: trotter} $\Gamma=1$, while in the second-order Suzuki formula of Eq.\ \eqref{eq: second-order-trotter} $\Gamma=2$ since we apply the product forwards and backwards). It is not hard to show that in the general $2p$-th-order formula, $\Gamma=2\times 5^{p-1}$. Finally, we know from \cite{PhysRevX.11.011020} that $N$ scales as $\mathcal{O}\left[(\alpha t)^{1+1/2p}/\epsilon^{1/2p}\right]$ and therefore the total time complexity is:
\begin{equation}
    \Gamma \times L\times N \in \mathcal{O}\left(5^{p-1}\frac{L(\alpha t)^{1+1/2p}}{\epsilon^{1/2p}}\right) \; .
    \label{eq: Higher-order Trotter-Suzuki runtime}
\end{equation}
This improves the scaling in terms of $\epsilon$ but with the side effect of introducing a constant factor of $\approx 5^{p-1}$. Even so, for cases where $\alpha \in \mathcal{O}(L)$, the total time complexity is reduced to scaling as $\mathcal{O}(L^{2+1/2p})$ in the number of terms $L$ of the Hamiltonian. It is worth noting that the Suzuki formulas of Eq.\ \eqref{eq: Higher-Order Trotter-Suzuki Formula} are not the only possible product formulas, and finding the best one for a given Hamiltonian is a hard task in general.

Randomised Trotter-Suzuki methods, such as qDRIFT
\cite{campbell2018random}, can further decrease the dependancy on $L$ in some cases. Instead of deterministically applying all Hamiltonian terms,
qDRIFT samples them randomly according to the probability distribution $p_j=|\alpha_j|/\alpha$, where $\alpha=\sum_{j=1}^{L} |\alpha_j|$. The protocol (see Figure~\ref{fig:qdrift}) proceeds as follows:
\begin{enumerate}[label=\textcolor{wineRed}{(\roman*)}]
    \item Define the probability distribution 
    $p_j = |\alpha_j|/\alpha$ and sample an index $j$.
    \item Apply the unitary evolution $e^{-i\,\mathrm{sign}(\alpha_j)\,H_j\, t\alpha/N}.$
    \item Repeat steps (1)--(2) a total of $N$ times.
\end{enumerate}
To bound the simulation error in diamond norm by $\epsilon$,
the number of iterations must be chosen as 
\begin{equation}
N(\alpha,t,\epsilon)
=\left\lceil\max\!\left(
10\alpha^2 t^2/\epsilon,\,
\frac{5}{2}\alpha t\right)\right\rceil \, .\label{eq:qDRIFT}     
\end{equation}

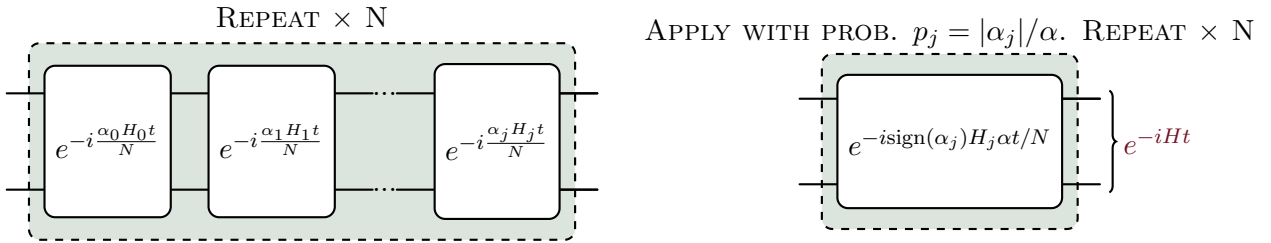
\begin{figure}[H]
    \centering
        \begin{quantikz}
           \qw &    \gate[2,style=rounded corners]{e^{-i \frac{\alpha_0 H_0t}{N}}} \gategroup[2,steps=4,style={dashed,rounded
            corners,fill=vertceladon!20, inner
            xsep=2pt},background,label style={label ,anchor=north,yshift=0.4cm}]{{\sc Repeat $\times$ N}} &    \gate[2,style={rounded corners}]{e^{ -i \frac{\alpha_1 H_1t}{N}}}&...&    \gate[2,style={rounded corners}]{e^{ -i\frac{\alpha_j H_jt}{N}}}& \qw\\
            & & & ... & &\qw
        \end{quantikz}
        \;
        \begin{quantikz}
           \qw &    \gate[2,style={rounded corners}]{e^{ -i \text{sign}(\alpha_j) H_j \alpha t/N}} \gategroup[2,steps=1,style={dashed,rounded
            corners,fill=vertceladon!20, inner
            xsep=2pt},background,label style={label ,anchor=north,yshift=0.4cm}]{{\sc Apply with prob. 
            $p_j=|\alpha_j|/\alpha$.  Repeat $\times$ N}} & \qw \rstick[2]{$\textcolor{wineRed}{e^{-iHt}}$}\\
            & & 
        \end{quantikz}
    \caption{Two circuits for Hamiltonian simulation approximating the unitary channel associated to $e^{-iHt}$. On the left is the standard first-order Trotter-Suzuki approach with a deterministic sequence. On the right is the qDRIFT circuit, where the unitary $e^{-i\text{sign}(\alpha_j)H_jt}$ is randomly applied  according to the distribution $p_j \propto |\alpha_j|$, corresponding to the decomposition $H=\sum_{j=0}^L\alpha_j H_j$.}
    \label{fig:qdrift}
\end{figure}
Therefore, qDRIFT improves the total time complexity  in terms of $L$ from $\mathcal{O}(L^{2+1/2p})$ to $\mathcal{O}(L^2)$ for $\alpha \in \mathcal{O}(L)$, and performs better the smaller the dependence of $\alpha$ is on $L$. For instance, in cases where $\alpha$ is constant, then the difference goes from $\mathcal{O}(L)$ for a deterministic protocol to $\mathcal{O}(1)$ for qDRIFT. Of course, higher-order product formulae can on the other hand offer a faster total time complexity in terms of the error $\epsilon$.  The next lemma,  proven in \cite{campbell2018random}, formalises this result:

\begin{lemma}[qDRIFT \cite{campbell2018random}]
Suppose one has access to the dynamics $e^{-iH_j\tau}$ $(\tau > 0)$ corresponding to a set of
Hamiltonians $\{H_j\}_j$ on $\mathcal{L}(\mathcal{H})$ (normalised as $\|H_j\|_{\rm op} = 1$). 
Then, the qDRIFT protocol with $N$ random  sampling iterations outputs a random unitary channel $\mathcal{F}_{\rm qDRIFT}$, which simulates the 
quantum channel $\mathcal{F}(\rho) := e^{-iHt}\rho\, e^{iHt}$ corresponding to the dynamics  $e^{-iHt}$ $(t > 0)$ of the Hamiltonian $H = \sum_j \alpha_j H_j$ for a set of
positive coefficients $\{\alpha_j\}_j$,  with a diamond-norm error 
\begin{equation}
    \frac{1}{2}\|\mathcal{F} - \mathcal{F}_{\rm qDRIFT}\|_\diamond \leq (2\alpha^2 t^2/N)e^{2\alpha t/N} \; ,
\end{equation}
where $\alpha := \sum_j \alpha_j$.
\label{lem:qDRIFT}
\end{lemma}

\subsection{Universal Hamiltonian Eigenvalue Transformation (UHET)}
\label{sec:uhet}

The Universal Hamiltonian Eigenvalue Transformation (UHET) algorithm \cite{odake2025universal} can perform the  simulation $e^{-iH\tau} \rightarrow e^{-if(H)t}$, given black-box access to $e^{-iH\tau}$ for variable $\tau$, for any sufficiently differentiable (specifically, 3-smooth -- see Definition \ref{def: J-smooth} in Appendix \ref{App: Theory of Fourier approximation}) function $f: [-1,1] \rightarrow \mathbb{R}$ defined on the eigenvalues of a Hamiltonian. (The norm of the traceless part of any Hamiltonian can be assumed to be bounded by 1, without loss of generality). The main idea is to simulate a Fourier series approximation to $f$ by alternating fixed quantum gates (together comprising the higher-order transformation),  with the Hamiltonian evolution for varying evolution times $\tau$, through a randomised Trotter-Suzuki procedure. 

First, assuming we have black-box access to the controlled version of $e^{-iH\tau}$, i.e.\ $\text{ctrl}_0(e^{-i H \tau}) := \ket{0}\bra{0} \otimes e^{-iH\tau} + \ket{1}\bra{1} \otimes I$, we perform a \textit{Fourier series simulation} subroutine to obtain an approximation to $e^{-if(H)t}$ via the Fourier series of $f$ on each eigenvalue of $H$.
The Fourier series of $f$ is given by
\begin{equation}
    f(x)=\sum_{k=-\infty}^{\infty}
    c_k e^{i\pi kx} \,,
    \label{eq:FourierSeriesExpForm}
\end{equation}
where $c_k \in \mathbb{C}$.  A good approximation to $f$ with error $O(\epsilon/t)$ can be constructed by truncating the exact Fourier series to include only terms for $k\leq K$ for some cutoff $K$. The choice of the fixed quantum gates as well as the evolution times $\tau$  depend on the Fourier coefficients $c_k$ and the cutoff $K$ of $f$. Figure \ref{fig:UHET} (top) provides a schematic circuit of the Fourier series simulation subroutine of UHET.

Formally, to ensure proper convergence of the Fourier coefficients, $f$ must be first transformed into an equivalent \emph{periodically} 3-smooth function $\tilde f$ (Definition \ref{def: J-smooth} in Appendix \ref{App: Theory of Fourier approximation}), which can always be done efficiently by defining
\begin{equation}
\tilde{f}(x) :=
\begin{cases}
    \begin{array}{ll@{\qquad}ll}
        f(2x-1)   & x \in [0,1]    \\
        g(x)   & x \in [-1, 0]
    \end{array} \; ,
\end{cases}
\label{eq:smoothextension}
\end{equation}
where $g$ is a sufficiently smooth polynomial interpolation of $f$ with matching boundaries conditions $g^{(j)}(0)=f^{(j)}(0),~ g^{(j)}(-1)=f^{(j)}(1), ~j\in\{0,1,2\}$ (see \cite{odake2025universal} Eq.\ (7)). In the following, we will assume that $f$ is already periodic for simplicity. Formally, the procedure for non-periodic functions just requires running the same subroutine with the Fourier coefficients of $\tilde f$, and applying an extra $Z$-rotation gate $ e^{- i k\pi Z /4}$ at the start and  $e^{i k\pi Z /4}$ at the end of each iteration (shaded in green in Figure \ref{fig:UHET} (top)).

\begin{figure}[htb]
    \centering
    \begin{subfigure}[b]{0.5\textwidth}
        \scalebox{0.90}{
    \begin{quantikz}[wire types={q,q,n,q}, classical gap=0.07cm, row sep=0.3cm]
         \inputD{+} &\ctrl[open]{1} \gategroup[4,steps=3,style={dashed,rounded
            corners,fill=vertceladon!20, inner
            xsep=2pt},background,label style={label ,anchor=north,yshift=0.4cm}]{{\sc Apply with prob. 
            $p_k=|c_k|/\beta$, repeat $\times N$}} & \gate{ e^{-i(\cos\phi_k X - \sin \phi_k Y)\beta t/N}} & \ctrl[open]{1} & \meterD{\rm Tr\vphantom{0}}\\
        & \gate[3,style={rounded corners}]{e^{-i\pi H k}}& &\gate[3,style={rounded corners}]{e^{i\pi H k}} & \rstick[3]{$\textcolor{wineRed}{e^{-if(H)t}}$}\\
        \lstick{\vdots}& &\lstick{\vdots}&&&&\\
        \qw &  &  & &  
    \end{quantikz}}
    \caption{Randomised UHET.}
    \end{subfigure}
    \begin{subfigure}[b]{0.5\textwidth}
    \scalebox{0.90}{
    \begin{quantikz}[wire types={q,q,n,q}, classical gap=0.07cm, row sep=0.3cm]
         \inputD{+} &\ctrl[open]{1} \gategroup[4,steps=3,style={dashed,rounded
            corners,fill=vertceladon!20, inner
            xsep=2pt},background,label style={label ,anchor=north,yshift=0.4cm}]{{\sc Repeat for $k\in\{-d, -d+1, \dots, d\}$, repeat $ \times N$}} & \gate{ e^{-i(\cos\phi_k X - \sin \phi_k Y)|c_k| t/N}} & \ctrl[open]{1} & \meterD{\rm Tr\vphantom{0}}\\
        & \gate[3,style={rounded corners}]{e^{-i\pi H k}}& &\gate[3,style={rounded corners}]{e^{i\pi H k}} & \rstick[3]{$\textcolor{wineRed}{e^{-if(H)t}}$}\\
        \lstick{\vdots}& &\lstick{\vdots}&&&&\\
        \qw &  &  & &  
    \end{quantikz}}
    \caption{Deterministic UHET.}
    \end{subfigure}
    \caption{Randomised (top) and deterministic  (bottom) UHET algorithms for transforming the eigenvalues of a traceless black-box Hamiltonian $H$, given access to its controlled evolution operator $\text{ctrl}_0 \left(e^{-iH\tau}\right)$ for arbitrary time $\tau$, via a sufficiently differentiable function $f$ (here assumed periodic, for simplicity). Here $c_k = |c_k|e^{i\phi_k}$ and $\beta = \sum_{k=-K}^K |c_k|$. The randomised algorithm proceeds by sampling the Fourier mode $e^{-ikH}$ of $f$ according to a probability distribution determined by the size of the Fourier coefficients of $f$.} 
    \label{fig:UHET}
\end{figure}
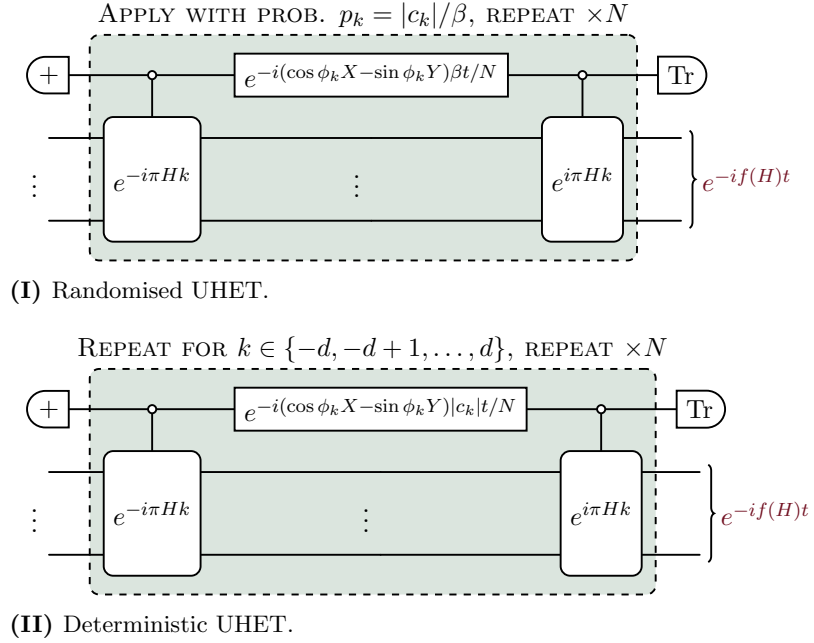

It is straightforward to verify that each iteration of the circuit for a given $k$ (inside the green box) implements the evolution of the following effective Hamiltonian:
\begin{equation}
    H_{\text{eff}}^{(k,\beta t/N,\phi_k)} :=
    \begin{bmatrix}
        0 & e^{i\phi_k}e^{i\pi kH}\\
        e^{-i\phi_k}e^{-i\pi kH} & 0
    \end{bmatrix}\beta t/N
\end{equation}
on the system and a single auxiliary qubit.
This effective Hamiltonian $H_{\text{eff}}^{(k,\beta t/N,\phi_k)}$ encodes the $k$-th frequency component of the Fourier series, along with the complex phase $e^{i\phi_k}$ of the Fourier coefficient $c_k$.  The sampling of each term corresponding to an index $k$ with probability $p_k=|c_k|/\beta$ (proportional to the size of the Fourier coefficients) can be seen as a qDRIFT procedure over the different $H_{\text{eff}}^{(k,\beta t/N,\phi_k)}$. The remaining two ingredients, summing over all frequencies and incorporating the modulus $|c_k|$ of each Fourier coefficient, are both handled by this qDRIFT procedure, resulting in a simulation of the transformed Hamiltonian
\begin{equation}
    \begin{bmatrix}
        0 & f(H) \\
        f(H) & 0
    \end{bmatrix}
    = X \otimes f(H) \,.
\end{equation}

From a scaling perspective, if the controlled version of the Hamiltonian dynamics is already available as a black box, for any choice of times $\tau$, then UHET can be directly implemented using the Fourier series simulation subroutine, which has both time complexity and total evolution time scaling as $O(\beta^2 t^2/\epsilon)$ inherited from qDRIFT, for normalisation factor $\beta = \sum_{k=-K}^K |c_k|$, desired time $t$ and diamond-norm error $\epsilon$. In cases where the Fourier coefficients $c_k$ decay rapidly, e.g.\ as $|c_k| \approx 1/\mathrm{poly}(k)$, so that the normalisation factor $\beta$  of the distribution $p_k = |c_k|/\beta$ remains bounded by a constant independent of the truncation cutoff $K$, then the time complexity scaling reduces to $\mathcal{O}(t^2/\epsilon)$ and similarly for the total evolution time.

This leads to the following theorem:
\begin{theorem}[Universal Hamiltonian Eigenvalue Transformation (randomised, controlled-Hamiltonian access) \cite{odake2025universal}]
For any 3-smooth function $f: [-1,1]\rightarrow \mathbb{R} $, there exists a corresponding periodically 3-smooth function $\tilde f: [-1,1]\rightarrow \mathbb{R}$ of the form of Eq.\ \eqref{eq:smoothextension}, such that for any time $t>0$ and error $\epsilon>0$, there exists a  truncation number $K \in \mathbb{N}$ and an iteration number $N\in \mathbb{N}$, with $N \in \mathcal{O}\left(\beta^2 t^2/\epsilon\right)$, such that for any Hamiltonian $H\in \mathfrak{su}(n)$ (traceless with $||H||_{\rm op}\leq 1$), the randomised application of the unitary sequence
\begin{equation}
        U_k^{H, \beta t/N,\tilde \phi_k} :=   ( e^{i k\pi Z /4} \otimes I) ~
      \text{ctrl}_0(e^{i\pi H k})
      ~ [ R_{\rm xy}\left(\beta t/N,\tilde \phi_k \right) \otimes I 
      ] ~
      \text{ctrl}_0(e^{-i\pi H k}) 
       ~ ( e^{- i k\pi Z /4} \otimes I)
\end{equation}
repeated $N$ times with probability $|\tilde c_k|/\beta$ satisfies
\begin{equation}
      \left\vert \left\vert   \sum_{k=1}^N \frac{|\tilde c_k|}{\beta}
     U_k^{H, \beta t/N,\tilde \phi_k} (\cdot)  {U_k^{H, \beta t/N,\tilde \phi_k} }^\dagger
  - e^{-i X \otimes  f(H) t} (\cdot ) e^{i X \otimes  f(H) t} \right\vert \right\vert_\diamond
\leq \epsilon
   \,,
\end{equation}
where $\tilde f_K(x):=\sum_{k=-K}^K \tilde c_k e^{-i\pi xk}$ is the truncated Fourier series of $\tilde f$, $\tilde c_k = |\tilde c_k|e^{i \tilde \phi}$ and $ \beta:=\sum_k |\tilde c_k|$. Moreover, both the time complexity and the average total evolution time of the algorithm are $\mathcal{O}(N)$.
\end{theorem}

One can also consider a deterministic variant of UHET, in which the moduli of the Fourier coefficients are not sampled according to $p_k$, but instead are directly encoded as parameters of the effective Hamiltonian $H^{(k,|c_k| t/N,\phi_k)}_{\rm eff}$.
Conceptually, the two approaches are related as follows:
\begin{equation}
    \overbrace{
    \exp\left(-i\begin{bmatrix}
        0 & e^{i\phi_k}e^{i\pi H k}\\
        e^{-i\phi_k}e^{-i\pi Hk} & 0
    \end{bmatrix}\frac{|c_k| t}{N}\right)}^{\text{Applied deterministically}} 
    \quad
    \overset{\text{qDRIFT}}{\longrightarrow}
    \quad
    \overbrace{
    \exp\left(-i\begin{bmatrix}
        0 & e^{i\phi_k}e^{i \pi Hk}\\
        e^{-i\phi_k}e^{-i\pi Hk}& 0
    \end{bmatrix}\frac{\beta t}{N}\right)}^{\text{Applied with probability } |c_k|/\beta} \,.
    \label{eq: relation between deterministic and randomized UHET}
\end{equation}
The deterministic version of UHET is shown schematically in Figure \ref{fig:UHET} (bottom).

Regarding the scaling of the deterministic protocol, 
the number of Trotter steps required is $\beta^2 t^2/\epsilon$ from Eq.~\eqref{eq: trotter}. Since at each Trotter iteration we implement $2K$ frequencies, and each requires evolution time $\tau_k=\mathcal{O}(k)$ to run the controlled evolution $e^{\pm i\pi Hk}$, the total evolution time is bounded above by 
\begin{equation}
    \mathcal{O}\left(\sum_{k=-K}^Kk\right)\times \mathcal{O}(\beta^2 t^2/\epsilon)=\mathcal{O}(K^2\beta^2 t^2/\epsilon)
\end{equation}
while the time complexity is $\mathcal{O}\left(K\beta^2t^2/\epsilon \right)$ elementary gates. For $3$-smooth functions, we have that $\beta \in \mathcal{O}(1)$, and $K \in \mathcal{O}(1/\epsilon_K^{1/2})$ (see Lemma \ref{th:decaypropertiesfourierseries} in Appendix \ref{App: Theory of Fourier approximation}). Since the time $t$ linearly amplifies the error, for the overall protocol to reach accuracy $\epsilon$, we need to choose $\epsilon_K=\epsilon/t$, in which case the total evolution time becomes $\mathcal{O}(t^{3}/\epsilon^{2})$ and the time complexity becomes  $\mathcal{O}(t^{5/2}/\epsilon^{3/2})$, worsening the scaling in $\epsilon$ compared to randomised UHET. However, one can also employ higher-order Suzuki formulas as in Eq.~\eqref{eq: Higher-Order Trotter-Suzuki Formula}, leading to improved scaling in both $\epsilon$ and $t$; we defer these calculations to Section~\ref{sec: Universal Hamiltonian Singular Value Transformation} for our new UHSVT algorithm, which will be better suited for deterministic product formulae.

The deterministic version of UHET can be characterised as follows, incorporating both the general case and the simple case when $f$ is already periodic, which we will use later in Section \ref{sec: Universal Hamiltonian Eigenvalue Transformation as a linearisation of Generalised QSP}.

\begin{theorem}[Universal Hamiltonian Eigenvalue Transformation (deterministic, controlled-Hamiltonian access), adapted from \cite{odake2025universal}]

For any periodically 3-smooth function $f: [-1,1]\rightarrow \mathbb{R} $, time $t>0$ and error $\epsilon>0$, there exists a  truncation number $K \in \mathbb{N}$ and an iteration number $N\in \mathbb{N}$, with $N \in \mathcal{O}\left(\beta^2 t^2/\epsilon\right)$, such that for any Hamiltonian $H\in \mathfrak{su}(n)$ (traceless and normalised such that $||H||_{\rm op}\leq 1$), 
\begin{equation}
    \mathtt{UHET}^{\rm det}_{f_K,t,N}(H) := \left[ \prod_{k=-K}^K 
      \text{ctrl}_0(e^{i\pi H k})
      ~ [ R_{\rm xy}\left(|c_k|t/N,\phi_k \right) \otimes I 
      ] ~
      \text{ctrl}_0(e^{-i\pi H k}) 
    \right]^N \overset{\epsilon}{\approx} e^{-i X \otimes  f(H) t}
   \,.
\end{equation}
where $ f_K(x):=\sum_{k=-K}^K  c_k e^{-i\pi xk}$ is the truncated Fourier series of $ f$, $ c_k = | c_k|e^{i  \phi}$, $ \beta:=\sum_k | c_k|$ and the rotation $R_{\rm xy}(\theta,\phi) := e^{-i[\cos(\phi) X -\sin(\phi)Y]\theta}$. Here, $\overset{\epsilon}{\approx}$ denotes equality up to global phase in the operator norm, up to error $\epsilon$.

Moreover, for any (not necessarily periodic) 3-smooth function $f: [-1,1]\rightarrow \mathbb{R} $, there exists a corresponding periodically 3-smooth function $\tilde f: [-1,1]\rightarrow \mathbb{R}$ of the form of Eq.\ \eqref{eq:smoothextension}, such that
\begin{align}
\nonumber 
    \widetilde {\mathtt{UHET}}^{\rm det}_{\tilde f_K,t,N}(H) 
    &:= \left[ \prod_{k=-K}^K 
    ( e^{i k\pi Z /4} \otimes I) ~
      \text{ctrl}_0(e^{i\pi H k})
      ~ [ R_{\rm xy}\left(|\tilde c_k|t/N,\tilde \phi_k \right) \otimes I 
      ] ~
      \text{ctrl}_0(e^{-i\pi H k}) 
       ~ ( e^{- i k\pi Z /4} \otimes I)
    \right]^N \\
    &\overset{\epsilon}{\approx} e^{-i X \otimes  f(H) t} \,,
\end{align}
where $\tilde f_K(x):=\sum_{k=-K}^K \tilde c_k e^{-i\pi xk}$ is the truncated Fourier series of $\tilde f$, $\tilde c_k = |\tilde c_k|e^{i \tilde \phi}$ and $ \beta:=\sum_k |\tilde c_k|$.
\label{thm: UHET}
\end{theorem}

It may seem surprising that qDRIFT is employed in \cite{odake2025universal} rather than a deterministic higher-order Trotter-Suzuki method. The reason is that the above comparison assumes the controlled Hamiltonian evolution is directly available. If only control-free evolution $e^{-iH\tau}$ is given, then a \textit{controllisation} subroutine \cite{dong2019controlled}, itself based on a randomised Trotter-Suzuki procedure, must first be run in time $\mathcal{O}(t^2 n/\epsilon)$ to synthesize the controlled dynamics, where $n$ is the number of qubits. This introduces an additional layer of randomness and leads to a total evolution time of $\mathcal{O}(t^4n/\epsilon^3)$ for the concatenated (controllisation + Fourier series simulation) version of UHET. In this setting, the compilation procedure for concatenated randomised subroutines introduced in \cite{odake2025universal} can be used to merge the two levels of randomness, arising from the Fourier series simulation and the controllisation steps, reducing the overall scaling back to $\mathcal{O}(1/\epsilon)$. However, this compilation procedure does not appear to extend to deterministic UHET, and such routines therefore exhibit worse scaling than the compiled randomised algorithm.

\begin{remark}[On the realness of $f$]
   Recall that the original UHET algorithm was designed to implement real-valued functions $f: \mathbb{R} \to \mathbb{R}$, which is enforced by the conjugate symmetry condition $c_k^* = c_{-k}$, or equivalently $\phi_k = -\phi_{-k}$ and $|c_k| = |c_{-k}|$. However, the realness of $f$ is not a fundamental requirement, since the Hamiltonian
\begin{equation}
    \begin{bmatrix}
        0 & f(H) \\
        f^*(H) & 0
    \end{bmatrix}
\end{equation}
remains a well-defined Hermitian matrix for any complex-valued function $f$ and, as we shall see in the next section, constitutes a \textit{Hamiltonian block-encoding} of $f(H)$. In \cite{odake2025universal}, the realness of $f$ is imposed for two reasons: 1) it preserves the Hermitian structure of $H$ because the intended goal was to transform Hamiltonians to Hamiltonian, and 2) it allows one to extract the target evolution $e^{-if(H)t}\ket{\psi}$ directly by initialising the ancillary register in $\ket{+}$. This extraction technique does not carry over to complex-valued $f$.
\end{remark}

\begin{remark}[Fourier series vs.\ polynomials] The function $f$ resulting from UHET is generally expressed as a Fourier series in $H$, since the UHET algorithm was originally designed to implement functions on Hamiltonians. However, the series $\sum_k e^{i\pi Hk}$ can equivalently be written as a Laurent polynomial in $U = e^{i\pi H}$, and therefore the UHET algorithms could equally well be interpreted as constructing a Laurent polynomial of $U$.
\end{remark}

\section{Overview of Quantum Signal Processing-based algorithms}
\label{sec: Overview of QSP-based functional transformations}
In this section, we provide an overview of Quantum Signal Processing (QSP),
first introduced in~\cite{low2017optimal}. In the first subsection, we present
QSP in its most general form, namely as the repeated application of a signal operator interleaved with a sequence of
processing operators. In the second subsection,
we explain how QSP can be used to transform the eigenvalues of a (controlled) black-box unitary, leading to the construction of GQSP. Finally, in the last subsection, we present (the Hamiltonian-based version of) QSVT, which builds upon QSP to transform the singular values of arbitrary matrices \textit{block-encoded} within a unitary.

It is important to clarify the terminology, which can otherwise be confusing.
The term QSP was originally introduced in a specific setting where the
processing operators are $Z$-rotations on a qubit, and the signal operator is an
$X$-rotation on a qubit \cite{low2017optimal}\cite{gilyen2019quantum}\cite{martyn2021grand}. Then \cite{dong2022ground} introduced a QSP-based algorithm to transform the eigenvalues of an arbitrary unitary $U$. Later \cite{motlagh2024generalized} generalised this method, by introducing GQSP, a form of signal processing where the processing operators are arbitrary $\mathrm{SU}(2)$ rotations and the signal operator is a controlled unitary. In this work, we adopt the terminology of \cite{laneve2025generalized}, that is: we use the term QSP to denote the
superoperator obtained by interleaving any signal and processing operators,
without imposing any particular choice of these operators.
Our terminology is summarised below:

\begin{enumerate}
    \item \textbf{QSP (Quantum Signal Processing)}: Alternating signal and processing operators,
    defined abstractly as a superoperator framework. Analytic, Laurent, and Chebyshev QSPs are three different instances corresponding to different parametrisations of the signal operator. Processing operators are typically chosen to be some subclass of SU(2) rotations (e.g.\ $X$ or $Z$ rotations). 

    \item \textbf{GQSP (Generalised QSP)}: The most general instance of QSP where the processing operators are allowed to be arbitrary $\mathrm{SU}(2)$-rotation operators, while the signal operator is taken to be a controlled-unitary operator.
    GQSP is typically applied to the task of transforming the eigenvalues of a unitary operator.
    \item  \textbf{QSVT (Quantum Singular Value Transformation)}:
    The application of the QSP framework to transform the singular values of a block-encoded
    operator $A$. If the matrix $A$ is block-encoded inside a Hamiltonian where we have access to its associated unitary evolution, then we refer to the application of the QSP framework to transform $A$ as \textbf{H-QSVT (Hamiltonian QSVT)}.
\end{enumerate}

\subsection{Quantum Signal Processing (QSP)}

At a high level, a QSP protocol can be seen as two key ingredients: a signal operators $s(z)$, which depends on a single variable $z$ (the signal), and an ordered list of processing operators $\{A_k\}_{k=0}^{d}$ where $\forall k, \, A_k \in \mathrm{SU}(2)$. A QSP protocol is then the following sequence:
\begin{equation}
A_d s(z) A_{d-1} s(z) A_{d-2} \cdots s(z) A_0 \; .
\end{equation}
In the literature, we can find a whole host of different types of signal operators and processing operators. 

Two important signal operators are the 
Analytic signal and the Laurent signal, respectively:
\begin{equation}
\tilde{w}(z) := \begin{bmatrix} z & 0 \\ 0 & 1 \end{bmatrix}, \qquad \tilde{v}(z) := \begin{bmatrix} z & 0 \\ 0 & z^{-1} 
\end{bmatrix} \,,
\end{equation}
where $z\in \mathbb{T}:=\{z=e^{-i\theta}| \theta\in [0,2\pi)\}$ (the unit circle). 

The most general processing operators are a set of general $\mathrm{SU}(2)$ operators of the form: 
\begin{equation}
A_k = R(\theta_k, \phi_k, \lambda_k) := \begin{bmatrix} e^{i(\phi_k + \lambda_k)} \cos(\theta_k) & e^{i\phi_k} \sin(\theta_k) \\ e^{i\lambda_k} \sin(\theta_k) & -\cos(\theta_k) \end{bmatrix} \;.
\label{eq: GQSP processing operators}
\end{equation}

Together with the Analytic signal operator, we obtain the following theorem:
\begin{theorem}[Analytic QSP (with arbitrary SU(2) processing) \cite{motlagh2024generalized}]
\label{thm:analytic-qsp}
For any polynomials $P,Q \in \mathbb{C}[z]$ such that:
\begin{enumerate}[label=\textcolor{wineRed}{(\roman*)}]
    \item \textbf{(Maximum degree condition)} $\deg(P) \leq d$ and $\deg(Q) \leq d \;,$ \label{cond:max-degree}
    \item \textbf{(Normalisation condition)} $\forall z \in \mathbb{T}, \; |P(z)|^2 + |Q(z)|^2 = 1 \;,$ \label{cond:normalization}
\end{enumerate}
there exist $\Theta=\{\theta_0,\theta_1,\theta_d\dots,\}$, $\Phi=\{\phi_0,\phi_1,\dots,\phi_d\}$ and $\lambda\in \mathbb{R}$ such that:
\begin{equation}
  \left(\prod_{k=d}^{1} R(\theta_k, \phi_k, 0) \tilde{w}(z)\right)  R(\theta_0, \phi_0, 0)R_{\rm z}(\lambda) = \begin{bmatrix} P(z) & -Q^*(z) \\ Q(z) & P^*(z) \end{bmatrix} \; .
\end{equation}
\end{theorem}
Here, $R_{\rm z}(\lambda) := e^{-i \lambda Z}$.
This theorem shows that the interleaving of signal and processing operators yields an $\mathrm{SU}(2)$ matrix whose entries are polynomials in the signal parameter $z$. The proof proceeds by induction on the degree of the polynomial.

More precisely, assume that any polynomial of degree $d-1$ can be implemented by a QSP sequence of length $d-1$. Appending an additional signal operator $\tilde{w}(z)$ increases the degree by one. By subsequently applying a processing operator $R(\theta_d,\phi_d,\lambda_d)$, one can appropriately tune the angles so as to obtain the desired degree-$d$ polynomial entries.

Of course, this theorem is based on a specific choice of processing operators and signal operators, namely general SU(2) processing operators and the Analytic signal (the choices later used to define GQSP on the eigenvalues of unitaries). Varying these sets will change the class of achievable polynomials. For instance, using a Laurent signal instead of the Analytic signal will result in Laurent polynomials $P',Q' \in \mathbb{C}[z,z^{-1}]$, i.e.\ polynomials with negative powers, with parity $d~ ({\rm mod} 
~2)$. Alternatively, restricting the set of processing operators to $X$ rotations ($\{e^{-i\theta_k X}\}$) in either Analytic or Laurent QSP will result in polynomials with only real coefficients.

The original formulation of QSP considered a third parametrisation of the signal, known as the Chebyshev signal (which, paired with $Z$-rotation processing operators, naturally lifts to QSVT for transforming the singular values of matrices).
Above, $\tilde{w}(z)$ and $\tilde{v}(z)$ were both parametrised with a complex number on the unit circle, $z \in \mathbb{T}$. Instead, the Chebyshev signal is given in terms of the variable $x=\cos(\theta)=(e^{i\theta}+e^{-i\theta})/2=(z+z^{-1})/2$ for $x \in [-1,1]$. To construct a signal depending on this parameterisation, one can simply conjugate the Laurent signal operator with Hadamard gates, yielding:
\begin{equation}
    H\tilde{v}(z)H=
    \frac{1}{2}\begin{bmatrix}
        z+z^{-1} & z-z^{-1}\\
        z-z^{-1} & z+z^{-1}
    \end{bmatrix}=\begin{bmatrix}
        \cos(\theta) & -i\sin(\theta) \\
        -i\sin(\theta) & \cos(\theta)
    \end{bmatrix}=
    \begin{bmatrix}
        x &-i\sqrt{1-x^2}\\
        -i\sqrt{1-x^2} & x
    \end{bmatrix} =: W(x) \; .
\end{equation}
This signal, denoted $W(x)$, is called the Chebyshev signal and, together with processing operators constrained to $Z$-rotations, leads to the following theorem:

\begin{theorem}[Chebyshev QSP (with $Z$-constrained processing) \cite{low2017optimal}\cite{martyn2021grand}\cite{gilyen2019quantum}]
For any polynomials $P \in \mathbb{C}[x]$ and $Q \in \mathbb{C}[x]$ such that:
  \begin{enumerate}[label=\textcolor{wineRed}{(\roman*)}]
    \item \textbf{(Maximum degree condition)} $\deg(P) \leq d$ and $\deg(Q) \leq d-1$,
    \label{cond:max-degree chebyshev}
    \item \textbf{(Normalisation condition)} $\forall x \in [-1,1], \; |P(x)|^2 + |Q(x)|^2(1-x^2) = 1$,
    \label{cond:normalization chebyshev}
    \item \textbf{(Parity condition)} $P$ has parity $d$ mod $2$ and $Q$ has parity $(d-1)$ mod $2$, \label{cond:parity}
  \end{enumerate}
there exists a list of angles $\Theta=\{\theta_0,\theta_1,...\theta_d\}$ such that:
\begin{equation}
 e^{-iZ\theta_d} \prod_{k=d-1}^{0} \underbrace{\begin{bmatrix}
        x &\pm i\sqrt{1-x^2}\\
        \pm i\sqrt{1-x^2} & x
    \end{bmatrix}}_{W(x)}
    \underbrace{\begin{bmatrix}
        e^{-i\theta_k} & 0\\
        0 & e^{i\theta_k}
    \end{bmatrix}}_{e^{-iZ\theta_k}=R_{\rm z}(\theta_k)}
     = \begin{bmatrix} P(x) & \pm iQ(x)\sqrt{1-x^2} \\ \pm iQ^*(x)\sqrt{1-x^2} & P^*(x) \end{bmatrix} \; .
     \label{eq: ChebyshevQSPproductmatrix}
\end{equation}
\label{thm: ChebQSP}
\end{theorem}
Observe that the class of achievable polynomial transformations is more restricted than in Analytic QSP (Theorem \ref{thm:analytic-qsp}), since the resulting polynomials are required to have definite parity -- a restriction ultimately arising from the constraint on the processing operators.
The relations between different types of QSP in the different pictures (Analytic, Laurent and Chebyshev) are further explored in \cite{laneve2025generalized,martyn2021grand}.

\subsection{Generalised Quantum Signal Processing (GQSP)}
In this section, we will see how to use quantum signal processing to apply a polynomial transformation on the eigenvalues of any unitary. 

 Suppose that we have access to a black-box unitary $U$ and its controlled version, that is, $\mathrm{ctrl}_0(U):=\ket{0}\bra{0}\otimes U+ \ket{1}\bra{1}\otimes I$. Since $U$ is unitary, we can write $U=e^{-iH}$ where $H$ is a Hermitian matrix, which means that by the spectral theorem, we can decompose it as $H = \sum_n E_n |E_n\rangle\langle E_n|$ where $\{E_n\}$ is the set of eigenvalues of $H$ and $\{|E_n\rangle\}$ the set of eigenvectors of $H$. Thus $e^{-iHt} = \sum_n e^{-iE_n t} |E_n\rangle\langle E_n|$ and:
\begin{equation}
\mathrm{ctrl}_0(U) = \bigoplus_n \begin{bmatrix} e^{-iE_n t} & 0 \\ 0 & 1 \end{bmatrix} \otimes |E_n\rangle\langle E_n| = \bigoplus_{E_n} \begin{bmatrix} z_n & 0 \\ 0 & 1 \end{bmatrix} \otimes |E_n\rangle\langle E_n|
\end{equation}
where $z_n = e^{-iE_n t}$. Hence, controlled Hamiltonian dynamics can be seen as a direct sum of Analytic signal operators parametrised by the eigenvalues of $U$.

Now, if we intertwine processing operators on the control auxiliary qubit while using an Analytic signal operator, we get the following circuit:
\begin{equation}
    \begin{quantikz}[wire types={q,q,n,q}, classical gap=0.07cm, row sep=0.3cm]
        \qw & \gate{A_0} & \ctrl[open]{1}  & \gate{A_1} & \ctrl[open]{1}  & \gate{A_2} & ...&  \ctrl[open]{1}  & \gate{A_d}&\qw \\
    \qw & & \gate[3,style={rounded corners}]{U} & & \gate[3,style={rounded corners}]{U} & & ... & \gate[3, style={rounded corners}]{U} & &\qw \\
    &\lstick{\vdots}&&\lstick{\vdots}&&\lstick{\vdots}&&&\lstick{\vdots}&&\\
        & & & & & & ... & & &
    \end{quantikz}
    \label{circ:AnalyticQSP-U}
\end{equation}
This is Analytic QSP using controlled Hamiltonian dynamics.
\noindent
In the eigensubspaces of $H$, processing operators on the auxiliary wire can be written as: $A_k \otimes I = \bigoplus_n A_k $. The circuit shown in Eq.\ \eqref{circ:AnalyticQSP-U} can therefore be decomposed into the following direct sum: 
\begin{equation}
\bigoplus_n A_d \tilde{w}(z_n) A_{d-1} \tilde{w}(z_n) \cdots \tilde{w}(z_n) A_0 
\end{equation}

The above equation tells us that applying processing operators between controlled unitaries is the same thing as doing QSP on each (qubit) eigensubspace of $H$. We call this \textit{qubitisation}. This observation results in the following theorem:

\begin{theorem}[Generalised Quantum Signal Processing]
\label{thm:GQSP}
For any polynomials $P,Q \in \mathbb{C}[z]$ satisfying the maximum degree condition \ref{cond:max-degree} and the normalisation condition \ref{cond:normalization}, there exist $\Theta=\{\theta_0,\theta_1,\dots,\theta_d\}$, $\Phi=\{\phi_0,\phi_1,\dots,\phi_d\}$ and $\lambda\in \mathbb{R}$ such that:
\begin{equation}
  \mathtt{GQSP}_{\Theta,\Phi,\lambda}(U) :=\left[\prod_{k=d}^{1} R(\theta_k, \phi_k, 0) \mathrm{ctrl}_0(U)\right]  R(\theta_0, \phi_0, 0)R_{\rm z}(\lambda) = \begin{bmatrix} P(U) & -Q^*(U) \\ Q(U) & P^*(U) \end{bmatrix} \; .
\end{equation}
\end{theorem}
\noindent

Note that here and in the following, we have suppressed writing tensor products with the identity $\otimes I$ when it is clear which systems are being acted on by each of the rotation operators $R(\cdot)$.

\subsection{Hamiltonian Quantum Singular Value Transformation (H-QSVT)}
\label{subsec: HQSVT}
The next functional transformation we describe is the \textit{Hamiltonian Quantum Singular Value Transformation} (H-QSVT). Its goal is to transform the singular values of an operator $A$ encoded into a Hamiltonian $H[A]$, provided access to the associated unitary $e^{-iH[A]\tau}$.

\begin{definition}[Hamiltonian Block-Encoding]
    Let $A \in \mathcal{L}(\mathcal{H}_0 \to \mathcal{H}_1)$ be an arbitrary operator, where $\mathcal{H}_0$ and $\mathcal{H}_1$ are Hilbert spaces whose dimensions need not match, so that $A$ may be non-square. Let $\Pi_0$ and $\Pi_1$ denote the orthogonal projectors onto $\mathcal{H}_0$ and $\mathcal{H}_1$, respectively. An $(\alpha, \epsilon)$-\emph{Hamiltonian block-encoding} of $A$ is a Hermitian operator
    \begin{equation}
        H[A] \in \mathcal{L}(\mathcal{H}_0 \oplus \mathcal{H}_1)
    \end{equation}
    satisfying
    \begin{equation}
        \bigl\| A - \alpha\, \Pi_1\, H[A]\, \Pi_0 \bigr\| \leq \epsilon \,.
    \end{equation}
    The associated unitary evolution operator is $e^{-iH[A]}$.
    \label{def:HamiltonianBlockEncoding}
\end{definition}

\noindent This model of block-encoding was introduced in \cite{lloyd2104hamiltonian} and further developed in \cite{kang2025quantum}. A Hamiltonian block-encoding of $A$ takes the general matrix form
\begin{equation}
    H[A] = \begin{bmatrix} D & A^\dagger \\ A & \widetilde{D} \end{bmatrix},
\end{equation}
where $D \in \mathcal{L}(\mathcal{H}_0)$ and $\widetilde{D} \in \mathcal{L}(\mathcal{H}_1)$ are Hermitian operators. Throughout our analysis, it will be convenient to work with a canonical form of this construction.

\begin{definition}[Canonical Form of Hamiltonian Block-Encoding]
    A Hamiltonian block-encoding $H[A]$ of $A$ is said to be in \emph{canonical form} if
    \begin{equation}
        H[A] = \begin{bmatrix} 0 & A^\dagger \\ A & 0 \end{bmatrix}.
    \end{equation}
    \label{def:standardform}
\end{definition}

\noindent The canonical form corresponds to setting $D = \widetilde{D} = 0$ in the general block structure. This can always be obtained from a general Hamiltonian block-encoding via a Trotter step by sequentially applying $e^{-iH[A]}$ and $Ze^{iH[A]}Z$ \cite{lloyd2104hamiltonian,lloyd2020quantum}. We call such a procedure \emph{Hamiltonian refocusing}. This canonical form is particularly convenient because it admits a clean decomposition in terms of the singular value structure of $A$:

\begin{lemma}[Singular Value Decomposition]
    Let $A \in \mathcal{L}(\mathcal{H}_0 \to \mathcal{H}_1)$. Then there exist unitaries $U \in \mathcal{L}(\mathcal{H}_1)$ and $V \in \mathcal{L}(\mathcal{H}_0)$ such that
    \begin{equation}
        A = U \Sigma V^\dagger,
    \end{equation}
    where $\Sigma \in \mathcal{L}(\mathcal{H}_0 \to \mathcal{H}_1)$ is a diagonal rectangular matrix whose entries are the (real, non-negative) singular values of $A$.
    \label{th:SVD}
\end{lemma}

\noindent For the remainder of this section, we take $A$ to be a square matrix, so that $\Sigma$ is square diagonal. The non-square case reduces to this one by padding the Hamiltonian with a zero block. Now, substituting the singular value decomposition of $A$ into the 

time-evolution corresponding to its canonical form Hamiltonian block-encoding gives
\begin{align}
    \exp\!\left(-i\begin{bmatrix} 0 & A^\dagger \\ A & 0 \end{bmatrix}\right)
    &= \exp\!\left(-i\begin{bmatrix} 0 & V\Sigma U^\dagger \\ U\Sigma V^\dagger & 0 \end{bmatrix}\right) \\
    &= \begin{bmatrix} V & 0 \\ 0 & U \end{bmatrix}
       \exp\!\left(-i\begin{bmatrix} 0 & \Sigma \\ \Sigma & 0 \end{bmatrix}\right)
       \begin{bmatrix} V^\dagger & 0 \\ 0 & U^\dagger \end{bmatrix} \\
    &= \begin{bmatrix} V & 0 \\ 0 & U \end{bmatrix}
       \begin{bmatrix} \cos\Sigma & -i\sin\Sigma \\ -i\sin\Sigma & \cos\Sigma \end{bmatrix}
       \begin{bmatrix} V^\dagger & 0 \\ 0 & U^\dagger \end{bmatrix},
       \label{eq: decomposition Hamiltonian block-encoding}
\end{align}
where $\cos\Sigma = \operatorname{diag}(\cos\zeta_0, \cos\zeta_1, \ldots)$ and similarly for $\sin\Sigma$. More generally, we write $f^{\mathrm{SV}}(A) \coloneqq U f(\Sigma) V^\dagger$ for a function $f$ applied to the singular values of $A$.

The central matrix block can be further decomposed into a direct sum over two-dimensional singular-value sectors:
\begin{equation}
    \begin{bmatrix} \cos\Sigma & -i\sin\Sigma \\ -i\sin\Sigma & \cos\Sigma \end{bmatrix}
    = \bigoplus_m \begin{bmatrix} \cos\zeta_m & -i\sin\zeta_m \\ -i\sin\zeta_m & \cos\zeta_m \end{bmatrix}.
    \label{eq: cossigma}
\end{equation}
This is another instance of \emph{qubitisation}: each two-dimensional sector $\mathcal{H}_m$ is associated with a singular value $\zeta_m$, while the whole Hilbert space decomposes as a direct sum of these ``qubit'' blocks:
 $\mathcal{H}_0 \oplus \mathcal{H}_1 = \bigoplus_m \mathcal{H}_m$. 
Moreover, since we assumed that $A$ is square, $\Pi_0 \cong \Pi_1$. Thus, the whole Hilbert space can also be partitioned \cite{vanrietvelde2021routed} as $\mathcal{H}_0 \oplus \mathcal{H}_1  = \mathcal{H}_Q \otimes \mathcal{H}_A$, where $\mathcal{H}_A$ is a Hilbert space of dimension equal to the (input or output) dimension of $A$ and $\mathcal{H}_Q$ is a qubit Hilbert space, defined by
\begin{equation}
	\begin{split}
	 \ket{0}\bra{0}_{\mathcal{H}_Q} \otimes I_{\mathcal{H}_A} &= \Pi_0 \,, \\ \ket{1}\bra{1}_{\mathcal{H}_Q} \otimes I_{\mathcal{H}_A} &= \Pi_1 \,. 
	\end{split}
\end{equation}
We will call $\mathcal{H}_Q$ the \textit{qubitised subspace}, since it encodes the effective qubit contained in the block-encoding structure.

Setting $x_m \coloneqq \cos\zeta_m$, each block takes the form
\begin{equation}
    \bigoplus_m \begin{bmatrix} x_m & -i\sqrt{1-x_m^2} \\ -i\sqrt{1-x_m^2} & x_m \end{bmatrix},
\end{equation}
which is precisely the Chebyshev signal structure appearing in Theorem~\ref{thm: ChebQSP}. Consequently, if one can implement $Z$-rotation processing operators on the singular-value blocks, or equivalently on the qubitised subspace, i.e.\
\begin{equation}
    e^{-i\phi(2\Pi_{0}-I)}\equiv\bigoplus_m \begin{bmatrix} e^{-i\phi} & 0 \\ 0 & e^{i\phi} \end{bmatrix} \equiv e^{i \phi Z }_{~~~\mathcal{H}_Q} \otimes I_{\mathcal{H}_A} \,,
    \label{eq: zrotationsingularsubspace}
\end{equation}
then the full machinery of Theorem~\ref{thm: ChebQSP} applies, now enabling a polynomial transformation of each singular value of $A$. 

To understand how such a $Z$-rotation on the qubitised subspace can be realised in our physical setting, recall that we only assumed query access to the dynamics $e^{-iH[A]\tau}$, for any time $\tau$, corresponding to a canonical Hamiltonian block-encoding of $A$. As such, we do not have access to any physical qubit $\mathcal{H}_Q$, on which we could apply a $Z$-rotation gate. 
Nevertheless, such a $Z$-rotation can be realised by   assuming only access to a controlled-NOT gate acting on an additional auxiliary qubit $\mathcal{H}_{\rm aux}$ and conditioned on the subspace $\mathcal{H}_0$ (or equivalently the projector $\Pi_0$):
\begin{equation}
    C_{\Pi_0}\text{-}X \; := \; X_{\mathcal{H}_{\rm aux}} \otimes \Pi_{0} + I_{\mathcal{H}_{\rm aux}} \otimes \Pi_{1} \;=\; X_{\mathcal{H}_{\rm aux}} \otimes \Pi_{0} + I_{\mathcal{H}_{\rm aux}} \otimes (I - \Pi_{0}) \,.
\end{equation}
Physically, the ability to perform $C_{\Pi_0}\text{-}X$ follows from simply knowing what $\Pi_0$ is, which we require in order to make sense of the original Hamiltonian block-encoding at all.
The required $Z$-rotations on the qubitised subspace can then be implemented via phase kickback using the above $C_{\Pi_{0}}$-NOT gate acting on an additional auxiliary qubit $\ket{0}_{\rm aux}$, as shown in Figure~\ref{fig:bob}.

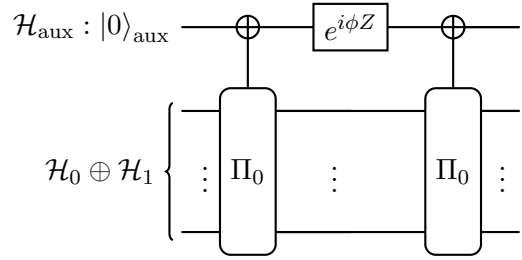
\begin{figure}[htb]
        \centering
        \begin{quantikz}[wire types={q,q,n,q}, classical gap=0.07cm]
            \lstick{$\mathcal{H}_{\rm aux} : \ket{0}_{\rm aux}$} & \targ{} & \gate{e^{i\phi Z}} & \targ{} & \qw \\
            \lstick[3]{$\mathcal{H}_0 \oplus \mathcal{H}_1$} & \gate[3,style={rounded corners}]{\Pi_{0}}\vqw{-1} & \qw & \gate[3,style={rounded corners}]{\Pi_{0}}\vqw{-1} & \qw \\
            & \lstick{\vdots} & \lstick{\vdots} & & \lstick{\vdots} \\
            & & \qw & &
        \end{quantikz}
        \caption{Implementation of the phase-rotation operator on the qubitised subspace is achieved using a $C_{\Pi_{0}}$-NOT operation that flips an additional auxiliary qubit $\mathcal{H}_{\rm aux}$, initialised in the state $\ket{0}_{\rm aux}$.}
\label{fig:bob}
\end{figure}

We can then alternate between the signal $e^{-iH[A]}$ and the $Z$-rotations on the qubitised subspace  to achieve a polynomial in terms of the cosine of $\Sigma$, that is:
\begin{align}
e^{-i\theta_d(2\Pi_0-I)} \prod_{k=d-1}^{0} e^{-i H[A]} e^{-i\theta_k(2\Pi_0-I)} 
&=
\nonumber
    \bigoplus_m
    \begin{bmatrix}
        e^{-i\theta_d} & 0 \\
        0 & e^{i\theta_d}
    \end{bmatrix}
    \prod_{k=d-1}^0
\begin{bmatrix}
    \cos(\zeta_m) & -i\sin(\zeta_m) \\
    -i\sin(\zeta_m) & \cos(\zeta_m)
\end{bmatrix}
\begin{bmatrix}
    e^{-i\theta_k} & 0 \\
    0 & e^{i\theta_k}
\end{bmatrix}
 \\ &=
 \nonumber
\bigoplus_m \begin{bmatrix}
    P(\cos(\zeta_m)) &  -iQ(\cos(\zeta_m))\sin(\zeta_m) \\
    -iQ^*(\cos(\zeta_m))\sin(\zeta_m) & P^*(\cos(\zeta_m))
\end{bmatrix}\\
&=
\begin{bmatrix}
    P(\cos(\Sigma)) &   -iQ(\cos(\Sigma))\sin(\Sigma) \\
    -iQ^*(\cos(\Sigma))\sin(\Sigma) &  P^*(\cos(\Sigma))
\end{bmatrix} \,,
\label{eq:polycossigma}
\end{align}
where the first equality is up to a change of basis.
The resulting polynomials have
the same parity and normalisation constraints as in the usual Chebyshev QSP theorem, leading us to the following theorem:
\begin{theorem}[Hamiltonian Quantum Singular Value Transformation (H-QSVT) \cite{lloyd2104hamiltonian}]
\label{thm: HQSVT}
For any polynomials $P,Q \in \mathbb{C}[x]$ satifying the conditions \textcolor{wineRed}{(i)}-\textcolor{wineRed}{(iii)} of Chebyshev QSP (Theorem \ref{thm: ChebQSP}) with $d\in \mathbb{N}$, there exists $\Phi=\{\phi_0,\phi_1,...,\phi_d\}$ such that: 
\begin{equation} e^{-i\phi_{d}(2\Pi_{0} - {I})} \prod_{k=d-1}^{0} e^{-iH[A]}e^{-i\phi_k(2\Pi_{0} - {I})}= 
\begin{bmatrix}
        V & 0\\
        0 & U
    \end{bmatrix}
\begin{bmatrix} P(\cos(\Sigma)) & -iQ(\cos(\Sigma))\sin(\Sigma) \\ -iQ^*(\cos(\Sigma))\sin(\Sigma) & P^*(\cos(\Sigma)) 
\end{bmatrix}
\begin{bmatrix}
        V^\dagger & 0\\
        0 & U^\dagger
    \end{bmatrix}
\;. \label{eq: HQSVT} 
\end{equation} 
\end{theorem}

Using this theorem, we would like to perform a polynomial transformation of $A$ while also preserving the Hamiltonian block-encoding structure; that is we seek to implement the function $e^{-iH[p_{\rm SV}(A)]}$:
\begin{equation}
    \exp\left({-i
    \begin{bmatrix}
        0 & p_{\rm SV}(A^\dagger)\\
        p_{\rm SV}(A) & 0
    \end{bmatrix}}\right)=
    \begin{bmatrix}
        V & 0\\
        0 & U
    \end{bmatrix}
    \begin{bmatrix}
        \cos(p(\Sigma)) & -i\sin(p(\Sigma)) \\
        -i\sin(p(\Sigma)) & \cos(p(\Sigma))
    \end{bmatrix}
    \begin{bmatrix}
        V^{\dagger} & 0\\
        0 & U^{\dagger}
    \end{bmatrix} \,,
    \label{eq: HQSVT desired transformation}
\end{equation}
where $p$ is a degree $d$ polynomial. To achieve the desired transformation,  inspection of Eq.\ \eqref{eq:polycossigma} shows that we need:

\begin{equation}
\label{eq:pq_arcos}
    P(x)\approx \cos(p(\arccos(x))) \quad \text{and} \quad Q^*(x)\approx \frac{\sin(p(\arccos(x)))}{\sqrt{1-x^2}} \,,   
\end{equation}
for $x=\cos(\Sigma)$.

Several difficulties arise from this approach. First, the function $\arccos(x)$ does not have definite parity, which is problematic in the context of Chebyshev QSP, where the achievable polynomials have either odd or even parity. Second, $\arccos(x)$ has singularities close to $x=\pm 1$. Approximating a function near singularities with a polynomial typically introduces bad convergence properties and decreases the quality of the approximation.

In the literature, two methods have been proposed to alleviate these problems. \cite{lloyd2104hamiltonian} noted that if a lower bound $\delta$ on the singular values of $A$ is assumed, then the resulting function can be rescaled to a well-behaved definite-parity polynomial. On the other hand,  \cite{kang2025quantum}, resolved the issues by applying $X$-rotations on the qubitised subspace, in effect lifting the $Z$-restriction on the processing operators. In the following subsection, we critically examine the method of \cite{kang2025quantum}.  \cite{lloyd2104hamiltonian} did not provide an explicit error analysis and the setting is slightly different to our case; as such, we develop our own method based on a lower bound on the singular values of $A$, inspired by both \cite{lloyd2104hamiltonian} and \cite{kang2025quantum}, which is presented and evaluated in Section \ref{subsec: H-QSVT based algorithm using no X gate}.

\subsubsection*{Hamiltonian QSVT with X-rotations on the qubitised subspace}

In \cite{kang2025quantum}, the difficulties with obtaining a polynomial approximation to arccos are resolved by applying the gate $i(X\otimes\mathbb{I})$, i.e.\ an $X$-rotation on the qubitised subspace. This lifts the $Z$-restriction on the processing operators, in effect allowing any SU(2) processing operator as in the most general form of QSP. The physicality of performing $X$-rotations in the setting of Hamiltonian block-encoding, in contrast to $Z$-rotations which can be implemented using phase kickback as shown above, will be discussed in Remark \ref{rem:Xrot}. For now, we note that this has the following effect:
\begin{equation}
    \begin{bmatrix}
        V & 0\\
        0 & U
    \end{bmatrix}
    \begin{bmatrix}
        \cos(\Sigma) & -i\sin(\Sigma)\\
        -i\sin(\Sigma) & \cos(\Sigma)
    \end{bmatrix}
    \begin{bmatrix}
        V^\dagger & 0\\
        0 & U^\dagger
    \end{bmatrix}
    \begin{bmatrix}
        0 & i\mathbb{I}\\
        i\mathbb{I} & 0
    \end{bmatrix}
    =
    \begin{bmatrix}
        V & 0\\
        0 & U
    \end{bmatrix}
    \begin{bmatrix}
        \sin(\Sigma) & i\cos(\Sigma)\\
        i\cos(\Sigma) & \sin(\Sigma)
    \end{bmatrix}
    \begin{bmatrix}
        U^\dagger & 0\\
        0 & V^\dagger
    \end{bmatrix}
    \label{eq: Hamiltonian block-encoding multiplication with X operator}
\end{equation}
As a result, the upper-left entry of Eq.~\eqref{eq: cossigma} becomes $\sin(\Sigma)$ instead of $\cos(\Sigma)$. Taking $x=\sin(\Sigma)$ as the new QSP variable, $\arccos(x)$ then becomes an $\arcsin$, which is considerably more convenient since it is an odd function. Its singularities also occur at $\pm 1$ just like $\arccos$, but now this corresponds to singular values satisfying $\Sigma= \pm \pi/2$. Consequently, one only requires an upper bound on $\|A\|$, which can always be taken as a constant by rescaling the evolution time $\tau$, rather than a lower bound, which would require an additional physical assumption on the spectrum of $A$.

A problem with multiplying by $X$ on the qubitised subspace is that it also changes the right-singular basis, as we can see in Eq.\ \eqref{eq: Hamiltonian block-encoding multiplication with X operator} ($U^\dagger$ and $V^\dagger$ are swapped). Consequently, we cannot simply repeatedly apply the operator of Eq.\ \eqref{eq: Hamiltonian block-encoding multiplication with X operator}, as the interleaved change of basis operators will not mutually cancel. 
To do so, we introduce the new operator:
\begin{equation}
    (Z \otimes \mathbb{I})(-iX\otimes \mathbb{I})e^{iH[A]}(Z\otimes \mathbb{I})=\begin{bmatrix}
        U & 0\\
        0 & V
    \end{bmatrix}
    \begin{bmatrix}
        \sin(\Sigma) & i\cos(\Sigma)\\
        i\cos(\Sigma) & \sin(\Sigma)
    \end{bmatrix}
    \begin{bmatrix}
        V^\dagger & 0\\
        0 & U^\dagger
    \end{bmatrix}
    \label{eq: Hamiltonian block-encoding multiplication with X operator and Z conjugaison}
\end{equation}
Now, by sequentially applying  the operators of Eq.\ \eqref{eq: Hamiltonian block-encoding multiplication with X operator} and Eq.\ \eqref{eq: Hamiltonian block-encoding multiplication with X operator and Z conjugaison}, and interleaving them with $e^{-i\phi(2\Pi_0-\mathbb{I})}$ processing operators,  the inner change of basis operators mutually cancel and we are left with a Chebyshev QSP in terms of the variable $x=\sin(\Sigma)$:
\begin{equation}
    \begin{bmatrix}
        P(\sin(\Sigma)) & iQ(\sin(\Sigma))\sqrt{1-\Sigma^2}\\
        iQ^{*}(\sin(\Sigma))\sqrt{1-\Sigma^2} & P^{*}(\sin(\Sigma))
    \end{bmatrix}
\end{equation}
Now, in \cite{kang2025quantum}  the authors choose $P(x) \approx \sin(p(\arcsin(x)))$ and $Q^*(x)\approx \cos(p(\arcsin(x)))/\sqrt{1-x^2}$  leading to:
\begin{equation}
\begin{bmatrix}
        V & 0\\
        0 & U
    \end{bmatrix}
    \begin{bmatrix}
         \sin(P(\Sigma))&  i\cos(P(\Sigma))\\
         i\cos(P
         ) 
         & 
         \sin(P(\Sigma))
    \end{bmatrix}
    \begin{bmatrix}
        U^\dagger & 0\\
        0 & V^\dagger
    \end{bmatrix} \,,
\end{equation}
such that a final application of $(iX\otimes\mathbb{I})$ yields the desired transformation of Eq.\ \eqref{eq: HQSVT desired transformation}. This leads to the following Theorem:
\begin{theorem}
    [Hamiltonian QSVT with $X$-rotations on the qubitised subspace, odd case (Theorem 2 in \cite{kang2025quantum})]
    Let $H[A]$ be a canonical Hamiltonian block-encoding 
    of a matrix $A$, such that $\|A\| < \pi/2$. Let $\epsilon \in \mathbb{R}^+$. 
    For any real odd polynomial $p \in \mathbb{R}[x]$ of degree $d$ such that 
    $||p||_{[-\pi/2,\pi/2]}\leq 1$, the unitary dynamics:
    \begin{equation}
        e^{-iH[p_{\rm SV}(A)]}=\exp\left({-i
    \begin{bmatrix}
        0 & p_{\rm SV}(A^\dagger)\\
        p_{\rm SV}(A) & 0
    \end{bmatrix}}\right)\
    \end{equation}
    can be constructed to accuracy $\epsilon$ with
\begin{equation}
        \tilde{d} \in \mathcal{O}\left(
         d\log\left(\frac{1}{\epsilon}\right)\right)
\end{equation}
calls to $e^{-iH[A]}$ and $X$-rotations on the qubitised subspace.  
\end{theorem}

A similar theorem can also be obtained for real even polynomials, albeit with additional subtleties. 

\begin{remark}[Allowing $X$-rotations on the qubitised subspace]\label{rem:Xrot}
The ability to perform $X$-rotations on the qubitised subspace is justified in \cite{kang2025quantum} by construction: the original Hamiltonian block encoding $H[A]$ is assumed to be constructed via an explicit auxiliary qubit tensored with the system encoding $A$. However, in our setting, only black-box access to the dynamics $e^{iH[A]}$, together with the $C_{\Pi_0}$-NOT gates that identify the location of $A$, is granted. This is the same setting as in standard QSVT \cite{gilyen2019quantum}, where the matrix $A$ is directly encoded as a block of a unitary. As such, it is not possible in our setting to apply $X$-rotations on the qubitsed subspace, as there is no accessible qubit on which to apply these rotations. 

Note, that in contrast, GQSP assumes direct access to the controlled evolution $\mathrm{ctrl}_0(U)$, giving us physical access to the qubitised two-dimensional subspace and therefore allowing arbitrary single-qubit SU(2) processing operations.
\end{remark}

\section{UHET as a linearisation of GQSP}
\label{sec: Universal Hamiltonian Eigenvalue Transformation as a linearisation of Generalised QSP}
We have now seen multiple algorithms to transform an unknown operator encoded in unitary quantum dynamics. Two of these algorithms, namely UHET (Theorem \ref{thm: UHET}) and GQSP (Theorem \ref{thm:GQSP})   share a similar structure, that is, controlled evolution operators intertwined with rotations on an auxiliary qubit. This motivates the search for a formal connection between the two.

In this section, we show that the (deterministic) UHET algorithm is in fact a linearised limit of the GQSP protocol, as illustrated in Figure \ref{fig:GQSPtoUHET} and formalised in the following Theorem:

\vspace{1ex}

\vspace{1ex}
\begin{theorem}[Relation between GQSP and UHET]
    Let $N,d,n\in \mathbb{N}$ and $U=e^{-i\pi \tilde{H}}\in {\rm SU}(n)$. 
    Let $\mathtt{UHET}_{f_d,t,N}^{\rm det}(\tilde{H})$ be a deterministic UHET protocol with controlled-Hamiltonian access (see Theorem \ref{thm: UHET}) associated to the truncated Fourier series
    \begin{equation}
         f_d(\tilde{H})=\sum_{k=-d}^d\tilde{\theta}_k e^{i\tilde{\phi}_k}e^{i\pi \tilde{H}k} 
    \end{equation}
    of some periodically 3-smooth function $f : [-1,1] \rightarrow \mathbb{C}$.
    Let $\mathtt{GQSP}_{\Theta t/{N},\Phi,\lambda}(U)$ be a degree $2d+1$ GQSP protocol (see Theorem \ref{thm:GQSP}) associated with the processing angles $\Theta t/N=\{\theta_{-d}t/N,\dots,\theta_0 t/{N},\dots,\theta_d t/{N}\}$, $\Phi=\{\phi_{-d},\phi_{-d+1},\dots,\phi_0,\dots,\phi_{d-1},\phi_d\}$ and $\lambda$, where  $\tilde{\theta}_k=-\theta_k$ and $\tilde{\phi}_k=\frac{\pi}{2}+(d-k)\pi+2\sum_{\ell=k}^d \phi_\ell$.
Then,
\begin{align}
    \mathtt{UHET}_{f_d,t,N}^{\rm det}(\tilde{H})
    \doteq \left[\text{ctrl}_0(U^{-d})\cdot \mathtt{GQSP}_{\frac{\Theta t}{N},\Phi,\lambda}(U) \cdot \text{ctrl}_0(U^{-d})R^{\dagger}_{\rm z}(\tilde{\lambda})\right]^N \; ,
\end{align}
where $\doteq$ means equality up to a global phase and $\tilde{\lambda}= \lambda-\frac{\pi}{2}-\frac{(2d+1)\pi}{2}-\sum_{\ell=-d}^{d}\phi_{\ell}$.  

As a direct result, 
\begin{equation}
    \mathtt{GQSP}_{\frac{\Theta t}{N},\Phi,\lambda}(U) \doteq \text{ctrl}_0(U^{d})\cdot
    e^{-iH[f(\tilde{H})]t/N}\cdot \text{ctrl}_0(U^d) R_{\rm z}(\tilde{\lambda}) +\mathcal{O}\left(\frac{\beta^2 t^2}{N^2}\right) \,,
\end{equation}
where $H[f(\tilde{H})]$ is a canonical Hamiltonian block-encoding of the matrix $f(\tilde{H})$ (see Definition \ref{def:standardform}).
\label{thm: UHET=GQSP}
\end{theorem}
\vspace{1ex}

Before proving this theorem, recall that in UHET, the evolution is interleaved with rotations in the XY-plane,
    \begin{equation}
        R_{\rm xy}(\theta,\phi)=e^{-i(\cos(\phi)X-\sin(\phi)Y)\theta}=
        \begin{bmatrix}
            \cos(\theta) & -ie^{i\phi}\sin(\theta)\\
            -ie^{-i\phi}\sin(\theta) &\cos(\theta)
        \end{bmatrix}\; .
    \label{eq: XYrotation}
    \end{equation}
In contrast, GQSP employs general SU(2) processing operators of the form in Eq.~\eqref{eq: GQSP processing operators}. Our first objective is therefore to rewrite the GQSP protocol in a form whose processing operators are XY-rotations; we denote such a protocol by XY-GQSP.

\begin{lemma}[Converting between standard GQSP and XY-GQSP]
Let $\mathtt{GQSP}_{\Theta,\Phi,\lambda}(U)$ be a GQSP protocol using the standard processing operators (see Eq.\ \eqref{eq: GQSP processing operators}) with $\Theta=\{\theta_{0}, \theta_{1}, ... ,\theta_d \}$ and $\Phi= \{\phi_0, \phi_1,..,\phi_d\}$. Then,
\begin{equation}
  \mathtt{GQSP}_{\Theta,\Phi,\lambda}(U) :=\left[\prod_{k=d}^{1} R(\theta_k, \phi_k, 0) \mathrm{ctrl}_0(U)\right]  R(\theta_0, \phi_0, 0)R_{\rm z}(\lambda)  \doteq \left[\prod_{k=d}^{1} R_{\rm xy}(\tilde{\theta}_k,\tilde{\phi}_k) \mathrm{ctrl}_0(U)\right] R_{\rm xy}(\tilde{\theta}_0,\tilde{\phi}_0)R_{\rm z}(\tilde{\lambda})  \; ,
\end{equation}
where $\doteq$ means equality up to a global phase, $R_{\rm xy}(\theta,\phi)$ is a rotation operator about an axis in the XY-plane specified by $(\theta,\phi)$ (see Eq.\ \eqref{eq: XYrotation}) and 
\begin{equation}
    \tilde{\phi}_k = \frac{\pi}{2}+(d-k)\pi+2\sum_{\ell=k}^{d}\phi_\ell \; , \quad \tilde{\theta}_k=-\theta_k, \quad \tilde{\lambda}= \lambda-\frac{\pi}{2}-\frac{d\pi}{2}-\sum_{\ell=0}^{d}\phi_{\ell} \,.
    \label{eq: relation between standard GQSP and XY-GQSP angles} 
\end{equation}
\label{lem: relation between standard GQSP and XY-GQSP}
\end{lemma}

The proof follows from straightforward calculations and is left for Appendix \ref{app:SU2-XY-GQSP-lemma}.

\begin{figure}[H]
    \centering
\scalebox{0.8}{
    \begin{quantikz}[row sep=0.3cm]
        &&\gate{R_{\rm xy}(\theta_{-d},\phi_{-d})}& \ctrl[open]{1} \gategroup[3,steps=1,style={dashed,rounded
        corners, inner xsep=2pt}]{}&\gate{R_{\rm xy}(\theta_{-d+1},\phi_{-d+1})} &\ctrl[open]{1} \gategroup[3,steps=1,style={dashed,rounded
        corners, inner xsep=2pt}]{}&...&\gate{R_{\rm x}(\theta_0)}&...&\ctrl[open]{1} \gategroup[3,steps=1,style={dashed,rounded
        corners, inner xsep=2pt}]{} & \gate{R_{\rm xy}(\theta_{d-1},\phi_{d-1})} & \ctrl[open]{1} \gategroup[3,steps=1,style={dashed,rounded
        corners, inner xsep=2pt}]{}& \gate{R_{\rm xy}(\theta_{d},\phi_{d})}  &   \\
    && & \gate[2, style={rounded corners}]{U}&&\gate[2, style={rounded corners}]{U}&\cdots&&\cdots&  \gate[2, style={rounded corners}]{U} & & \gate[2, style={rounded corners}]{U}  & & \\
    &&&&&&\cdots&&\cdots&&&&&
    \end{quantikz}}

    \vspace{1ex}
     $\shortparallel$
    \vspace{1ex}

\scalebox{0.8}{
    \begin{quantikz}[row sep=0.3cm]
        &\octrl{1}&\octrl{1}&\gate{R_{\rm xy}(\theta_{-d},\phi_{-d})}& \ctrl[open]{1} \gategroup[3,steps=2,style={dashed,rounded
        corners, inner xsep=2pt}]{}&\ctrl[open]{1}&\gate{R_{\rm xy}(\theta_{-d+1},\phi_{-d+1})} &\ctrl[open]{1} \gategroup[3,steps=2,style={dashed,rounded
        corners, inner xsep=2pt}]{}&\ctrl[open]{1}&...&\octrl{1}&\gate{R_{\rm x}(\theta_0)}&\octrl{1}&\cdots \\
    &\gate[2, style={rounded corners}]{U^d}&\gate[2, style={rounded corners}]{U^{-d}}& & \gate[2, style={rounded corners}]{U^{d}}&\gate[2, style={rounded corners}]{U^{-d+1}}&&\gate[2, style={rounded corners}]{U^{d-1}}&\gate[2, style={rounded corners}]{U^{-d+2}}&\cdots&\gate[2, style={rounded corners}]{U^0}&&\gate[2, style={rounded corners}]{U^{0}}& \cdots\\
    &&&&&&&&&\cdots&&&&\cdots
    \end{quantikz}}\\
    \scalebox{0.8}{
    \begin{quantikz}[row sep=0.3cm]
    \cdots&\ctrl[open]{1} \gategroup[3,steps=2,style={dashed,rounded
        corners, inner xsep=2pt}]{} &\ctrl[open]{1}& \gate{R_{\rm xy}(\theta_{d-1},\phi_{d-1})} & \ctrl[open]{1} \gategroup[3,steps=2,style={dashed,rounded
        corners, inner xsep=2pt}]{}&\ctrl[open]{1}& \gate{R_{\rm xy}(\theta_{d},\phi_{d})}  & \octrl{1} & \octrl{1} & \\
    \cdots&  \gate[2, style={rounded corners}]{U^{-d+2}}&\gate[2, style={rounded corners}]{U^{d-1}} & & \gate[2, style={rounded corners}]{U^{-d+1}}& \gate[2, style={rounded corners}]{U^{d}} & & \gate[2, style={rounded corners}]{U^{-d}} & \gate[2, style={rounded corners}]{U^{d}} &\\
    \cdots&&&&&&&&&
    \end{quantikz}}
    
    \vspace{1ex}
     $\shortparallel$
    \vspace{1ex}

\scalebox{0.8}{
    \begin{quantikz}[row sep=0.3cm]
    &\octrl{1}&\octrl{1}\gategroup[3,steps=3,style={dashed,rounded
        corners, inner xsep=0.5pt}]{}&\gate{R_{\rm xy}(\theta_{-d},\phi_{-d})}& \ctrl[open]{1} &\ctrl[open]{1} \gategroup[3,steps=3,style={dashed,rounded
        corners, inner xsep=0.5pt}]{}&\gate{R_{\rm xy}(\theta_{-d+1},\phi_{-d+1})} &\ctrl[open]{1} & \cdots &\octrl{1} \gategroup[3,steps=3,style={dashed,rounded
        corners, inner xsep=0.5pt}]{}&\gate{R_{\rm x}(\theta_0)}&\octrl{1}&\cdots \\
    &\gate[2, style={rounded corners}]{U^d}&\gate[2, style={rounded corners}]{U^{-d}}& & \gate[2, style={rounded corners}]{U^{d}}&\gate[2, style={rounded corners}]{U^{-d+1}}&&\gate[2, style={rounded corners}]{U^{d-1}}&\cdots&\gate[2, style={rounded corners}]{U^{0}}&&\gate[2, style={rounded corners}]{U^{0}}& \cdots\\
    &&&&&&&&\cdots&&&&\cdots
    \end{quantikz}}
\scalebox{0.8}{
 \begin{quantikz}[row sep=0.3cm]
        \cdots&\ctrl[open]{1} \gategroup[3,steps=3,style={dashed,rounded
        corners,inner xsep=0.5pt}]{}& \gate{R_{\rm xy}(\theta_{d-1},\phi_{d-1})} & \ctrl[open]{1} &\ctrl[open]{1} \gategroup[3,steps=3,style={dashed,rounded
        corners, inner xsep=0.5pt}]{}& \gate{R_{\rm xy}(\theta_{d},\phi_{d})}  & \octrl{1} & \octrl{1} & \\
\cdots &\gate[2, style={rounded corners}]{U^{d-1}} & & \gate[2, style={rounded corners}]{U^{-d+1}}& \gate[2, style={rounded corners}]{U^{d}} & & \gate[2, style={rounded corners}]{U^{-d}} & \gate[2, style={rounded corners}]{U^{d}} &\\
\cdots&&&&&&&&
\end{quantikz}}\\
\bigskip
$\bigg\downarrow$ Limit of small angles ($\Theta\rightarrow\Theta t/N$)

\scalebox{0.8}{
\begin{quantikz}[row sep=0.3cm]
         & \ctrl[open]{1}&\ctrl[open]{1} \gategroup[3,steps=3,style={dashed,rounded
            corners, fill=vertceladon!20, inner
            xsep=2pt},background,label style={label ,anchor=north,yshift=0.4cm}]{{\sc Repeat for $k\in \{-d,-d+1,...,d\}$}} & \gate{R_{\rm xy}(\theta_{k}t/N,\phi_k)} & \ctrl[open]{1} & \ctrl[open]{1}&\\
        &\gate[2, style={rounded corners}]{U^d} & \gate[2, style={rounded corners}]{U^k}& &\gate[2, style={rounded corners}]{U^{-k}} & \gate[2, style={rounded corners}]{U^d} & \\
        &&&&&&
\end{quantikz}}
$\overset{\text{qDRIFT}}{\longrightarrow}$
\scalebox{0.8}{
\begin{quantikz}[row sep=0.3cm]
         & \ctrl[open]{1}&\ctrl[open]{1} \gategroup[3,steps=3,style={dashed,rounded
            corners,fill=vertceladon!20, inner
            xsep=2pt},background,label style={label ,anchor=north,yshift=0.4cm}]{{\sc Apply with prob.\ $p_k=c_k/\beta$, repeat $\times N$}} & \gate{R_{\rm xy}(\beta t/N,\phi_k)} & \ctrl[open]{1} & \ctrl[open]{1}&\\
        &\gate[2,style={rounded corners}]{U^{d}} & \gate[2,style={rounded corners}]{U^{k}}& &\gate[2, style={rounded corners}]{U^{-k}} & \gate[2, style={rounded corners}]{U^{d}} & \\
        &&&&&&
\end{quantikz}}
    \caption{Transformation from an XY-GQSP protocol with processing angles $(\Theta, \Phi)$
to the UHET algorithm. Here, $R_{\rm x}(\theta):=e^{-i\theta X} = R_{\rm xy}(\theta,0)$. The first step consists of rewriting the signal operator as
$\mathrm{ctrl}_0(U) \;\longrightarrow\; \mathrm{ctrl}_0(U^{-k+1})\,\mathrm{ctrl}_0(U^{k})$.
Next, we take the small-angle approximation by rescaling the phase parameters
$\Theta \;\longrightarrow\; \Theta t/N$, with $t /N \ll 1$, (linerisation)
so that the sequence approaches a continuous time evolution via Trotterisation, which is equivalent to deterministic UHET (with controlled-Hamiltonian access). Finally,
the resulting procedure is randomised using the qDRIFT protocol to obtain the original randomised UHET (with controlled-Hamiltonian access). (Note that for simplicity, in contrast to Theorem \ref{thm: UHET=GQSP}, here we directly start with an XY-GQSP protocol rather than the equivalent standard SU(2) GQSP, so that the XY-GQSP angles are denoted ($\Theta,\Phi$) without tildes and the parameter $\lambda$ is not required.)}
\label{fig:GQSPtoUHET}
\end{figure}
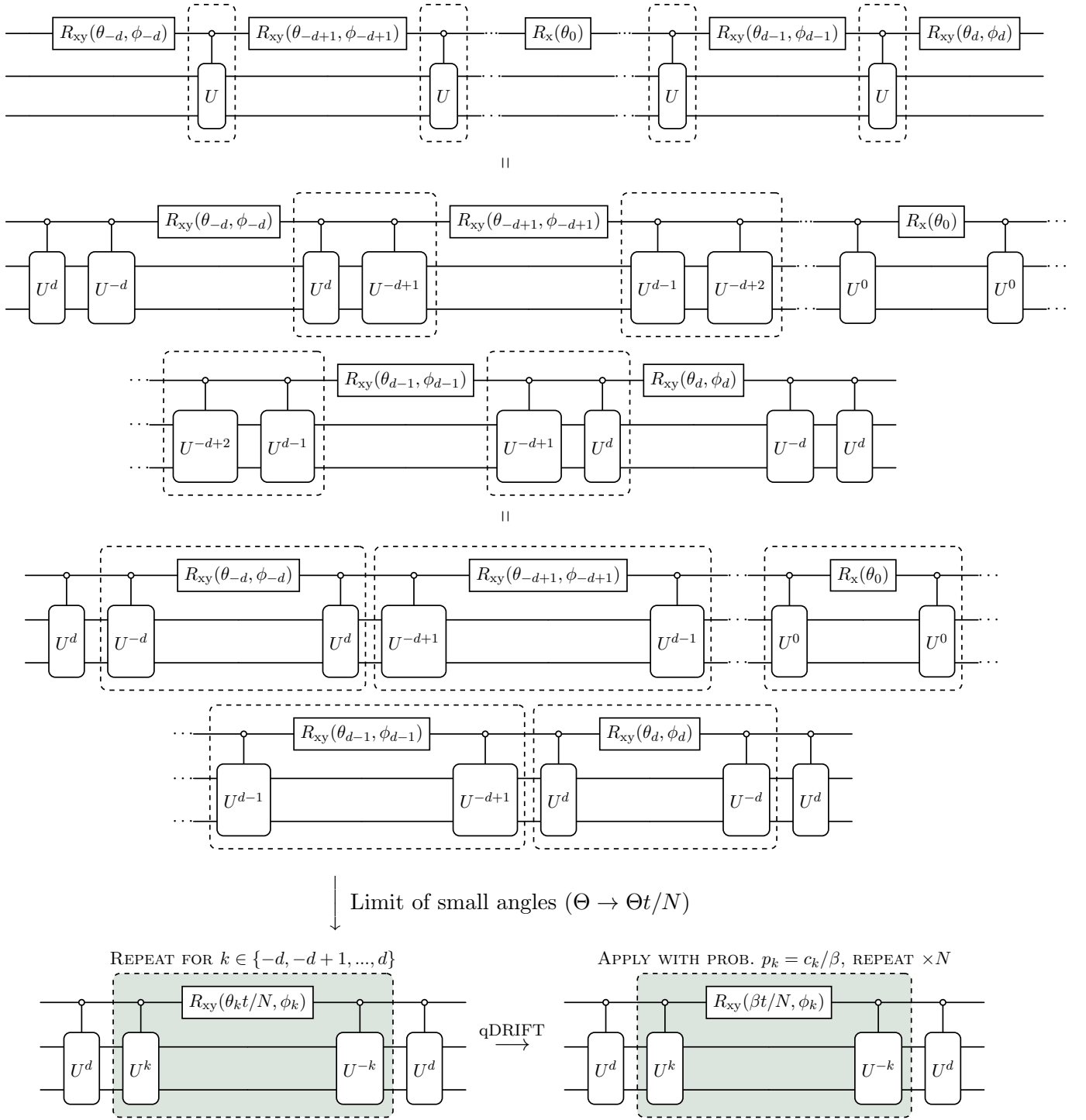

\vspace{1ex}
\noindent \textbf{Proof of Theorem \ref{thm: UHET=GQSP}.}
We begin by invoking Lemma \ref{lem: relation between standard GQSP and XY-GQSP}  to transform the circuit $\mathtt{GQSP}_{\Theta t/N,\Phi,\lambda}(U)$ into an XY{-GQSP with angles $\tilde{\Theta}t/N,\tilde{\Phi},\tilde\lambda$ given by Eq. \eqref{eq: relation between standard GQSP and XY-GQSP angles}. Since the support of $\Theta$ and $\Phi$ is now in $[-d,d]$ rather than $[0,d]$, the final absorption of the angles in $\lambda$ will be summed from $-d$ to $d$ instead of $0$ to $d$, yielding: 
\begin{equation}
    \tilde{\lambda} 
    =\lambda-\frac{\pi}{2}-\frac{(2d+1)\pi}{2}-\sum_{\ell=-d}^{d}\phi_{\ell}.
\end{equation} 
Now, the only remaining step in constructing a circuit that looks exactly like a UHET iterate is the conjugation of the signal operators by $\rm ctrl_0$-$U^{\pm k}$. To make these $\pm k$ appear inside the XY-GQSP, we use:
\begin{equation}
\mathrm{ctrl}_0(U)=\mathrm{ctrl}_0(U^{k})\mathrm{ctrl}_0(U^{-k+1}) \,.
\end{equation}
Applying this transformation inside an XY-GQSP protocol results in:
\begin{align}
    \mathtt{GQSP}_{\Theta t/N, \Phi,\lambda}&\doteq \mathrm{ctrl}_0(U^{d})\mathrm{ctrl}_0(U^{-d})\left[\prod_{k=d}^{-d+1} R_{\rm xy}(\tilde{\theta}_k t/N,\tilde{\phi}_k)\mathrm{ctrl}_0(U^{k})\mathrm{ctrl}_0(U^{-k+1})\right]
    \nonumber\\ &\qquad \qquad \qquad \times R_{\rm xy}(\tilde{\theta}_{-d}t/N,\tilde{\phi}_{-d})\mathrm{ctrl}_0(U^{-d})\mathrm{ctrl}_0(U^{d})R_{\rm z}(\tilde{\lambda}) \nonumber \\
    &=\mathrm{ctrl}_0(U^{d})\left[\prod_{k=d}^{-d}\mathrm{ctrl}_0(U^{-k}) R_{\rm xy}(\tilde{\theta}_k t/N,\tilde{\phi}_k)\mathrm{ctrl}_0(U^{k})\right]\mathrm{ctrl}_0(U^{d})R_{\rm z}(\tilde{\lambda}) \; ,
     \label{eq: GQSP transformed into a UHET circuit}
\end{align}
where from the first to the second line, we transferred the $\mathrm{ctrl}_0(U^{-k+1})$ to the next iterate. Now, the middle product is exactly one UHET iterate. By isolating it, we obtain the following result:
\begin{align}
    &\left[\prod_{k=d}^{-d}\mathrm{ctrl}_0(U^{-k}) R_{\rm xy}(\tilde{\theta}_kt/N,\tilde{\phi}_k)\mathrm{ctrl}_0(U^{k})\right] = \mathrm{ctrl}_0(U^{-d})\cdot\mathtt{GQSP}_{{\Theta} t/N, {\Phi},{\lambda}}(U)\cdot \mathrm{ctrl}_0(U^{-d})R^\dagger_{\rm z}(\tilde{\lambda})\\
    &\overset{\times N}{\longrightarrow}\left[\prod_{k=d}^{-d}\mathrm{ctrl}_0(U^{-k}) R_{\rm xy}(\tilde{\theta}_k t/N,\tilde{\phi}_k)\mathrm{ctrl}_0(U^{k})\right]^N = \left[\mathrm{ctrl}_0(U^{-d})\cdot\mathtt{GQSP}_{\Theta t/N, \Phi,\lambda}(U)\cdot \mathrm{ctrl}_0(U^{-d})R^\dagger_{\rm z}(\tilde{\lambda}) \right]^N  \;.
\end{align}
Now, the left-hand side of the last equation is simply the deterministic UHET algorithm (with controlled-Hamiltonian access). We can therefore write:
\begin{equation}
    \mathtt{UHET}_{f_d,t,N}^{\rm det}(\tilde H) = \left[\mathrm{ctrl}_0(U^{-d})\cdot\mathtt{GQSP}_{{\Theta t /N}, {\Phi},{\lambda}}(U)\cdot \mathrm{ctrl}_0(U^{-d})R^\dagger_{\rm z}(\tilde{\lambda}) \right]^N \,,
\end{equation}
which concludes the proof of Theorem \ref{thm: UHET=GQSP}. 
\qed

This theorem demonstrates that the UHET structure can be recovered from a GQSP protocol. As a direct consequence, we can say that a specific limiting case of a GQSP protocol, namely $\Theta \rightarrow \Theta t/N$ for $t/N \ll 1$, yields the time evolution with a short time $t/N$ of an effective Hamiltonian block-encoding of $f(H)$. In this sense, UHET can be understood as a Trotterised or \textit{linearised} version of GQSP. An interesting feature is that the \textit{processing operator angles serve as the Fourier coefficients of $f$ }in this linearised version. 

Note, that for simplicity, we focused here on the version of UHET for periodically 3-smooth functions; an analogous result can be found for the general UHET protocol for arbitrary 3-smooth functions.

\begin{remark}[Relationship to harmonic analysis]
We finish with a final comment relating the above discussion to harmonic analysis. Recent work \cite{laneve2025generalized} established that GQSP can be identified with the nonlinear Fourier transform (NLFT) over SU(2). In the nonlinear Fourier framework, an SU(2)-valued function is decomposed into an infinite product of SU(2)-valued `frequency' factors, in close analogy with how an ordinary Fourier series decomposes a function into elementary oscillatory modes. 
This analogy also provides an intuitive interpretation of the linearisation procedure introduced earlier. In the regime where the nonlinear Fourier coefficients are sufficiently small, the products defining the NLFT can be expanded to first order in the coefficient magnitude. At this order, the nonlinear interactions between the different frequency components disappear, and the resulting transformation reduces to the ordinary linear Fourier transform. In this sense, the standard Fourier transform can be viewed as the linearised limit of the nonlinear Fourier transform, in close analogy with the way UHET arises as a linearisation of GQSP. We describe these connections in more detail in Appendix \ref{App: The Nonlinear Fourier Transform}.
\end{remark}

\section{Universal Hamiltonian Singular Value Transformation}
\label{sec: Universal Hamiltonian Singular Value Transformation}

In this section, inspired by the linearisation from GQSP to UHET, we develop a new Fourier series simulation algorithm for transforming the
singular values of an arbitrary matrix
$A$, accessed as a Hamiltonian block-encoding in canonical form (Algorithm \ref{alg:SineFourierSeries}), which can be seen as a linearisation of H-QSVT (Section \ref{subsec: HQSVT}).

This algorithm requires only $Z$-rotations on the qubitised subspace,
corresponding to the standard QSVT setting in which control is available
with respect to the blocks of the Hamiltonian associated with $A$ (in contrast to the $X$-rotations required in \cite{kang2025quantum}), and also does not require a lower bound on the singular values of $A$ (in contrast to \cite{lloyd2104hamiltonian}).

This setting enables the singular values of $A$ to be transformed by any (sufficiently differentiable) function $f$, with the sole condition that $f(0)=0$. 

We conclude with a comparison of our method to existing H-QSVT approaches. To put the comparisons on fair ground, we additionally develop a new H-QSVT-based algorithm (Algorithm \ref{alg:QSVT-based algorithm}) that does not require $X$-rotations on the qubitised subspace, but rather requires a lower bound on the singular values of $A$, combining ideas from \cite{kang2025quantum,lloyd2104hamiltonian}. We envisage our UHSVT and H-QSVT algorithms to have applications in a broad range of quantum arithmetic tasks, as well as in Hamiltonian learning tasks when one seeks to learn the properties of a given block within an unknown (control-free) Hamiltonian evolution.
\subsection{Linearisation of H-QSVT}

We begin by applying the linearisation procedure previously used to transform GQSP into UHET, now extending it to H-QSVT. Recall that a key step in this procedure is decomposing the controlled evolution as $\mathrm{ctrl}_0(U) = \mathrm{ctrl}_0(U^k) \, \mathrm{ctrl}_0(U^{-k+1})$. Analogously, we decompose the Hamiltonian evolution operator in H-QSVT as
\begin{equation}
e^{-i\pi H[A]} = e^{-i\pi  H[A](k+1)} \, e^{i\pi  H[A]k}.
\end{equation}Substituting this decomposition into the H-QSVT sequence from Theorem \ref{thm: HQSVT} yields:
\begin{align}
& e^{-i\theta_K(2\Pi_0-I)} \prod_{k=K-1}^{0} e^{-i\pi H[A]} e^{-i\theta_k(2\Pi_0-I)} \nonumber \\
&= \left( e^{-i\pi  H[A]K} e^{i\pi  H[A]K} \right) e^{-i\theta_K(2\Pi_0-I)} \prod_{k=K-1}^{0} \left( e^{-i\pi  H[A](k+1)} e^{i\pi H[A]k} \right) e^{-i\theta_k(2\Pi_0-I)} \nonumber \\
&= e^{-i\pi  H[A]K} \prod_{k=K}^{0} \left( e^{i\pi  H[A]k} e^{-i\theta_k(2\Pi_0-I)} e^{-i\pi  H[A]k} \right).
\end{align}
The resulting sequence exhibits a structure closely resembling that of UHET. Specifically, it consists of a rotation operator, $e^{-i\theta_k(2\Pi_0-I)}$, conjugated by the Hamiltonian evolution $e^{\pm i\pi k H[A]}$, which corresponds to an evolution time of $\tau_k = \pi k$. To understand how this structure yields linearisation, we next analyse the effective Hamiltonian generated by each conjugated rotation operator. Using the identity $W e^{i H} W^\dagger = e^{i W H W^\dagger}$, we obtain:
\begin{align}
    e^{i\pi H[A]k} \exp\left( -i\theta_k(2\Pi_0 - I) \right) e^{-i\pi H[A]k} 
    = \exp\left( -i \underbrace{\theta_k e^{i\pi H[A]k} (2\Pi_0 - I) e^{-i\pi H[A]k}}_{H_{\mathrm{eff}}^{(k)}} \right).
\end{align}
Therefore, using the singular value decomposition of $A=U\Sigma V^\dagger$ to transform $e^{-i\pi H[A]}$ in an analogous way to Eq.~\eqref{eq: decomposition Hamiltonian block-encoding}, the effective Hamiltonian can be written as:
\begin{align}
    H_{\mathrm{eff}}^{(k)} 
    &= \theta_k 
    \begin{bmatrix}
        \cos(\pi\Sigma k) & i\sin(\pi\Sigma k) \\
        i\sin(\pi\Sigma k) & \cos(\pi\Sigma k)
    \end{bmatrix}
    \overbrace{
    \begin{bmatrix}
        I & 0 \\
        0 & -I
    \end{bmatrix}}^{2\Pi_0 - I}
    \begin{bmatrix}
        \cos(\pi\Sigma k) & -i\sin(\pi\Sigma k) \\
        -i\sin(\pi\Sigma k) & \cos(\pi\Sigma k)
    \end{bmatrix} \nonumber \\
    &= \theta_k
    \begin{bmatrix}
        \cos^2(\pi\Sigma k) - \sin^2(\pi\Sigma k) & -2i\sin(\pi\Sigma k)\cos(\pi\Sigma k) \\
        2i\sin(\pi\Sigma k)\cos(\pi\Sigma k) & \sin^2(\pi\Sigma k) - \cos^2(\pi\Sigma k)
    \end{bmatrix} \nonumber \\
    &= \theta_k
    \begin{bmatrix}
        \cos(2\pi\Sigma k) & -i\sin(2\pi\Sigma k) \\
        i\sin(2\pi\Sigma k) & -\cos(2\pi\Sigma k)
    \end{bmatrix} \,,
    \label{eq: effective hamiltonian for fourier mode k}
\end{align} 
where in the last line, we used the identities $\cos^2(\theta)-\sin^2(\theta)=\cos(2\theta)$ and $2\sin(\theta)\cos(\theta)=\sin(2\theta)$, and for simplicity we have omitted the basis change unitaries from the singular value decomposition. The final step of linearisation involves taking the small-angle limit, setting $\theta_k \to \theta_k/N$, and applying the first-order Trotter-Suzuki product formula:
\begin{equation}
    e^{-i \frac{\theta_K}{N} (2\Pi_0 - I)} \prod_{k=K-1}^{0} e^{-i\pi H[A]} e^{-i \frac{\theta_k}{N} (2\Pi_0 - I)} 
    = e^{-i\pi H[A]K} \exp\left( -i \sum_{k=0}^K \frac{H^{(k)}_{\mathrm{eff}}}{N} \right) + \mathcal{O}\left( \frac{1}{N^2} \right).
\end{equation}
Thus, up to first order in $1/N$, the linearised version of H-QSVT implements the effective evolution governed by:
\begin{equation}
    \frac{1}{N}\sum_{k=0}^K H^{(k)}_{\mathrm{eff}} 
    = \frac{1}{N}
    \begin{bmatrix}
        \sum_{k=0}^K \theta_k \cos(2\pi\Sigma k) & -i\sum_{k=0}^K \theta_k \sin(2\pi\Sigma k) \\
        i\sum_{k=0}^K \theta_k \sin(2\pi\Sigma k) & -\sum_{k=0}^K \theta_k \cos(2\pi\Sigma k)
    \end{bmatrix}.
    \label{eq: effective first order approximation for H-QSVT}
\end{equation}
The diagonal entries correspond to real cosine Fourier series, whereas the off-diagonal entries yield imaginary sine Fourier series. Crucially, the coefficients of these Fourier series are directly set by the phase angles $\theta_k$ of the H-QSVT processing operators. This is in direct analogy to the linearisation of GQSP, where the angles from the processing operators are the coefficients of the (complete) Fourier series. 
\subsection{Sine series simulation (UHSVT)}
We now apply this linearised version of H-QSVT to construct a new algorithm, \emph{Universal Hamiltonian Singular Value Transformation} (UHSVT), which produces a Hamiltonian block-encoding of a target function $f^{\mathrm{SV}}(A)$. Specifically, the algorithm transforms the initial Hamiltonian evolution as follows:
\begin{equation*}
    \exp\left( -i
    \begin{bmatrix}
        0 & A^\dagger \\
        A & 0
    \end{bmatrix} t \right)
    =
    \begin{bmatrix}
        V & 0 \\
        0 & U
    \end{bmatrix}
    \exp\left( -i
    \begin{bmatrix}
        0 & \Sigma \\
        \Sigma & 0
    \end{bmatrix} t \right)
    \begin{bmatrix}
        V^\dagger & 0 \\
        0 & U^\dagger
    \end{bmatrix}
\end{equation*}

\begin{equation}
    \bigg\downarrow{\;\mathrm{UHSVT}\;}
\end{equation}

\begin{equation*}
    \begin{bmatrix}
        V & 0 \\
        0 & U
    \end{bmatrix}
    \exp\left( -i
    \begin{bmatrix}
        0 & f^{\mathrm{SV}}(\Sigma) \\
        f^{\mathrm{SV}}(\Sigma) & 0
    \end{bmatrix} t \right)
    \begin{bmatrix}
        V^\dagger & 0 \\
        0 & U^\dagger
    \end{bmatrix}
    =
    \exp\left( -i
    \begin{bmatrix}
        0 & f^{*\mathrm{SV}}(A^\dagger) \\
        f^{\mathrm{SV}}(A) & 0
    \end{bmatrix} t \right).
\end{equation*}

To achieve this transformation, one could directly employ the first-order linearised H-QSVT from Eq.~\eqref{eq: effective first order approximation for H-QSVT}. However, an immediate issue arises: the resulting diagonal terms are non-zero, breaking the desired block-off-diagonal Hamiltonian structure. 
To resolve this issue,  we construct new effective Hamiltonians $H^{(k,s)}_{\mathrm{eff}}$, where $s\in \{0,1\}$ is a boolean variable controlling the sign of the diagonal entries. As we will show, by alternating the sign of the on-diagonal terms across Trotter steps,  these diagonal contributions mutually annihilate under the Trotterisation averaging procedure.

Regarding the off-diagonal entries, Eq.~\eqref{eq: effective first order approximation for H-QSVT} reveals that the implemented function takes the form of a sine Fourier series. Recall that for a general function $f: [-L, L] \to \mathbb{C}$, its sine Fourier series is defined as:
\begin{equation}
    f(x) = \sum_{k=0}^{\infty} b_k \sin\left( \frac{\pi k x}{L} \right),
\end{equation}
where the Fourier coefficients are given by:
\begin{equation}
    b_k = \frac{1}{L} \int_{-L}^{L} f(x) \sin\left( \frac{\pi k x}{L} \right) {\rm d}x \,.
\end{equation}
A broad class of functions admits accurate approximations by truncated
sine Fourier series (see \cite{walker2001fourier} for details). 

Our strategy will therefore be to approximates $f^{\rm SV}(A)$ with a sine Fourier series.  Because the singular values of $A$ are real and non-negative, we can assume a known constant bounds the largest singular value. Without loss of generality we can take $A$ to have singular values in $[0,1]$
(the singular values can always be rescaled to this domain by appropriately rescaling the evolution time), and  thus we take the function $f$ to have domain $f : [0,1] \rightarrow \mathbb{C}$ .

To admit a sine–series representation, we extend $f$ to an
odd function defined on the larger interval $[-2,2]$. For this to be possible, the original function $f$ needs to be sufficiently differentiable (at least 2-smooth, see Appendix \ref{App_subsec: Fourier series}) and satisfy $f(0)=0$ (see Appendix \ref{App_subsec: Extension of smooth functions}). 
We introduce the odd-periodic extension $\tilde{f}_{\mathrm{odd}}$ defined by:
\begin{equation}
\tilde{f}_{\mathrm{odd}}(x)=
\begin{cases}
    \begin{array}{ll@{\qquad}ll}
        f(x),    & x \in [0,1],        & -f(-x),    & x \in [-1, 0] \\
        g(x),    & x \in [1, 2], & -g(-x),    & x \in [-2, -1]
    \end{array} \; ,
\end{cases}
\label{eq:oddextension}
\end{equation}
where the auxiliary function $g$ is a polynomial interpolation chosen such that its even order derivatives (up to at least the second derivative) vanish at the boundaries $x=\pm2$,
and such that $g$ smoothly matches $f$ at $x=1$. This construction ensures that $\tilde{f}_{\mathrm{odd}}$ is \textit{an odd-periodically $2$-smooth function} (see Definition \ref{def: odd-periodically J-smooth} in Appendix \ref{App: Theory of Fourier approximation}), which in turn guarantees a sufficiently
rapid decay of the sine Fourier coefficients as a function of the
cutoff parameter $K$ (see Appendix \ref{App_subsec: Extension of smooth functions}). 

We define the truncated sine series of $\tilde{f}_{\mathrm{odd}}$ as 
\begin{equation}
   \tilde{f}_{\mathrm{odd},K} =
\sum_{k=0}^{K} b_k
\sin\!\left(\frac{\pi k x}{2}\right) \; , 
\end{equation}
and choose the cutoff $K$ of the Fourier coefficients
$\{b_k\}_{k=1}^{K}$ such that
\begin{equation}
\epsilon_K=\left\|f(x)-\tilde{f}_{\mathrm{odd},K}(x)\right\|_{[0,1]}\le \epsilon/4t \,,
\label{eq: truncerrorsinefourierseries}
\end{equation}
where $\|h(x)\|_{[a,b]}$ is the infinity norm and corresponds to the maximum value taken by $f$ in the interval $[a,b]$.

The next step is to add up all the sine frequencies up to the threshold frequency $K$ using a Trotter-Suzuki method. Here, we focus on the randomised (qDRIFT-based) version to adhere to the original formulation of UHET in \cite{odake2025universal} but we will discuss the scaling of the deterministic approaches later. 
The number of qDRIFT iterations is selected as
$N=N(\beta,t,\epsilon/2)$ according to
Eq.~\eqref{eq:qDRIFT}, where
$\beta=\sum_{k=0}^{K}|b_k|$. 

\begin{figure}[htb]
    \centering
        \scalebox{0.80}{\begin{quantikz}[wire types={q,q,n,q},classical gap=0.07cm]
            \lstick{$\mathcal{H}_Q$} & \slice[style={wineRed, shorten <=-0.5cm, shorten >=-0.5cm},label style={at={(0,0)}, anchor=north, yshift=-3.5cm}]{\textcolor{wineRed}{(\textrm{III})}} \gategroup[4,steps=9,style={dashed,rounded
            corners,fill=vertceladon!20, inner
            xsep=2pt},background,label style={label ,anchor=north,yshift=0.4cm}]{{\sc Apply with prob.\ 
            $p_{k,s}=|b_k|/2\beta$, repeat $\times N$}}  & \gate{Z^s} \slice[style={wineRed, shorten <=-0.5cm, shorten >=-0.5cm},label style={at={(0,0)}, anchor=north, yshift=-3.5cm}]{\textcolor{wineRed}{(\textrm{II})}} &\gate{R^\dagger_{\rm z}(\frac{\phi_k}{2}-\frac{\pi}{4})}) \slice[style={wineRed, shorten <=-0.5cm, shorten >=-0.5cm},label style={at={(0,0)}, anchor=north, yshift=-3.5cm}]{\textcolor{wineRed}{(\textrm{I})}} &  \gate[4,style={rounded corners}]{e^{-i \pi H[A]k/2}}  & \gate{R_{\rm z}[(-1)^s\beta t/N]} & \gate[4,style={rounded corners}]{e^{i \pi H[A] k /2}} \slice[style={wineRed, shorten <=-0.5cm, shorten >=-0.5cm},label style={at={(0,0)}, anchor=north, yshift=-3.5cm}]{\textcolor{wineRed}{(\textrm{I})}} & \gate{R_{\rm z}(\frac{\phi_k}{2}-\frac{\pi}{4})} \slice[style={wineRed, shorten <=-0.5cm, shorten >=-0.5cm},label style={at={(0,0)}, anchor=north, yshift=-3.5cm}]{\textcolor{wineRed}{(\textrm{II})}} & \gate{Z^s} \slice[style={wineRed, shorten <=-0.5cm, shorten >=-0.5cm},label style={at={(0,0)}, anchor=north, yshift=-3.5cm}]{\textcolor{wineRed}{(\textrm{III})}} & & \rstick[4]{\textcolor{wineRed}{$e^{-iH[f^{\rm SV}(A)]t}$}} \\
            & & & & & & & & & &\\
            & & \lstick{\vdots} & \lstick{\vdots} & & & & \lstick{\vdots}& \lstick{\vdots} & &\\
            & & & & & & & & & &
        \end{quantikz}}
    \caption{Sine Fourier series simulation for the singular values of $A$ block-encoded into a (control-free) Hamiltonian evolution $H[A]$. By choosing the frequency $k\in\{ 1,2,...,K\}$ and bit $s\in\{0,1\}$ with probability $p_{k,s}$ related to the sine-Fourier coefficients of the odd extension $\tilde{f}_{\rm odd}$ of any sufficiently differentiable function $f$ satisfying $f(0)=0$, one can achieve a functional transformation $A\rightarrow f^{\rm SV}(A)$ of the singular values of $A$ in the form of a truncated sine Fourier series up to order $K$. Note that the \textit{qubitised subspace} $\mathcal{H}_Q$ is not realised by any physical qubit, but the $Z$-rotations on this Hilbert space can be implemented with an additional auxiliary qubit, as shown in Figure \ref{fig:bob}.
\label{fig:sineFourierseries}}
\end{figure}
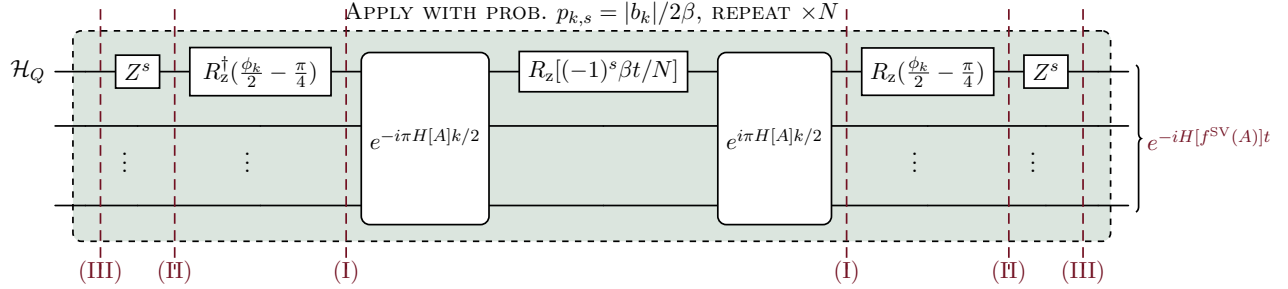

We can understand the procedure by analysing the effective Hamiltonians generated at each stage. For a fixed $k$, we consider two Hamiltonians $H^{(k,s)}_{\rm eff}$ (for $s \in \{0, 1\}$), which are sampled with equal probability of $1/2$. Then, we perform a qDRIFT procedure over all $k$, choosing each $k$ with probability $|b_k|/\beta$, so that the choice of any pair $(k,s)$ is sampled with probability $p_{k,s}=|b_k|/2\beta$.

Using the identity $U e^{-iH} U^\dagger = e^{-i(UHU^\dagger)}$, we compute $H^{(k,s)}_{\rm eff}$ through the three steps of the circuit: \textcolor{wineRed}{(\textrm{I})}, \textcolor{wineRed}{(\textrm{II})}, and \textcolor{wineRed}{(\textrm{III})}, illustrated in Figure \ref{fig:sineFourierseries}.

\begin{enumerate}[label={\textcolor{wineRed}{(\Roman*)}}]
    \item \textbf{Initial Evolution:} Let $e^{-iH^{(k,s)}_{\rm eff}} \leftarrow e^{i \pi H[A] k/2} R_{\rm z}[(-1)^{s} \beta t/N] e^{-i\pi H[A] k/2}$. Using a similar calculation as the one that led to Eq.~\eqref{eq: effective hamiltonian for fourier mode k} we obtain:  
    \begin{align} 
    H^{(k,s)}_{\rm eff}=(-1)^s\begin{bmatrix}
        \cos(\pi \Sigma k) & -i\sin(\pi \Sigma k)\\
        i\sin(\pi \Sigma k) & -\cos(\pi \Sigma k)
    \end{bmatrix}\beta t/N \;.
\end{align} 

\item\textbf{Phase Rotation:} Applying $e^{-iH^{(k,s)}_{\rm eff}} \leftarrow R_{\rm z}(\frac{\phi_k}{2}-\frac{\pi}{4}) e^{-iH^{(k,s)}_{\rm eff}}R^\dagger_{\rm z}(\frac{\phi_k}{2}-\frac{\pi}{4})$ yields:
\begin{align} 
    H_{\rm eff}^{(k,s)}
    &=(-1)^s\begin{bmatrix}
        e^{-i(\frac{\phi_k}{2}-\frac{\pi}{4})}I & 0\\
         0 & e^{i(\frac{\phi_k}{2}-\frac{\pi}{4})}I
    \end{bmatrix}
    \begin{bmatrix}
        \cos(\pi \Sigma k) & -i\sin(\pi \Sigma k)\\
        i\sin(\pi \Sigma k) & -\cos(\pi \Sigma k)
    \end{bmatrix}
    \begin{bmatrix}
        e^{i(\frac{\phi_k}{2}-\frac{\pi}{4})}I & 0\\
         0 & e^{-i(\frac{\phi_k}{2}-\frac{\pi}{4})}I
    \end{bmatrix}\frac{\beta t}{N} \nonumber\\
    &=(-1)^s\begin{bmatrix}
        \cos(\pi \Sigma k) & e^{-i\phi_k} \sin(\pi \Sigma k)\\
         e^{i\phi_k}\sin(\pi \Sigma k) & -\cos(\pi \Sigma k)
    \end{bmatrix}\frac{\beta t}{N} \; .
\end{align}
Here, $e^{i\phi_k}$ corresponds to the phase of the Fourier coefficients $b_k=|b_k|e^{i\phi_k}$
\item \textbf{Sign Control:} Applying $e^{-iH^{(k,s)}_{\rm eff}} \leftarrow Z^s e^{-iH^{(k,s)}_{\rm eff}} Z^s$ results in:
\begin{align} 
    H_{\rm eff}^{(k,s)}
    &=(-1)^s\begin{bmatrix}
        I& 0\\
         0 & (-1)^sI
    \end{bmatrix}
    \begin{bmatrix}
        \cos(\pi \Sigma k) & e^{-i\phi_k} \sin(\pi \Sigma k)\\
         e^{i\phi_k}\sin(\pi \Sigma k) & -\cos(\pi \Sigma k)
    \end{bmatrix}
    \begin{bmatrix}
        I & 0\\
         0 & (-1)^sI
    \end{bmatrix}\beta t/N \nonumber\\
    &=\begin{bmatrix}
        (-1)^s\cos(\pi \Sigma k) &  e^{-i\phi_k}\sin(\pi \Sigma k)\\
         e^{i\phi_k}\sin(\pi \Sigma k) & (-1)^{s+1}\cos(\pi \Sigma k)
    \end{bmatrix}\beta t/N \; .
\end{align}

\end{enumerate}
\noindent As mentioned earlier, $s$ serves as a control boolean variable for the sign of the diagonal elements. This ensures that in the overall qDRIFT protocol, the diagonal components of the Hamiltonian are suppressed by averaging out over $s$, preserving the canonical block-encoding structure. After the qDRIFT procedure, the final averaged  Hamiltonian becomes:

\begin{align}
    H_{\rm eff}= N \times \sum_{k=1,s\in \{0,1\}}^K p_{k,s}H^{(k,s)}_{\rm eff}\  &=\sum_{k=1}^K\frac{|b_k|}{2\beta}\begin{bmatrix}
       0 & 2 e^{-i\phi_k}\sin(\pi \Sigma k)\\
        2 e^{i\phi_k}\sin(\pi \Sigma k) & 0
    \end{bmatrix}\beta t \nonumber \\
    &=\begin{bmatrix}
       0 & \sum_{k=1}^K|b_k|e^{-i\phi_k}\sin(\pi \Sigma k)\\
        \sum_{k=1}^K|b_k|e^{i\phi_k}\sin(\pi \Sigma k) & 0
    \end{bmatrix}t \nonumber\\
    &\approx\begin{bmatrix}
       0 & f^{*\rm SV}(A^\dagger)\\
        f^{\rm SV}(A)& 0
    \end{bmatrix}t=H[  f^{\rm SV}(A)]t \,.
\end{align}
It is now clear how each component contributes to the construction of
$f^{\rm SV}(A)$. 
\textcolor{wineRed}{(\textrm{I})} generates the sine terms associated
with frequency $k$ of the sine Fourier expansion in the off-diagonal
block of the effective Hamiltonian. However, this step also introduces
undesired diagonal contributions. 
\textcolor{wineRed}{(\textrm{II})} subsequently incorporates the phase
factor of the $k$-th Fourier coefficient, $e^{i\phi_k}$. Note that for a real function $f: [0,1]\rightarrow \mathbb{R}$, the coefficients $b_k$ only take real values, and therefore in that case $e^{i\phi_k}$ is restricted to $\pm1$ to implement the sign of the coefficient.
Finally, \textcolor{wineRed}{(\textrm{III})} flips the sign of the
diagonal entries, ensuring that these unwanted terms cancel in the
overall time evolution.

\subsection{Error and scaling of UHSVT}

The time complexity of the UHSVT algorithm is proportional to the number of
iterations $N$ of the main loop, which follows the standard
qDRIFT scaling,
$N \in \mathcal{O}\!\left({\beta^{2}t^{2}/\epsilon}\right)$,
where $\beta=\sum_{k=1}^{K}|b_k|$. Since the Fourier coefficients satisfy a power law decay for sufficiently smooth  functions, i.e.\ 
$|b_k| \in \mathcal{O}(k^{-J})$ for any odd-periodically $J$-smooth $\tilde{f}_{\rm odd}$ (see Lemma \ref{th:decaypropertiessinefourierseries} in  Appendix \ref{App_subsec: Sine series}), which is guaranteed to exist for any function $f:[0,1]\to \mathbb{C}$ such that $f^{(2m)}(0)=0$ for all $2m < J$ (see Lemma \ref{lem: J-smooth odd extension} in Appendix \ref{App_subsec: Extension of smooth functions}), $\beta$ converges to a finite constant
independent of $K$:
\begin{equation}
\beta=\sum_{k=1}^{K}|b_k|\in \mathcal{O}\left(\sum_{k=1}^K \frac{1}{k^J}\right)=\mathcal{O}(1) \textrm{ for odd-periodically $J$-smooth $\tilde{f}_{\rm odd}$ with $J\geq 2$}. 
\end{equation}
Therefore, the  total time complexity scales as:
\begin{equation}
\mathcal{O}\left(\frac{t^2}{\epsilon}\right) \; .
\end{equation}

For the total evolution time, we first compute the average evolution time of a single qDRIFT iteration, which is obtained by weighting
the evolution time $\mathcal{O}(k)$ (for letting the Hamiltonians run for time $k$) by the sampling probabilities
$p_{k,s}$,
\begin{equation}
\sum_{k=1}^{K}\sum_{s\in\{0,1\}}
p_{k,s}\,\mathcal{O}(k)
=
\mathcal{O}\!\left(
\sum_{k=1}^{K}\frac{|b_k|\,k}{2\beta}
\right).
\end{equation}
Multiplying by the total number of iterations yields the
expected total evolution time of
\begin{equation}
   \mathcal{O}\!\left(
\frac{\beta^{2}t^{2}}{\epsilon}
\right)
\sum_{k=1}^{K}\frac{|b_k|\,k}{\beta}
=
\mathcal{O}\!\left[\beta
\left(\sum_{k=1}^{K}k|b_k|\right)
\frac{t^{2}}{\epsilon}
\right] \,.
 \label{eq: Qdrift UHSVT total evolution in terms of sine series coefficient}
\end{equation}
Since the Fourier coefficients satisfy the decay behaviour
$|b_k|=\mathcal{O}(k^{-J})$ (see above), we have that $\beta$ converges and:
\begin{align} 
    &\sum_{k=1}^K k|b_k| \in \begin{cases}
        \mathcal{O}(\log(K)) & \textrm{ for odd-periodically $2$-smooth $\tilde{f}_{\rm odd}$} \\
        \mathcal{O}(1) & \textrm{ for odd-periodically $J$-smooth ($J\geq3$) $\tilde{f}_{\rm odd}$ \,,}
    \end{cases}
    \label{eq: convergence of the moments of sine series coefficients}
\end{align}
where the $\log(K)$ appears because $\sum_{k=1}^K 1/k \in \mathcal{O}(\log(K))$.

Additionally, since $  \epsilon_K \in \mathcal{O}(K^{-(J-1)})$ for any  odd-periodically $J$-smooth $\tilde{f}_{\rm odd}$ 
(see Lemma \ref{th:decaypropertiessinefourierseries} in  Appendix \ref{App_subsec: Sine series}) and that we must choose $\epsilon_K\leq \epsilon/t$ from Eq.\ \eqref{eq: truncerrorsinefourierseries}, the total evolution time scales as:
\begin{equation}
\begin{cases}
\mathcal{O}\left[\frac{t^2}{\epsilon}\log(\frac{t}{\epsilon})\right] & \text{for odd-periodically $2$-smooth $\tilde{f}_{\rm odd}$}  
\\
\mathcal{O}\!\left(t^{2}/\epsilon\right) & \text{for odd-periodically $J$-smooth ($J\geq 3$) $\tilde{f}_{\rm odd}$} \,. 
\end{cases}
\end{equation}

The full UHSVT procedure is given in Algorithm \ref{alg:SineFourierSeries}, with the proof of scaling provided in Appendix \ref{app:proof_rand_UHSVT}. Concisely, UHSVT can be formalised into a theorem, the proof of which is equivalent to that of Algorithm \ref{alg:SineFourierSeries}.

\begin{theorem}[Universal Hamiltonian Singular Value Transformation (randomised, canonical Hamiltonian block-encoding access)]

Let $H[A]$ be the \textit{canonical Hamiltonian block-encoding} (Definition \ref{def:standardform}) of matrix $A$.
Then, for any 2-smooth function $f: [0,1]\rightarrow \mathbb{C} $ with $f(0)=0$, time $t>0$, and error $\epsilon>0$, there exist $K ,N\in \mathbb{N}$, with $N \in \mathcal{O}\left(t^2/\epsilon\right), K \in \mathcal{O}\left(t/\epsilon\right)$, such that for any square matrix $A$ (normalised such that $||A||_{\rm op}\leq 1$)
the randomised application of the unitary sequence
\begin{equation}
 U_{k,s}^{H[A],\beta t/N, \phi_k} :=
     Z^s R_{\rm z}^{\rm QS} \! \left( \! \frac{2\phi_k\! -\! \pi}{4} \!\right)
        e^{i  H[A] \frac{k \pi}{2}} 
  R_{\rm z}^{\rm QS}\! \left( \! (-1)^{s}\frac{ \beta  t \! }{N}\right) 
      e^{-i H[A] \frac{k \pi}{2}}  
       R_{\rm z}^{\rm QS}\! \left(\! \frac{\pi \! -\! 2\phi_k}{4}\!  \right) Z^s
\end{equation}
repeated $N$ times with probability $| b_k|/(2\beta)$ satisfies
\begin{equation}
      \left\vert \left\vert   \sum_{k=1}^N \sum_{s=0}^1
      \frac{| b_k|}{2\beta}
     U_{k,s}^{H[A], \beta t/N,\tilde \phi_k} (\cdot)  {U_{k,s}^{H[A], \beta t/N,\tilde \phi_k} }^\dagger
  - e^{-i  H[f^{\rm SV}(A)] t} (\cdot ) e^{i  H[f^{\rm SV *}(A)] t} \right\vert \right\vert_\diamond
\leq \epsilon
   \,,
\end{equation}
where $\tilde f_{{\rm odd}, K}(x):=\sum_{k=0}^K b_k \sin \left(\pi kx/2 \right)$ is the truncated Fourier series of the odd periodic extension $\tilde f_{\rm odd}$ of $f$ defined in Eq.\ \eqref{eq:oddextension}, $ b_k = | b_k|e^{i  \phi}$, $\beta=\sum_{k=1}^{K}|b_k|$ and  $R_{\rm z}^{\rm QS}(\theta):=  e^{-i\theta(2\Pi_{0}-I)} $ is a $Z$-rotation on the qubitised subspace. Moreover, the time complexity of the algorithm is  $\mathcal{O}(N)$, while the average total evolution time is $\mathcal{O}(\frac{t^2}{\epsilon}\log(\frac{t}{\epsilon}))$.
\end{theorem}

\begin{algorithm}[H]
    \floatname{algorithm}{Algorithm}
    \caption{Universal Hamiltonian Singular Value Transformation (randomised)}
    \begin{algorithmic}[1]
        \Statex{\textbf{Input:}}        
        \begin{itemize}
            \item A finite number of queries to a black-box Hamiltonian dynamics $e^{\pm iH[A]\tau}\in \mathcal{L}(\mathcal{H}_0\oplus \mathcal{H}_1)$ with $\tau \geq 0$ of a seed Hamiltonian $H[A]$, which is a Hamiltonian block-encoding of an arbitrary matrix $A\in \mathcal{L}(\mathcal{H}_0\rightarrow \mathcal{H}_1)$ in the canonical form (see Def. \ref{def:standardform}) and such that $\|A\|\leq 1$.
            \item A $2$-smooth function  $f:[0,1]\rightarrow \mathbb{C}$ such that $f(0)=0$.
            \item Allowed error $\epsilon >0$.
            \item Time $t>0$.
        \end{itemize}
        \Statex{\textbf{Output:}}
        A random unitary approximating
        \begin{align*}
            \exp\left(
                -i\begin{bmatrix}
                    0 & f^{*\rm SV}(A^{\dagger}) \\
                    f^{\rm SV}(A) & 0
                \end{bmatrix}t
            \right) \; ,
        \end{align*}
        with an error according to the diamond norm upper bounded by $\epsilon$.
    \Statex \hrulefill
    	\Statex{\textbf{Time complexity:}}
    	$\Theta(t^2\epsilon^{-1})$
            \Statex{\textbf{Total evolution time:}}
            $\mathcal{O}(t^2\epsilon^{-1}\log(t\epsilon^{-1}))$.
    	\Statex{\textbf{Used Resources:}}
        	\Statex \hskip 1.0em System: Hilbert space $\mathcal{H}=\mathcal{H}_0 \oplus \mathcal{H}_1$, plus one auxiliary qubit $\mathcal{H}_{\rm aux}$.
    	\Statex \hskip1.0em Gates: $e^{-iH[A]\tau}$ on $\mathcal{H}_0 \oplus \mathcal{H}_1$, $C_{\Pi_0}$-NOT gates on $(\mathcal{H}_0\oplus \mathcal{H}_1)\otimes \mathcal{H}_{\rm aux}$ (see Fig. \ref{fig:bob}) and $R_{\rm z}(\theta)$ on $\mathcal{H}_{\rm aux}$. 
    	\Statex \hrulefill
        \Statex{\textbf{Procedure:}}
        \Statex \hspace{-1.5em}\textit{Pre-processing:} 
        \State Define odd-periodic extension $\tilde{f}_{\text{odd}}$ of $f$ according to Eq.\ \eqref{eq:oddextension} and compute the sine Fourier coefficients $b_k=\frac{1}{2}\int_{-2}^{2}
\tilde{f}_{\text{odd}}(x)\sin\!\left(\frac{\pi k x}{2}\right)\! {\rm d}x$ of $\tilde{f}_{\text{odd}}(x)$ for $k\in \{1,2,...,K\}$, where the cutoff $K$ satisfies $
\left|\tilde{f}_{\text{odd}}(x)-\sum_{k=1}^{K}b_k\sin(\pi k x/2)\right|_{\infty}\le \epsilon/4t .
$
\State Compute the qDRIFT iteration number $N=N(\beta,t,\epsilon/2)$ according to Eq.\ \eqref{eq:qDRIFT} for $\beta=\sum_k|b_k|$.
        \Statex \hspace{-1.5em}
        \textit{Main Process:} 
        \State{Initialise} $U_{\rm current}\gets I$
        \For{$m=1,\ldots ,N$}
        \State randomly choose $k\in \{1,2...,K\}$ and the boolean $s\in \{0,1\}$ with probability $p_{k,s}=|b_k|/2\beta$
        \State Initialise $ e^{-iH^{(k,s)}_{\rm eff}} \gets I $
        \State \textcolor{wineRed}{(\textrm{I})} \; \; $ e^{-iH^{(k,s)}_{\rm eff}} \leftarrow e^{i \pi H[A] k/2} R_{\rm z}((-1)^{s} \beta t/N) e^{-i\pi H[A] k/2} $
        \State \textcolor{wineRed}{(\textrm{II})} \; $ e^{-iH^{(k,s)}_{\rm eff}} \leftarrow R_{\rm z}(\frac{\phi_k}{2}-\frac{\pi}{4}) e^{-iH^{(k,s)}_{\rm eff}} R^\dagger_{\rm z}(\frac{\phi_k}{2}-\frac{\pi}{4})$
        \State \textcolor{wineRed}{(\textrm{III})} \; $e^{-iH^{(k,s)}_{\rm eff}} \gets Z^{s}e^{-iH^{(k,s)}_{\rm eff}} Z^{s}$
        \State $e^{-iH^{(k,s)}_{\rm eff}}\gets e^{-iH^{(k,s)}_{\rm eff}} U_{\rm current} $
        \EndFor
        \State {\textbf{Return}} $U_{\rm current}$
    \end{algorithmic}
    \label{alg:SineFourierSeries}
\end{algorithm}
\newpage
\subsection{Deterministic UHSVT}
For completeness, we also construct a deterministic version of UHSVT. Just like the deterministic version of UHET, this implements the modulus of the Fourier coefficients with the angle of the central rotation. Therefore, in the circuit of Figure \ref{fig:sineFourierseries}, the central $Z$-rotation need to be replaced with:
\begin{equation}
    R_{\rm z}[(-1)^s|b_k|t/N] \,.
\end{equation}
For the first-order deterministic Trotter formula, the scaling matches that of the deterministic UHET protocol defined in Section~\ref{sec: Overview of higher-order quantum algorithms for Hamiltonian dynamics}. The number of Trotter steps required to reach accuracy $\epsilon$ is $N = \mathcal{O}(\beta^2 t^2/\epsilon)$ (see Eq.~\eqref{eq: trotter}). At each Trotter step, we implement $K$ frequencies, so the total time complexity scales as $K \times N$, giving:
\begin{equation}
\mathcal{O}\left(\frac{K\beta^2t^2}{\epsilon}\right) \,.
\end{equation}
Since the Fourier truncation error for odd-periodically $3$-smooth  $\tilde{f}_{\rm odd}$  satisfies $K = \mathcal{O}(\epsilon_K^{-1/2})$, $\beta \in \mathcal{O}(1)$ and $\epsilon_K\in \mathcal{O}(\epsilon/t)$, the total time complexity scales as $\mathcal{O}(t^{5/2}/\epsilon^{3/2})$.
For odd-periodically $2$-smooth $\tilde{f}_{\rm odd}$, a similar calculation yields $\mathcal{O}\!\left(t^{3}/\epsilon^{2}\right)$. This highlights the compilation advantage provided by the qDRIFT approach in the first-order Trotter setting. 

For a $2p$-th-order Suzuki formula, Eq.~\eqref{eq: Higher-Order Trotter-Suzuki Formula} gives a number of Trotter steps $N \in \mathcal{O}\!\left((\beta t)^{1+1/2p}/\epsilon^{1/2p}\right)$, while adding $\Gamma=2\times 5^{p-1}$ stages to the algorithm. Therefore, by an argument analogous to the 1st-order case, the time complexity scales as $ K \times N \times \Gamma$. 

From above, we have that $\beta \in \mathcal{O}(1)$ and $K\in \mathcal{O}[(\epsilon/t)^{-1/(J-1)}]$ for odd-periodically $J$-smooth functions, 
so that the total  time complexity is therefore:
\begin{equation}
    \mathcal{O}\!\left(2\times 5^{p-1} \frac{t^{1+1/(J-1)+1/2p}}{\epsilon^{1/(J-1)+1/2p}}\right).
\end{equation}
For instance, for $J=3$ and $p=2$, this yields $\mathcal{O}(t^{7/4}/\epsilon^{3/4})$, which is already  an improvement over the randomised 1st-order Trotter approach. In general, the smoother the function $\tilde{f}_{\rm odd}$, the better the deterministic approach performs.
These calculations are summarised in Table \ref{fig:comparaison randomised vs deterministic}.

\begin{table}[H]
    \centering
    \rowcolors{1}{white}{gray!15}
     \begin{tabular}{c|c c c c}
  
    Method  &Odd-period. $2$-smooth  &  Odd-period. $3$-smooth  &  Odd-period. $J$-smooth  & \\
  \hline
  Alg.\ \ref{alg:SineFourierSeries}/qDRIFT & $\mathcal{O}\left(\frac{t^2}{\epsilon}\right)$&$\mathcal{O}\left(\frac{t^2}{\epsilon}\right)$& $\mathcal{O}\left(\frac{t^2}{\epsilon}\right)$ &\\
  
  First-order Trotter & $\mathcal{O}\left(\frac{t^3}{\epsilon^2}\right)$&$\mathcal{O}\left(\frac{t^{5/2}}{\epsilon^{3/2}}\right)$& $\mathcal{O}\left(\frac{t^{2+1/(J-1)}}{\epsilon^{1+1/(J-1)}}\right)$& \\
  Second-order Trotter& $\mathcal{O}\left(\frac{t^{5/2}}{\epsilon^{3/2}}\right)$&$\mathcal{O}\left(\frac{t^{2}}{\epsilon}\right)$& $\mathcal{O}\!\left( \frac{t^{3/2+1/(J-1)}}{\epsilon^{1/(J-1)+1/2}}\right)$& \\
  Fourth-order Trotter & $\mathcal{O}\left(\frac{t^{9/4}}{\epsilon^{5/4}}\right)$&$\mathcal{O}\left(\frac{t^{7/4}}{\epsilon^{3/4}}\right)$&$\mathcal{O}\!\left( \frac{t^{5/4+1/(J-1)}}{\epsilon^{1/(J-1)+1/4}}\right)$& \\
\end{tabular}
    \caption{Total time complexity comparison between randomised (qDRIFT-based) UHSVT (Algorithm \ref{alg:SineFourierSeries}) and deterministic UHSVT with higher-order Trotter-Suzuki product formulas, for odd-periodically $J$-smooth functions $\tilde{f}_{\rm odd}$. 
    The smoother the function $\tilde{f}_{\rm odd}$, the better the deterministic approach performs. In particular, for $3$-smooth $\tilde{f}_{\rm odd}$ and above, the fourth-order product formula outperforms the qDRIFT-based UHSVT. 
    Note that for any function $f$ that is (at least) $J$-smooth (with $J\geq 2$) and satisfies
    $f^{(2m)}(0)=0$ for all $2m < J$, a corresponding odd extension $\tilde{f}_{\rm odd}$ can be constructed that is odd-periodically $J$-smooth 
    (see Lemma \ref{lem: J-smooth odd extension} in Appendix \ref{App_subsec: Extension of smooth functions}).}
    \label{fig:comparaison randomised vs deterministic}
\end{table}

Finally, the total evolution time is computed similarly to the time complexity, with the additional observation that implementing the unitary $e^{\pm i\pi k H[A]}$ requires an evolution time of $\mathcal{O}(k)$. Consequently, a single stage of the Trotter protocol contributes an evolution time $N\times\mathcal{O}\!\left(\sum_{k=1}^{K}k\right)\in\mathcal{O}(NK^2)$. Since the protocol consists of $\Gamma$ Trotter stages, the total evolution time scales as $\mathcal{O}(NK^2\Gamma)$.
Compared with the time complexity, this introduces an additional factor of $K$ in the overall scaling, since the average evolution time of the effective Hamiltonians grows linearly with the Fourier mode index $k$.

Deterministic first-order UHSVT can be summarised in the following theorem, proven in Appendix \ref{app:proof_det_UHSVT}:

\begin{theorem}[Universal Hamiltonian Singular Value Transformation (deterministic, canonical Hamiltonian block-encoding access)]
\label{thm: UHSVT}
Let $H[A]$ be the \textit{canonical Hamiltonian block-encoding} (Definition \ref{def:standardform}) of matrix $A$.
Then, for any 2-smooth function $f: [0,1]\rightarrow \mathbb{C} $ with $f(0)=0$, time $t>0$ and error $\epsilon>0$, there exists a  truncation number $K \in \mathbb{N}$ and an iteration number $N\in \mathbb{N}$, with $N \in \mathcal{O}\left(t^2/\epsilon\right), K \in \mathcal{O}\left(t/\epsilon\right)$, such that for any square matrix $A$ (normalised such that $||A||_{\rm op}\leq 1$)
the sequence
\begin{align}
\nonumber
    \mathtt{UHSVT}^{\rm det}_{\tilde f_{{\rm odd}, K},t,N}(H[A]) &:= \left[ \prod_{k=1}^K \prod_{s=0}^1
     Z^s R_{\rm z}^{\rm QS} \! \left( \! \frac{2\phi_k\! -\! \pi}{4} \!\right)
        e^{i  H[A] \frac{k \pi}{2}} 
  R_{\rm z}^{\rm QS}\! \left( \! (-1)^{s}\frac{ | b_k| t \! }{N}\right) 
      e^{-i H[A] \frac{k \pi}{2}}  
       R_{\rm z}^{\rm QS}\! \left(\! \frac{\pi \! -\! 2\phi_k}{4}\!  \right) Z^s
    \right]^N 
    \\ &\overset{\epsilon}{\approx} e^{-i  H[f^{\rm SV}(A)] t}
   \,,
\end{align}
where $\tilde f_{{\rm odd}, K}(x):=\sum_{k=0}^K b_k \sin \left(\pi kx/2 \right)$ is the truncated Fourier series of the odd periodic extension $\tilde f_{\rm odd}$ of $f$ defined in Eq.\ \eqref{eq:oddextension}, $ b_k = | b_k|e^{i  \phi}$, and  $R_{\rm z}^{\rm QS}(\theta):=  e^{-i\theta(2\Pi_{0}-I)} $ is a $Z$-rotation on the qubitised subspace. Here, $\overset{\epsilon}{\approx}$ denotes equality up to global phase in the operator norm, up to error $\epsilon$.
\end{theorem}

\subsection{Hamiltonian QSVT-based algorithm with lower bound on singular values}
\label{subsec: H-QSVT based algorithm using no X gate}
Recall that in Section \ref{subsec: HQSVT}, we showed a QSVT-based algorithm to perform the transformation:
\begin{equation}
    e^{-iH[A]}\mapsto e^{-iH[p^{\rm SV}(A)]}
\end{equation}
for a real definite-parity polynomial $p$ \cite{kang2025quantum}. For odd polynomials, this algorithm uses $\mathcal{O}(d\log(1/\epsilon))$ calls to $e^{-iH[A]}$ which is better than the $\mathcal{O}(1/\epsilon)$ dependency of UHSVT. However, as discussed in Section \ref{subsec: HQSVT}, the problem with such a method is that it relies on an $X$-gate in the qubitised subspace.
As discussed in Remark \ref{rem:Xrot}, such $X$-gates are not in general possible in a
black-box setting, where we do not have an explicit auxiliary qubit that encodes the qubitisation. In contrast, as shown in Section \ref{subsec: HQSVT}, $Z$-rotations on the qubitised subspace can be implemented in a black-box manner using only $C_{\Pi_0}$-NOT operations and an additional auxiliary qubit. 

An alternative approach, which does not require $X$-rotations on the qubitised subspace, is to assume a lower bound $\delta$ on the singular values of $A$, as done in \cite{lloyd2104hamiltonian}. Recall the required polynomial approximations from Section \ref{subsec: HQSVT}, Eq.\ \eqref{eq:pq_arcos}:
\begin{equation}
    P(x)\approx \cos(p(\arccos(x))) \quad \text{and} \quad Q^*(x)\approx \frac{\sin(p(\arccos(x)))}{\sqrt{1-x^2}} \,,   
\end{equation}
for $x=\cos(\Sigma)$, which were problematic due to arccos having singularities at $\pm 1$ and no definite parity. 

Firstly, regarding parity, as before, we can assume an upper bound of $A$ without loss of generality. Given that the singular values are non-negative, we pick $||A||\leq \pi/2$ so that $\cos \Sigma$ lies in $[0,1]$. Therefore, arccos will be defined only on the domain $[0,1]$, removing the parity issue.
If we wish to approximate a function on $[-1,1]$ by a polynomial with definite parity, a natural strategy is first to construct an approximation to $P(x) $ on $[0,1]$ and then extend it to $[-1,1]$ by symmetrisation.

Secondly, in the lower-bound approach, we additionally rescale $A$ with a constant $\gamma$ such that $\mathcal{O}(\delta^2)\leq \cos(\Sigma/\gamma)\leq 1-\mathcal{O}(\delta^2)$, restricting the domain of convergence away from the badly behaved boundaries. If we only ask our polynomial approximations to work well in this domain, then we can find such $P$ and $Q$ in the desired interval. In particular,  $x=\cos(\Sigma/\gamma)$ now lies in the interval  $[\delta^2,1-\delta^2]$, meaning that a smooth symmetrisation around the origin can be performed when constructing a definite-parity polynomial.

By adapting the proof of Hamiltonian QSVT in \cite{kang2025quantum}, which assumed the ability to perform $X$-rotations on the qubitised subspace, to the case where we instead assume a lower bound $\delta$ on the singular values of $A$, we obtain the following theorem:

\begin{theorem}[Hamiltonian QSVT with lower bound on singular values]
    Let $H[A]$ be a \textit{canonical Hamiltonian block-encoding} (Definition \ref{def:standardform}) 
    of a matrix $A$, with a lower bound $\delta \in \mathbb{R}^+$ on the singular values of $A$, such that $\delta \leq \|A\| \leq \pi/2$, and let $\epsilon \in \mathbb{R}^+$. 
    For any real polynomial $p \in \mathbb{R}[x]$ of degree $d$ such that $||p||_{[0,\pi/2]}\leq 1$, there  exists a constant $\gamma \in\mathbb{R^+}$ and a list  $\Theta = \{\theta_0, \theta_1, \ldots, \theta_{\tilde{d}}\}$ of $\tilde{d}$ angles such that:
    \begin{equation}
        e^{-i\theta_{\tilde{d}}(2\Pi_{0} - I)}
        \prod_{k=\tilde{d}-1}^{0}
        e^{-iH[A]/\gamma}e^{-i\theta_k(2\Pi_{0} - I)}
        \overset{\epsilon}{\approx} e^{-iH[p^{\rm SV}(A)]} ,
        \label{eq:HQSVT_circuit}
    \end{equation}
where $\overset{\epsilon}{\approx}$ denotes equality up to global phase in the operator norm up to error $\epsilon$, and where
\begin{equation}
        \tilde{d} \in
    \begin{cases}
        \mathcal{O}\!\left[
         \frac{d}{\delta^2}
        \log\!\left(\frac{1}{\delta\epsilon}\right)\log\!\left(\frac{d}{\delta\epsilon}\right)\right] \quad \quad \quad \quad  \text{for $p$ odd }\; 
        \\ 
            \mathcal{O}\!\left[
         \frac{d}{\delta^{2.5}}
        \log\!\left(\frac{1}{\delta\epsilon}\right)+\frac{1}{\delta^2}\log^2\!\left(\frac{1}{\delta\epsilon}\right)\right] \quad \text{for $p$ even or without a definite parity} . 
     \end{cases}
\end{equation}
    \label{thm: HQSVT with lower bound}
\end{theorem}

\noindent The proof of Theorem \ref{thm: HQSVT with lower bound} is deferred to Appendix~\ref{App_subsec: Hamiltonian QSVT in the odd case} for the odd case and \ref{App_subsec: Hamiltonian QSVT in the general case} for the general case, with the earlier parts of Appendix \ref{App: Proof for Hamiltonian QSVT} building the theory required for establishing these proofs.  Our proof follows similar techniques as the one in \cite{kang2025quantum}, except that we need to approximate $\arccos(x)$ instead of $\arcsin(x)$. The corresponding H-QSVT circuit is shown in Figure \ref{fig:HSVT}.

\begin{figure}[H]
\centering
    \begin{quantikz}[wire types={q,q,n,q}, classical gap=0.07cm]
        \qw & \gate{e^{i\theta_0Z}} & \gate[4,style={rounded corners}]{e^{-iH[A]/\gamma}} & \gate{e^{i\theta_1Z}} & \gate[4,style={rounded corners}]{e^{-iH[A]/\gamma}} & \gate{e^{i\theta_{2}Z}} & \cdots &  \gate[4,style={rounded corners}]{e^{-iH[A]/\gamma}}& \gate{e^{i\theta_{\tilde{d}} Z}}& \rstick[4]{\textcolor{wineRed}{$e^{-iH[p^{\rm SV}(A)]}$}} \\
    \qw & & & & & & \cdots & & &  \\
    &\lstick{\vdots}&&\lstick{\vdots}&&\lstick{\vdots}&&&\lstick{\vdots}&&\\
    \qw & & & & & & \cdots & & & 
    \end{quantikz}
\caption{Hamiltonian QSVT for a matrix $A$ block-encoded into the Hamiltonian of a unitary dynamics $e^{-iH[A]/\gamma}$, assuming a lower bound $\delta$ on the singular values of $A$. The result is a new block-encoded matrix $p^{\rm SV}(A)$, with singular values equal to those of $A$ transformed by polynomial $p$. Note that, for simplicity, the $Z$-rotation on the qubitised subspace (see Eq.\ \eqref{eq: zrotationsingularsubspace}) is schematically shown as a $Z$-rotation on a wire, instead of its physical implementation using an extra auxiliary qubit (see Figure \ref{fig:bob}). 
} 
\label{fig:HSVT}
\end{figure}
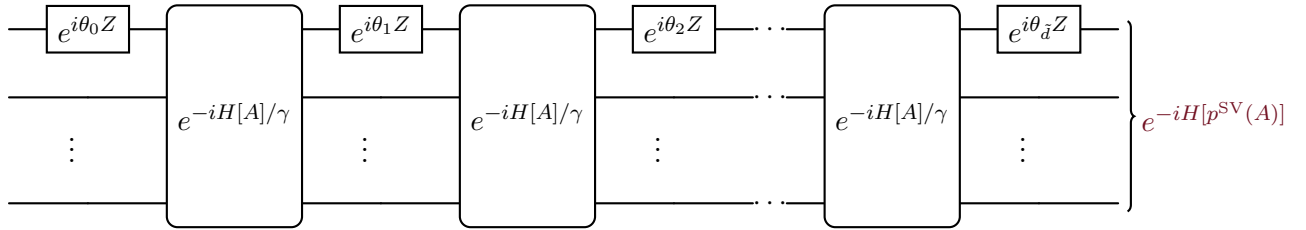

Now that we can implement a Hamiltonian block-encoding of $p^{\rm SV}(A)$, we wish to lift this to any sufficiently differentiable function $f^{\rm SV}(A)$ just like in UHSVT. Fortunately, for this task, it suffices to find a polynomial approximation $p(x)\approx f(x)$. Although there is no universal method for an optimal polynomial 
approximation, Chebyshev polynomials provide a natural basis, 
closely related to Fourier expansions. Just like in the Fourier series case, for $J$-smooth functions, 
the degree of the Chebyshev approximation typically scales as $d\in\mathcal{O}\!\left(\epsilon^{-1/(J-1)}\right)$,  while for analytic functions (see Definition \ref{def: analytic functions}), $d\in \mathcal{O}(\log(1/\epsilon))$ (see Appendix \ref{App_subsec: Approximating functions with Chebyshev polynomials}).Combining these facts with Theorem \ref{thm: HQSVT with lower bound} we obtain Algorithm \ref{alg:QSVT-based algorithm}, the proof of which is given in Appendix \ref{App:proof_alg2}, and the corresponding Theorem, whose proof is equivalent:

\begin{algorithm}[H]
    \floatname{algorithm}{Algorithm}
    \caption{Hamiltonian QSVT-based algorithm with lower bound on singular values}
    \begin{algorithmic}[1]
        \Statex{\textbf{Input:}}        
        \begin{itemize}
            \item A finite number of queries to a black-box Hamiltonian dynamics $e^{\pm iH[A]\tau}\in \mathcal{L}(\mathcal{H}_0\oplus \mathcal{H}_1)$ with $\tau \geq 0$ of a seed Hamiltonian $H[A]$, which is a Hamiltonian block-encoding in the canonical form (see Def. \ref{def:standardform}) of a matrix $A \in \mathcal{L}(\mathcal{H}_0 \to \mathcal{H}_1)$, whose singular values $\{\zeta_m\}_m$ are bounded by $\delta \leq \zeta_m \leq \pi/2$.
            \item A $J$-smooth ($J\geq 2$) or analytic function $f:[0,1]\rightarrow \mathbb{R}$  such that $||f||_{\infty}\leq 1$.
            \item Allowed error $\epsilon >0$.
        \end{itemize}
        \Statex{\textbf{Output:}}
        A unitary $U_{\rm approx}$ approximating  $e^{-iH[f^{\rm SV}(A)]}$, corresponding to a canonical Hamiltonian block-encoding of $f^{\rm SV}(A)$,
        with an error according to the {operator norm} upper bounded by $\epsilon$. 
    \Statex \hrulefill
            \Statex{\textbf{Time complexity:}}
             $\mathcal{O}\!\left[
         \frac{1}{\delta^{5/2}\epsilon^{1/(J-1)}}
        \log\!\left(\frac{1}{\delta\epsilon}\right)\right]$ for $J$-smooth functions and $\mathcal{O}\!\left[
         \frac{1}{\delta^{5/2}}
        \log\!\left(\frac{1}{\delta\epsilon}\right)\log\!\left(\frac{1}{\epsilon}\right)\right]$ for analytic functions.
    	\Statex{\textbf{Used Resources:}}
        	\Statex \hskip 1.0em System: Hilbert space $\mathcal{H}=\mathcal{H}_0 \oplus \mathcal{H}_1$, and an additional auxiliary qubit register $\mathcal{H}_{\rm aux}$.
    	\Statex \hskip1.0em Gates: $e^{-iH[A]\tau}$ on $\mathcal{H}_0 \oplus \mathcal{H}_1$, $C_{\Pi_0}$-NOT gates on $(\mathcal{H}_0\oplus \mathcal{H}_1)\otimes \mathcal{H}_{\rm aux}$ (see Fig. \ref{fig:bob}) and $R_{\rm z}(\theta)$ on $ \mathcal{H}_{\rm aux}$. 
    	\Statex \hrulefill
        \Statex{\textbf{Procedure:}}
        \Statex \hspace{-1.5em}\textit{Pre-processing:} 
        \State Compute the Chebychev polynomial coefficients $a_k=\frac{1}{2}\int_{-1}^{1}\frac{f(x)T_k(x)}{\sqrt{1-x^2}}
{\rm d}x$ of $f(x)$ for $k\in \{0,1,2,...,K\}$ where the cutoff $K$ satisfies $
\left\|f(x)-\sum_{k=0}^{K}a_kT_k(x)\right\|\leq \epsilon
$. Define $p(x)=\frac{1}{1+\epsilon}\sum_{k=0}^Ka_kT_k(x)$ (where $1/(1+\epsilon)$ is for normalisation).
\State Compute $\gamma$ so that $\cos(A/\gamma)\in [\delta^2,1-\delta^2]$.
\State Construct $\tilde{p}(x)$ and $\tilde{q}(x)$ such that:
\begin{equation*}
    \left\|\tilde{p}(x)-\cos(p(\arccos(x)))\right\|_{[\delta,1]}\leq \epsilon,\qquad \left\|\tilde{q}(x)-\frac{\sin(p(\arccos(x)))}{\sqrt{1-x^2}}\right\|_{[\delta,1]}\leq \epsilon \; . 
\end{equation*}
Define $\tilde{d}$ to be the degree of $\tilde{p}$.
\State Compute $P,Q$ such that $\tilde{p}=\mathfrak{Re}[P]$ and $\tilde{q}=\mathfrak{Im}[Q]$, using the constructive algorithm of \cite{gilyen2019quantum} derived from their proof of Theorem \ref{thm: ChebQSPrelaxed}.
\State Compute the Chebychev QSP processing angles $\tilde{\Theta}=\{\theta_0,\theta_1,...,\theta_{\tilde{d}}\}$ associated to $P$ and $Q$ (see \cite{skelton2025hitchhiker}).

        \Statex \hspace{-1.5em}
        \textit{Main Process:} 
        \State {\textbf{Return}} $$U_{\rm approx}= \left(e^{i\theta_{\tilde{d}}(2\Pi_{0} - I)}
        \prod_{k=\tilde{d}-1}^{0}
        e^{-iH[A]/\gamma}e^{i\theta_k(2\Pi_{0} - I)}\right) \; .$$ 
    \end{algorithmic}
    \label{alg:QSVT-based algorithm}
\end{algorithm}

\newpage

\begin{theorem}
\label{thm: lower bound QSVT based algorithm for functions approximation}
Let $H[A]$ be the \textit{canonical Hamiltonian block-encoding} (Definition \ref{def:standardform}) of matrix $A$.
Then, for any $J$-smooth ($J\geq 2$) or analytic function $f: [0,1]\rightarrow \mathbb{R} $ , time $t>0$ and error $\epsilon>0$, there exists a list of $\tilde{d}\in \mathbb{N}$ angles $\{\phi_0,\phi_1,..., \phi_{\tilde{d}}\}$ and a real number $\gamma$, such that for any square matrix $A$ (normalised such that $\delta\leq||A||_{\rm op}\leq \pi/2$)
the sequence
\begin{align}
\nonumber
     e^{-i\theta_{\tilde{d}}(2\Pi_{0}-I)}
        \prod_{k=\tilde{d}-1}^{0}
        e^{-iH[A]/\gamma}e^{-i\theta_k(2\Pi_{0}-I)}\overset{\epsilon}{\approx} e^{-iH[f_{\rm SV}(A)]}\; , 
\end{align}
where $\overset{\epsilon}{\approx}$ denotes equality up to a global phase in the operator norm, with error $\epsilon$. Moreover, for arbitrary $f$ with no assumption on parity, $\tilde{d}$ scales like
\begin{equation}
    \tilde{d}\in\begin{cases}
       \mathcal{O}\!\left[
         \frac{1}{\delta^{5/2}}
        \log\!\left(\frac{1}{\delta\epsilon}\right)\log\left(\frac{1}{\epsilon}\right)\right]   &\text{for analytic  $f$}\; 
        \\ 
            \mathcal{O}\!\left[
         \frac{1}{\delta^{5/2}\epsilon^{1/(J-1)}}
        \log\!\left(\frac{1}{\delta\epsilon}\right)\right] \quad \quad &\text{for $J$-smooth $f$} . 
     \end{cases}
\end{equation}
which is equivalent to the total time complexity.
\end{theorem}

\subsection{Comparison of the methods}

We finish by comparing our UHSVT algorithm (Alg.\ \ref{alg:SineFourierSeries}) with the two Hamiltonian QSVT-based algorithms we have discussed for the same task. The first is the one presented in \cite{kang2025quantum} and Section \ref{sec: Overview of QSP-based functional transformations}, which relies on an $X$-rotation operator in the qubitised subspace. 
The second is the one we developed in the previous subsection \ref{subsec: HQSVT}  and instead requires a lower bound $\delta$ on the singular values (Alg.\ \ref{alg:QSVT-based algorithm}), which was inspired by the setting in \cite{lloyd2104hamiltonian}. The overall time complexity of the three algorithms is given in Table \ref{fig:complexitycomparaison}.

\begin{table}[H]
    \centering
    \rowcolors{1}{white}{gray!15}
     \begin{tabular}{c| c c|c }
  
    Method   & $J$-smooth  ($J\geq2$)& Analytic & Assumptions \\
  \hline
  Alg.\ \ref{alg:SineFourierSeries} &$\mathcal{O}\left(\frac{1}{\epsilon}\right)$&$\mathcal{O}\left(\frac{1}{\epsilon}\right)$&  
  $f(0)=0$
  \\

  Alg. \ref{alg:QSVT-based algorithm}  & $\mathcal{O}\left(\frac{1}{\delta^{5/2} \epsilon^{1/J-1}}\log\left(\frac{1}{\delta\epsilon}\right)\right)$ &  $\mathcal{O}\left(\frac{1}{\delta^{5/2}}\log(\frac{1}{\delta \epsilon})\log(\frac{1}{\epsilon})\right)$& $A$'s singular values $\geq \delta$
  \\
  \cite{kang2025quantum}  & $ \mathcal{O}\left(\frac{1}{\epsilon^{1/J-1}}\log\left(\frac{1}{\epsilon}\right)\right)$ & $\mathcal{O}\left(\log^2(\frac{1}{\epsilon})\right)$&$X$-gates on qubitised subspace  \\

\end{tabular}
    \caption{Time complexity comparison between two QSVT-based algorithms and our UHSVT algorithm for implementing a functional transformation $f$ of the singular values of a matrix $A$ canonically block-encoded into a Hamiltonian evolution $e^{-iH[A]}$ (taking time $t=1$). The latter two algorithms are based on (Hamiltonian) QSVT and either rely on the assumption that the singular values are
    lower-bounded by some constant $\delta>0$, or assume the ability to perform $X$-gates on the qubitised subspace. In contrast, UHSVT requires only a condition on the function $f$, namely that it vanishes at the origin.
    }
    \label{fig:complexitycomparaison}
\end{table}

In summary, when restricting to $Z$-rotations on the qubitised subspace (that is, in the usual QSVT setting where one only assumes black-box access to the dynamics $e^{-iH[A]}$ and the ability to perform controlled-subspace operators $C_{\Pi}$-NOT) our UHSVT method is better suited for transforming the singular values of a matrix $A$ block-encoded inside a Hamiltonian, since it does not require a lower bound on its singular values. However, when we have a known lower bound $\delta$ on the singular values of $A$, the Hamiltonian QSVT of Algorithm \ref{alg:QSVT-based algorithm} achieves performance comparable to UHSVT in terms of $\epsilon$ for functions with a small degree of smoothness and superior for functions with a large degree of smoothness, but with the drawback of having a $\mathcal{O}(1/\delta^{5/2})$ dependency (or $1/\delta^2$ if $f$ is odd). 
It may appear surprising that UHSVT outperforms QSVT-based methods in certain regimes, particularly with respect to the lower bound parameter $\delta$. This arises because the branch cut inherent to the Hamiltonian block-encoding requires the implementation of $\arccos$, which is poorly behaved near these branch cuts. UHSVT circumvents this issue by operating directly at the level of Hamiltonians, where no such singularity arises.
We also notice that a lower bound $\delta$ would allow us to extend the applicability of UHSVT to functions without the $f(0)=0$ constraint. Indeed, we could construct a polynomial interpolation on $[0,\delta]$ that matches the function and its derivatives at $x=\delta$ while satisfying the appropriate conditions at $x=0$. After taking the odd extension, this allows us to construct a periodically $2$-smooth function on the larger interval $[-2,2]$.

In cases where we do have access to $X$-rotations of the qubitised subspace, the QSVT-based algorithm of \cite{kang2025quantum} has superior performance to our UHSVT algorithm as presented above for the relevant classes of functions, assuming the block-encoding Hamiltonian is given in canonical form. However, we note that if given access to $X$-rotations, just like in the case where we have access to a lower bound on the singular values,we could also improve our method to include functions which do not obey the condition $f(0)=0$. Applying the rotation $iX$, has the effect of shifting the spectrum of $A$ by a factor of $1$ (from $[0,1]$ to  $[1,2]$), and then we could construct a new polynomial interpolation on the region $[0,1]$.

Another advantage of our method is that UHSVT is the only approach capable of implementing a complex-valued function of the singular values of $A$. Constructing complex functions via QSP-based methods appears considerably more involved, as one must implement the following unitary:
\begin{align}
    \exp\left(-i\begin{bmatrix}
        0 & f^\dagger(\Sigma)\\
        f(\Sigma) &  0
    \end{bmatrix}\right)
    &=
    \bigoplus_m\exp\left(-i\bigl(\mathfrak{Re}[f(\xi_m)]\,X+\mathfrak{Im}[f(\xi_m)]\,Y\bigr)\right)\\
    &=\bigoplus_m \begin{bmatrix}
        \cos(\theta_m) & -i\sin(\theta_m)e^{-i\phi_m}\\
        -i\sin(\theta_m)e^{i\phi_m} & \cos(\theta_m)
    \end{bmatrix},
\end{align}
where $\theta_m = \|f(\xi_m)\|_2$ and $\phi_m = \arccos\left(\frac{\mathfrak{Re}(f(\xi_m))}{||f(\xi_m)||_2}\right)$. Applying QSP in this context thus reduces to finding polynomial approximations to $\cos(\theta_m)$ and $\sin(\theta_m)e^{i\phi_m}$ simultaneously, which appears challenging given the parity and normalisation constraints inherent to ($Z$-constrained) QSP.

\subsection{Application to inverse block-encoding}
The most straightforward application of UHSVT is to the problem of inverse block-encoding, which was the original motivation of \cite{lloyd2104hamiltonian} for developing Hamiltonian QSVT. Given a Hamiltonian block-encoding of a matrix $A$ in canonical form, we seek a unitary block-encoding  
\begin{equation}
  U[A/\alpha] := \begin{bmatrix}
        \cdot & \cdot \\
         A/\alpha & \cdot
    \end{bmatrix}
\end{equation}
of the same matrix (where $\alpha\in \mathbb{R}^{>0}$), that is we want an algorithm to perform
\begin{equation}
    \exp\left(-i\begin{bmatrix}
        0 & A^\dagger \\
        A & 0
    \end{bmatrix}\right)\longmapsto \begin{bmatrix}
        \cdot & \cdot \\
         A/\alpha & \cdot
    \end{bmatrix}.
\end{equation}
Applying UHSVT with a function $f:[0,1]\rightarrow i\mathbb{R}$ on $A$ yields:
\begin{equation}
    \exp\left(-i\begin{bmatrix}
        0 & -if^{\rm SV}(A^\dagger) \\
        if^{\rm SV}(A) & 0
    \end{bmatrix}\right)=\begin{bmatrix}
        \cos(f^{\rm SV}(A)) & -\sin({f^{\rm SV}(A^\dagger)}) \\ 
          \sin({f^{\rm SV}(A)}) & \cos({f^{\rm SV}(A^\dagger)})
    \end{bmatrix}.
\end{equation}
Taking $f^{\rm SV}(A)=\arcsin(A)$ then yields:
\begin{equation}
    \begin{bmatrix}
        \sqrt{1-A^\dagger A} & -A^\dagger \\ 
          A & \sqrt{1-A A^\dagger}
    \end{bmatrix},
\end{equation}
where we used the identity $\cos(\arcsin(x)) = \sqrt{1-x^2}$. Since $\arcsin$ is an odd function with vanishing even-order derivatives, it is compatible with UHSVT without additional constraints. However, $\arcsin$ has singularities at $\pm 1$, so $A$ must be rescaled so that its singular values lie in $[-1/2, 1/2]$, leading to a block-encoding of $A/2$ rather than $A$. Moreover, since $\arcsin$ is not periodic, it must be periodised over a larger interval before applying the Fourier (Appendix \ref{App_subsec: Extension of smooth functions}). Since $\arcsin$ is $C^\infty$, the deterministic UHSVT approach is preferable, and one may choose to match arbitrarily many derivatives in the polynomial interpolation (at the cost of inverting a $\mathcal{O}(J)$ linear system, see Appendix \ref{App_subsec: Extension of smooth functions}). Matching $J$ derivatives yields a $J$-smooth periodic extension, leading to an algorithm for inverse block-encoding with time complexity:
\begin{equation}
    \mathcal{O}\left(\frac{1}{\epsilon^{1/2+1/(J-1)}}\right) \; ,
\end{equation}
for desired accuracy $\epsilon$ (taking the second-order Suzuki version of deterministic UHSVT). 

Of course, a similar algorithm could be implemented to compute $A/\alpha$ for arbitrary $\alpha > 1$ rather than $A/2$. For this task, one only has to rescale by $1/\alpha$ and then perform the polynomial interpolation routine for $[1/\alpha,1]$. Chebyshev QSP can also be used to perform the same task, as originally introduced in \cite{lloyd2104hamiltonian}, but as before requires the assumption of a lower bound $\delta$ on the singular values of $A$, which introduces an inverse polynomial dependence in $1/\delta$ in the scaling of the algorithm. 

\section{Summary and outlook}
\label{sec:conclusion}

In this work, we presented a step towards a fully Hamiltonian-oriented framework for quantum functional programming. By identifying the Universal Hamiltonian Eigenvalue Transformation (UHET) algorithm \cite{odake2025universal} as a (randomised) linearisation of Generalised Quantum Signal Processing (GQSP) \cite{motlagh2024generalized}, we established a concrete connection between the two recent paradigms of quantum functional programming, namely higher-order quantum computation \cite{chiribella2008transforming,chiribella2008quantum,quintino2019probabilistic,odake2024higher} and the QSP/QSVT framework \cite{low2017optimal,gilyen2019quantum,martyn2021grand}. 
Building on this equivalence, we introduced a new algorithm which can be constructed as a (randomised) linearisation of QSVT.
This algorithm, which we call Universal Hamiltonian Singular Value Transformation (UHSVT), can transform the singular values of any unknown matrix $A$ encoded in a block of a Hamiltonian by any sufficiently differentiable function $f$ (vanishing at the origin), given access only to the Hamiltonian dynamics encoding $A$. 

Our algorithm runs with time complexity $O( t^2/\epsilon)$ elementary gates and total evolution time \\ $O( t^2 \epsilon^{-1} \log ( t \epsilon^{-1}))$ of the original Hamiltonian, for desired error $\epsilon$ and output evolution time $t$.  We showed how this algorithm can be applied in  scenarios beyond those of prior QSVT-based works \cite{lloyd2104hamiltonian, kang2025quantum}, namely without assuming either a lower bound on the singular values of $A$ or the ability to perform $X$-rotations on the qubitised subspace, thus achieving scaling under minimal assumptions at the state-of-the-art.
We also showed an application to the problem of inverse block-encoding. As a byproduct, we leveraged the mathematical tools developed in \cite{kang2025quantum} to rigorously construct a QSVT-based method in the spirit of \cite{lloyd2104hamiltonian} that does not rely on applying $X$-rotations between on the qubitised subspace.

Overall, we have presented a ``map'' connecting QSP-based and higher-order transformations-based algorithms for Hamiltonian dynamics, via linearisation in a way analogous to that in which the Fourier series is a linearisation of the non-linear Fourier transform. 

Our work opens up the study of a general framework for linearisation of QSP-based protocols. For example, new higher-order algorithms could be derived as linearisations of other QSP-based protocols, while new QSP-based algorithms could be constructed from non-linear versions of higher-order algorithms in more general settings. One may also ask whether linearised versions exist for well-known quantum routines that fall outside the QSP framework, such as the quantum Fourier transform, linear combination of unitaries \cite{childs2012lcu}, or Grover's search algorithm, and if so, in what regimes such linearisations prove useful.

This work also leads to several avenues for future research on Hamiltonian block-encoding-based algorithms. For instance, the analysis of UHSVT assumes access to the evolution of a Hamiltonian block-encoding in canonical form, whereas access to a general Hamiltonian block-encoding would first require a diagonal suppression subroutine to obtain the canonical form before applying UHSVT. This subroutine would be analogous to the controllisation procedure of \cite{odake2025universal}, which is concatenated with the UHET Fourier series simulation subroutine, and further \textit{compiled} using correlated randomness between the two subroutines to reduce the overall time complexity. It would be interesting to see if a similar compilation procedure can be developed for our setting, which exploits correlated randomness between the diagonal suppression subroutine and UHSVT sine simulation subroutine to reduce the overall error.
Finally the literature on product formulas, including Trotter--Suzuki and qDRIFT, is extensive, leaving considerable room for optimising our algorithms. Potential directions include hybrid deterministic--randomised Trotter--Suzuki techniques \cite{Hagan2023compositequantum, ynxb-p2xq}, tighter product-formula error estimates \cite{PhysRevX.11.011020}, and the use of Richardson extrapolation to improve the scaling of observable measurements \cite{watson2025randomlycompiledquantumsimulation}.

\subsection*{Acknowledgements}

We are happy to thank Mio Murao and Yuan Su for interesting discussions. We acknowledge the support of the Natural Sciences and Engineering Research Council of Canada (NSERC), grant number RGPIN-2025-04419, and support from the Institut Courtois, Faculté des arts et des sciences, Université de Montréal (Chaire Courtois de l’Institut Courtois).

\printbibliography[heading=bibintoc]
\newpage

\appendix
\renewcommand{\thesection}{\Alph{section}}
\renewcommand{\thesubsection}{\Alph{section}.\arabic{subsection}}
\renewcommand{\theequation}{\Alph{section}.\arabic{equation}}

\setcounter{equation}{0}
\section*{Appendices}
\addcontentsline{toc}{section}{Appendices}

\section{Theory of Fourier approximation}
\label{App: Theory of Fourier approximation}
In this appendix, we introduce the theory of function approximation via
Fourier series, which forms the mathematical core underlying the
UHSVT and UHET frameworks,  proving various results referenced in the main text, as well as the lemmas required for the proofs of the main theorems and algorithms in later appendices.

\subsection{Fourier series}
\label{App_subsec: Fourier series}
Let $f:[-L,L]\rightarrow\mathbb{C}$ be a complex-valued function.
Its complex Fourier series expansion is given by
\begin{equation}
    f(x)=\sum_{k=-\infty}^{\infty}
    c_k e^{-i\pi kx/L},
\end{equation}
where the Fourier coefficients are obtained from the orthogonality of
the functions $e^{-i\pi kx/L}$ on $[-L,L]$,
\begin{equation}
    c_k=\frac{1}{2L}
    \int_{-L}^{L}
    f(x)e^{-i\pi kx/L}\,{\rm d}x .
    \label{eq:FourierCoefficients}
\end{equation}

\vspace{0.5em}

In practice, functions are approximated by truncating the Fourier
series at order $K$. The truncation error satisfies
\begin{equation}
    \epsilon_K
    =
    \left\|
    f(x)-
    \sum_{k=-K}^{K}
    c_k e^{-i\pi kx/L}
    \right\|_{\infty}
    \le
    \sum_{|k|>K}|c_k|,
    \label{eq:errortruncfourier}
\end{equation}
where $|g(x)|_{\infty,x\in[a,b]}$ is the maximum value taken by $g$ inside the domain $[a,b]$. This shows that the approximation accuracy is directly governed by the
decay rate of the Fourier coefficients.

Because the Fourier series approximates functions using smooth periodic
basis functions, sufficiently rapid convergence is obtained only for functions
possessing sufficient smoothness and compatibility with periodic
boundary conditions. The following definitions formalise these notions.

\begin{definition}[(Piecewise) $C^{J}$ functions \cite{odake2025universal}]
A function $f:[-L,L]\rightarrow\mathbb{C}$ is said to be
\emph{piecewise $C^{J}$} if there exists a finite set of points $ \mathcal{P} = \{-L=x_0 < x_1 < \cdots < x_m=L\}$ such that $f$ is $J$ times continuously differentiable on each open interval $(x_j,x_{j+1})$. Moreover, the one-sided limits $\lim_{x\to x_j^-}f(x)$ and $\lim_{x\to x_j^+}f(x)$ exist, although they need not coincide. If $ \mathcal{P} = \{-L,L\}$, $f$ is said to be $C^J$.
\end{definition}

\begin{definition}[$J$-smooth functions \cite{odake2025universal}]
A function $f:[-L,L]\rightarrow\mathbb{C}$ is 
\emph{$J$-smooth} if $f$ is $C^{J-1}$ and its $J$-th derivative $f^{(J)}$
is piecewise $C^{2}$. If discontinuities occur in the $J$-th derivative, that is, $\lim_{x\to x_j^-}f^{(J)}(x) \neq\lim_{x\to x_j^+}f^{(J)}(x)$
for at least one point $x_j$, the function is said to be \emph{strictly $J$-smooth}. A $J$-smooth function $f$ is said to be \emph{periodically $J$-smooth} if all derivatives up to order $J$ satisfy periodic boundary conditions, $f^{(j)}(L)=f^{(j)}(-L)$ for all $ 0\le j\le J$. Finally, a function is $C^{\infty}$ (or simply \emph{smooth}) if all its derivatives exist and are continuous. 
\label{def: J-smooth}
\end{definition}

For functions in the above definitions, the error of the truncation grows polynomially with the smoothness of the function, as we see in the next Lemma:

\begin{lemma}[Decay of Fourier coefficients (\cite{odake2025universal})]
Let $f:[-L,L]\rightarrow\mathbb{C}$ be a periodically
$J$-smooth function.
Then its Fourier coefficients
$c_k$ defined in Eq.~\eqref{eq:FourierCoefficients}
satisfy
\begin{equation}
    |c_k| \in \mathcal{O}(|k|^{-J}),
\end{equation}
and the truncation error (\ref{eq:errortruncfourier}) obeys
\begin{equation}
    \epsilon_K \in \mathcal{O}(K^{-(J-1)}) \,.
\end{equation}

\label{th:decaypropertiesfourierseries}
\end{lemma}

\begin{proof}
Starting from the definition of the Fourier coefficients,
\begin{equation}
    c_k=\frac{1}{2L}
    \int_{-L}^{L}
    f(x)e^{-i\pi kx/L}\,dx ,
\end{equation}
we integrate by parts using
\[
u=f(x),
\qquad
dv=e^{-i\pi kx/L}dx .
\]
Since
\[
v=\frac{-L}{i\pi k}e^{-i\pi kx/L},
\]
we obtain
\begin{align}
c_k
&=
\frac{1}{2L}
\left[
\frac{-L}{i\pi k}
f(x)e^{-i\pi kx/L}
\right]_{-L}^{L}
+
\frac{1}{2L}
\frac{L}{i\pi k}
\int_{-L}^{L}
f'(x)e^{-i\pi kx/L}dx .
\end{align}

Because $f$ is periodically $J$-smooth,
$f(L)=f(-L)$, and therefore the boundary term vanishes.
Repeating this integration-by-parts procedure $J$ times yields
\begin{equation}
    c_k=
    \frac{1}{2L}
    \left(\frac{-L}{i\pi k}\right)^J
    \int_{-L}^{L}
    f^{(J)}(x)e^{-i\pi kx/L}dx .
\end{equation}
Taking absolute values and applying the triangle inequality gives
\begin{equation}
\label{eq:c_k_ineq}
    |c_k|
    \le
    \left(\frac{L}{\pi |k|}\right)^J
    \frac{1}{2L}
    \int_{-L}^{L}
    |f^{(J)}(x)|dx .
\end{equation}
Since $f^{(J)}$ is piecewise $C^2$, it is bounded and integrable,
which implies
\[
|c_k| \in \mathcal{O}(|k|^{-J}).
\]
Finally, using Eq.~\eqref{eq:errortruncfourier},
\begin{equation}
\epsilon_K
\le
\sum_{|k|>K}|c_k|
\in
\mathcal{O}\!\left(
\sum_{k=K}^{\infty}\frac{1}{|k|^{J}}
\right).
\end{equation}
Bounding the series by an integral gives
\begin{equation}
\epsilon_K
\in
\mathcal{O}
\left(
\int_{K}^{\infty}\frac{dk}{k^{J}}
\right)
=
\mathcal{O}(K^{-(J-1)}),
\end{equation}
which concludes the proof for $J<\infty$. 
 \end{proof}

\subsection{Sine series}
\label{App_subsec: Sine series}
In the case of Algorithm~\ref{alg:SineFourierSeries}, the functional
approximation is restricted to Fourier expansions involving only
sine functions,
\begin{equation}
    f(x)=
    \sum_{k=1}^{\infty}
    b_k
    \sin\!\left(\frac{\pi k x}{L}\right) \,,
\end{equation}
which is a special case of the sine--cosine form of the Fourier series:
\begin{equation}
    f(x)=
    \sum_{k=0}^{\infty}
    a_k\cos\!\left(\frac{\pi kx}{L}\right)
    +
    b_k\sin\!\left(\frac{\pi kx}{L}\right) \,,
\end{equation}
which itself can be retrieved from Eq.\ \eqref{eq:FourierSeriesExpForm} by grouping terms
indexed by $k$ and $-k$ and using Euler's formula, 
\begin{align}
    f(x)
    &=\sum_{k=0}^{\infty}
    \left(
    c_k e^{-i\pi kx/L}
    +c_{-k} e^{i\pi kx/L}
    \right) \nonumber\\
    &=\sum_{k=0}^{\infty}
    (c_k+c_{-k})\cos\!\left(\frac{\pi kx}{L}\right)
    -i(c_k-c_{-k})
    \sin\!\left(\frac{\pi kx}{L}\right).
\end{align}
Defining $a_k = c_k+c_{-k}$ and  $b_k = -i(c_k-c_{-k})$ yields the sine--cosine representation.
Since the sine and cosine functions also form orthogonal families on
$[-L,L]$, the coefficients $a_k$ and $b_k$ can equivalently be computed
through projection integrals:
\begin{equation}
a_k=\frac{1}{2L}\int_{-L}^{L}
f(x)\cos\!\left(\frac{\pi k x}{L}\right)\!dx,
\qquad
b_k=\frac{1}{2L}\int_{-L}^{L}
f(x)\sin\!\left(\frac{\pi k x}{L}\right)\!dx \,.
\label{eq:SineCosineFourierSeries}
\end{equation}

Since $\sin(\pi kx/L)$ is both periodic and odd, the class of
functions that can be efficiently approximated using such an expansion
is restricted to functions that are themselves odd and periodic on
the interval $[-L,L]$. Recall that a function is said to be \emph{odd} if $f(-x)=-f(x)$.
Combining oddness with periodicity implies $f(L)=f(-L)=-f(L)$, which immediately yields $f(L)=0$. Similarly, oddness also enforces $f(0)=0$. Consequently, functions admitting an efficient sine-series representation must vanish both at the origin and at the end of the domain.

Moreover, odd functions have no cosine contribution in their
sine--cosine decomposition. Indeed, recall that the cosine
coefficients are given by Eq.\ \eqref{eq:SineCosineFourierSeries}. Since $\cos(\pi kx/L)$ is an even function on $[-L,L]$,
the product $f(x)\cos(\pi kx/L)$ is odd whenever $f$ is odd.
Because the integral of an odd function over a symmetric interval
vanishes, we obtain $a_k=0$ for all  $k$ . Therefore, odd functions admit a Fourier expansion purely in terms of
sine functions. 
The approximation error related to a degree $K$ truncation of a sine series is thus:
\begin{equation}
    \epsilon_K=\left\|f(x)-\sum_{k=1}^{K}b_k \sin \left( \frac{\pi k x}{L}\right) \right\|_{[0,L]}\leq \sum_{k=K+1}^{\infty} \left\|b_k\sin \left( \frac{\pi k x}{L}\right)\right\|_{[0,L]}\leq \sum_{k=K+1}^{\infty} \left|b_k\right|
\end{equation}

We now seek a class of odd functions that enjoys decay properties
analogous to periodically $J$-smooth functions
(see Lemma~\ref{th:decaypropertiesfourierseries}),
but tailored to a sine-only decomposition. 

\begin{definition}[Odd-periodically $J$-smooth functions]
\label{def: odd-periodically J-smooth}
A function $f:[-L,L]\mapsto \mathbb{C}$ is said to be \emph{odd-periodically $J$-smooth} if it is $J$-smooth on $[-L,0)\cup (0,L]$ (but not necessarily at $0$) and all the even derivatives up to order $J$ vanish at the origin and at the endpoints:
\begin{equation}
     f^{(2m)}(0)=f^{(2m)}(-L) = f^{(2m)}(L) = 0, \qquad 0 \le 2m < J.
\end{equation}
    
\end{definition}

\begin{lemma}[Decay of sine Fourier coefficients of odd functions]
Let $f:[-L,L]\rightarrow\mathbb{C}$ be an odd-periodically $J$-smooth  function on $[-L,L]$. Then, the sine Fourier coefficients
$b_k$ satisfy
\begin{equation}
    |b_k| \in \mathcal{O}(|k|^{-J}) \,,
\end{equation}
and the truncation error obeys
\begin{equation}
    \epsilon_K \in \mathcal{O}(K^{-(J-1)}) \,.
\end{equation}

\label{th:decaypropertiessinefourierseries}
\end{lemma}

Note that here, the periodicity condition is only enforced on the even derivatives; we will see why this restriction is looser for odd functions in the proof below, 
following the same technique as in Lemma~\ref{th:decaypropertiesfourierseries}.
\begin{proof}
We start from the integral formula for the sine Fourier coefficients and split the integral between $[-L,0^-]$ and $[0^+,L]$ to account for the origin. Integrating
by parts repeatedly, each integration increases the order of the derivative
of $f$ and introduces a factor $(\pi k/L)^{-1}$. After $2m+1$ integrations by parts, the boundary term involves $[f^{(2m)}(x)\cos( \pi k x/L)]_{-L}^{L}$. Since $\cos(\pi k (\pm L)/{L})=\cos(\pm \pi k)\neq0$,
this boundary term vanishes iff $f^{(2m)}(-L) = f^{(2m)}(L)=0$. Furthermore, the splitting of the integral around the origin introduces an additional term of:
$$\lim_{\{0^+,0^-\}\rightarrow 0}[f^{(2m)}(x)\cos(\pi k x/L)]_{0_-}^{0^+}\; ,$$ 
which vanishes because of the conditions imposed on $f$.

At the next step of integration by parts, the boundary term becomes
$[f^{(2m+1)}(x)\sin(\pi k x/L)]_{-L}^{L}$. Now, $\sin(\pi k (\pm L)/L)=\sin(\pm \pi k)=0$ for any bounded function $f$, and the derivatives do not need to match. For the origin term, these steps also add a term of:
$$\lim_{\{0^+,0^-\}\rightarrow 0}[f^{(2m+1)}(x)\sin(\pi k x/L)]_{0_-}^{0^+}\; ,$$
which also vanishes because  $\sin(0)=0$.
Thus, only the even-order derivatives must vanish at the boundaries and at the origin in order to obtain the desired decay.
\end{proof}

\subsection{Extension of smooth functions}
\label{App_subsec: Extension of smooth functions}

In the case of UHSVT, we are interested in functions defined on $[0,L]$ since the
function is applied to the singular values of an arbitrary matrix $A$, which are by definition positive and (can without loss of generality be rescaled to be) in the range $[0,1]$.
To obtain sufficiently fast convergence of the sine Fourier coefficients, we extend $f$ to an odd function on $[-L,L]$, which leads to the odd extension of $f$:
\begin{equation}
    f_{\text{odd}}(x)=\begin{cases}
        f(x) & \text{for } x\in [0,L]\\
        -f(-x) & \text{for } x\in [-L,0] .
    \end{cases}
    \label{eq:oddextensionL}
\end{equation}
If we want the sine Fourier series of ${f}_{\rm odd}$ to converge sufficiently fast, we only require that it satisfies the conditions of Lemma \ref{th:decaypropertiessinefourierseries}, namely that it is an odd-periodically $J$-smooth function. This, in turn, imposes constraints on $f$ itself as we shall now see:

\begin{lemma}[Decay properties of odd extensions]
Let $f:[0,L]\rightarrow\mathbb{C}$ be a (at least)
$J$-smooth function such that all even derivatives (including the $0$th order derivative, which is the function itself) up to order $J$ vanish at the endpoints: 
    \begin{equation}
        f^{(2m)}(0) = f^{(2m)}(L) = 0 \,, \qquad 0 \le 2m < J \,.
    \end{equation}
Then, $f_{\text{odd}}(x): [-L,L]\rightarrow\mathbb{C}$ defined by Eq.\ \eqref{eq:oddextensionL} is odd-periodically $J$-smooth. Consequently, 
the sine Fourier coefficients of $f_{\text{odd}}(x)$  decay as $|b_k| \in \mathcal{O}(|k|^{-J})$
and the truncation error obeys $
    \epsilon_K \in \mathcal{O}(K^{-(J-1)}) $.
\label{th:decaypropertiessinefourierseriesoddextension}
\end{lemma}

\begin{proof}

Since $f$ is $J$-smooth, $f_{\text{odd}}(x)$ automatically satisfies the conditions to be $J$-smooth everywhere, except  at the origin.
Now, by construction, the even-order derivatives  less than $J$ of $f$ 
vanish at the origin and the endpoints, hence $f_{\text{odd}}(x)$ is odd-periodically $J$-smooth and thus satisfies the conditions of Lemma \ref{th:decaypropertiessinefourierseries}.
\end{proof}

If we now want to approximate a function $f:[0,L]\rightarrow\mathbb{C}$ where the even derivatives do not vanish at the endpoint $L$, we may extend the domain from $[0,L]$ to $[0,2L]$ and approximate $f$ with a sine series in terms of $\sin(\pi k x/2L)$. For this purpose we define a new function $\Tilde{f}$ such that:
\begin{equation}
    \Tilde{f}(x)=\begin{cases}
        f(x) & \text{for } x\in [0,L]\\
        g(x) & \text{for } x\in [L,2L] 
    \end{cases}
    \label{eq:fextension}
\end{equation}
(see Figure \ref{fig:extendedxcube} for a graphical example).

To ensure sufficiently fast decay, $\tilde f(x)$ must satisfy the conditions of Lemma \ref{th:decaypropertiessinefourierseriesoddextension}, and so the even derivatives of $g(x)$ must vanish at $2L$. We also need the derivatives of $g(x)$ and $f(x)$ to match at $L$. In our particular case, we only need to ensure a convergence of $|b_k|\in\mathcal{O}(1/k^2)$ so that the sine Fourier series converges to the desired function, i.e.\ we need $\tilde f$ to be 2-smooth. To obtain this decay rate, we need the following properties:
\begin{align}
    &f(0)=0 \nonumber\\
    &g(L)=f(L), \;  g^{(1)}(L)=f^{(1)}(L) , \; g^{(2)}(L)=f^{(2)}(L), \;  \nonumber \\
    &g(2L)=0
    \label{eq:constraintsg(x)}
\end{align}
To satisfy theses properties, we must choose an ansatz for $g(x)$. A suitable choice is simply a polynomial in terms of the variable $\xi=\frac{x-L}{L}$:
\begin{equation}
g(x)\mapsto p(\xi)=C_0+C_1\xi+C_2\xi^2+C_3\xi^3 .
\end{equation}

Plugging this ansatz into \eqref{eq:constraintsg(x)} we obtain the following set of linear equations;
\begin{align}
    &C_0=f(L) \; ,  C_1=Lf^{(1)}(L) \; , C_2=1/2L^2f^{(2)}(L) \;, \nonumber \\
    &C_0+C_1+C_2+C_3=0 \; .
\end{align}
Thus the first three coefficients of $p(\xi)$ are already known, and the last one is simply:
\begin{equation}
    C_3=-C_0-C_1-C_2\;.
    \label{eq:coefficientextensionpolynomial}
\end{equation}

This construction leads to the following lemma, which we will use in the proof of the scaling of Algorithm \ref{alg:SineFourierSeries}:

\begin{lemma}
Let $f:[0,L]\rightarrow\mathbb{C}$ be a 
$2$-smooth function such that
       $ f(0)=0$.
Then, the odd extension $\tilde{f}_{\text{odd}}(x):[-2L,2L]\rightarrow \mathbb{C}$ (see Eq.\ \eqref{eq:oddextensionL}) of the periodic extension $\tilde f :[0,2L]\rightarrow \mathbb{C}$ (see Eqs.\ \eqref{eq:fextension}-\eqref{eq:coefficientextensionpolynomial}) is odd-periodically $2$-smooth, and consequently its Fourier coefficients  and truncation error satisfy
\begin{equation}
    |b_k| \in \mathcal{O}(|k|^{-2}) \,,
\quad
    \epsilon_K \in \mathcal{O}(1/K) \,.
\end{equation}
\label{lem: odd-perdiocially 2-smooth decay}
\end{lemma}

\begin{proof}
    By construction, $\tilde f$ is a $2$-smooth function with all derivatives less than $2$ vanishing at the endpoints (one of which is the origin). Therefore, by Lemma \ref{th:decaypropertiessinefourierseriesoddextension},
    $\tilde{f}_{\text{odd}}(x)$ is an odd-periodically $2$-smooth  function on $[-2L,2L]$, and thus its Fourier coefficients satisfy  the stated properties. 
\end{proof}

\begin{figure}[H]
    \centering
\includegraphics[width=0.5\linewidth]{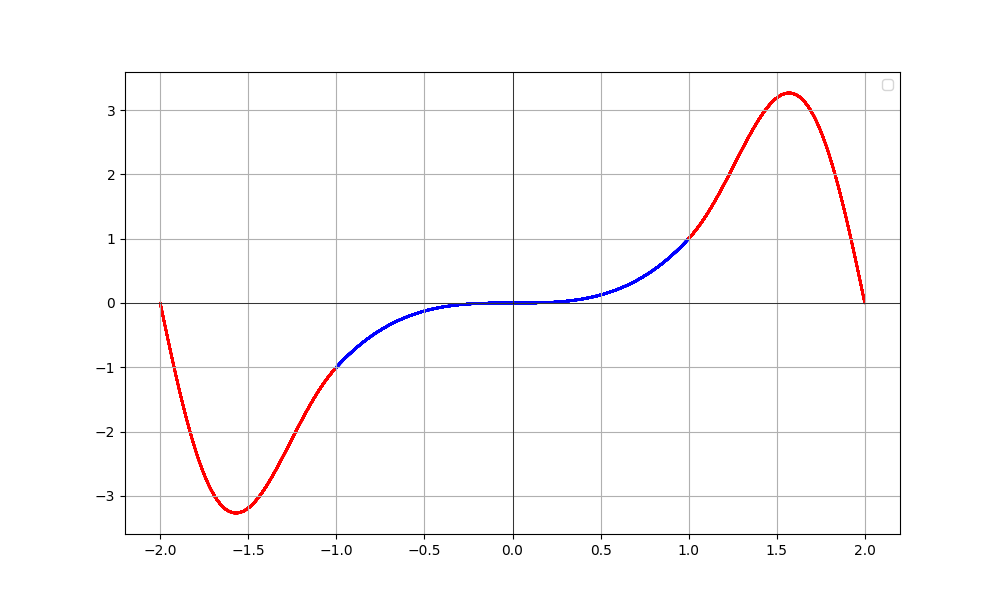}
    \caption{Odd-periodic extension of the function $f(x)=x^3$ (\textcolor{blue}{in blue}), originally defined for $x\in[0,1]$. The domain is extended to $[-2,2]$ by introducing a polynomial extrapolation $g(x)$ (\textcolor{red}{in red}) on the interval $[1,2]$, and then applying the odd extension defined in Eq.~\eqref{eq:oddextensionL}. This construction improves the convergence rate of the sine Fourier coefficients.}
    \label{fig:extendedxcube}
\end{figure}

We can of course generalise this method even further to reach any degree of smoothness at the boundary. For the qDRIFT methods, this is not important, since the scaling in terms of $\epsilon$ is the same for all $J$-smooth functions if $J\geq 2$ in condition that $f(0)=0$. However, if we want to use deterministic higher-order product formulas, then the scaling improves as the smoothness increases; it is therefore advantageous to construct a smoother odd extension. For this task, the technique is similar. Suppose we have a function $f:[0,L]\mapsto \mathbb{C}$ which is not periodic and we want a $J$-smooth periodic odd function on $[-2L,2L]$. Then, we first extend the function $f$ from $[L,2L]$ with a polynomial extension in terms of the variable $\xi=(x-L)/L$:
\begin{equation}
    p({\xi})=\sum_{j=0}^{J'} C_j \xi^j \; ,
    \label{eq: polynomialextension}
\end{equation}
for a parameter $J'$ to be determined. Then, to ensure that our extension satisfies the conditions of Lemma \ref{th:decaypropertiessinefourierseriesoddextension}, we match the derivatives of $p(\xi)$ with $f(x)$ at $x=L$ for the first $J$ derivatives; we also want the even derivatives to be zero at the boundaries.
\begin{align}
    p^{(j)}(0)&=L^jf^{(j)}(L) \quad \forall j\in [0,1,...,J] \\
    p^{(2j)}(1)&=0 \quad \forall j \in [0,1,...,\lceil J/2 \rceil -1] \,,
\end{align}
which simplifies to:
\begin{align}
    C_j=\frac{L^j}{j!}f^{(j)}(L) \quad &\forall j\in [0,1,...,J] \; , \\
    \sum_{j=k}^{J'} \frac{j!}{(j-k)!}C_j=0 \quad &\forall k \in [0,2,...,2 \cdot \lceil J/2 \rceil-2 ] \; .
\end{align}
We  must therefore choose $J'=J+\lceil J/2\rceil-1$ for this linear system to be solvable. The first $J$ coefficients of $p$ are easily computable by the first set of conditions. For the remaining $\lfloor J/2\rfloor$ coefficients, we need to solve the second set of conditions. This amounts to inverting a linear system with $\mathcal{O}(J/2)$ parameters which can be done with a linear solver scaling as $\mathcal{O}(J^3)$ in the number of parameters of the matrix. 

This can be formalised as the following lemma:

\begin{lemma}
\label{lem: J-smooth odd extension}
Let $f:[0,L]\rightarrow\mathbb{C}$ be a (at least)
$J$-smooth function such that:
    \begin{align}
        f^{(2m)}(0)=0 \quad \forall 0 \leq 2m < J \,.
    \end{align} 
Then, there exists an odd extension $\tilde{f}_{\rm odd}: [-2L,2L]\rightarrow\mathbb{C}$ which is an odd-periodic $J$-smooth function, and thus the sine Fourier coefficients of $\tilde{f}_{\rm odd}$ satisfy
\begin{equation}
    |b_k| \in \mathcal{O}(|k|^{-J}) \,, \quad
    \epsilon_K \in \mathcal{O}(1/K^{J-1}) \,.
\end{equation}
\end{lemma}

\begin{proof}
   We choose $\tilde{f}_{\rm odd}$ to be the odd extension (see Eq.\ \eqref{eq:oddextensionL}) of the function $\tilde{f}(x):[0,2L]\to \mathbb{C}$:
   \begin{equation}
    \tilde{f}(x)=\begin{cases}
        f(x) & \text{for } x\in [0,L]\\
        g(x) & \text{for } x\in [L,2L] \,,
    \end{cases}
\end{equation}
where $g(x)=p(\frac{x-L}{L})$ is the polynomial extension of $f$ on the domain $[L,2L]$ defined by Eq.\ \eqref{eq: polynomialextension} chosen such that $\tilde f$ satisfies the conditions of Lemma \ref{th:decaypropertiessinefourierseriesoddextension}. Then, by Lemma \ref{th:decaypropertiessinefourierseriesoddextension}, the function $\tilde{f}_{\rm odd}$ is odd-periodically $J$-smooth and hence satisfies the claimed decay behaviour and truncation error.
\end{proof}

\section{Proofs for GQSP}
\label{app:SU2-XY-GQSP-lemma}

\textbf{Proof of Lemma \ref{lem: relation between standard GQSP and XY-GQSP}.}
First, we use the following identity:
\begin{align}
    R(\theta, \phi,0) &=-ie^{-i\phi/2}\left(e^{i(\phi+\frac{\pi}{4})Z}e^{i\theta X}e^{i\frac{\pi}{4}Z}\right)\\
    &\doteq
    \left(e^{i\phi Z}e^{i\frac{\pi}{4}Z}e^{i\theta X}e^{i\frac{\pi}{4}Z}\right)
\end{align}
Plugging this expression inside an XY-GQSP yields:
\begin{align}
    \mathtt{GQSP}_{\Theta,\Phi,\lambda}(U) 
    &\doteq \left[\prod_{k=d}^{1}\left(e^{i(\phi_k+\frac{\pi}{4})Z}e^{i\theta_kX}e^{i\frac{\pi}{4}Z}\right)  \mathrm{ctrl}_0(U)\right]  \left(e^{i(\phi_0+\frac{\pi}{4})Z}e^{i\theta_0X}e^{i\frac{\pi}{4}Z}\right)e^{-i\lambda Z} \\&=\left[\prod_{k=d}^{1}\left(e^{i(\phi_k+\frac{\pi}{4})Z}e^{i\theta_kX}\right)  \mathrm{ctrl}_0(U)e^{i\frac{\pi}{4}Z}\right]  \left(e^{i(\phi_0+\frac{\pi}{4})Z}e^{i\theta_0X}\right)e^{-i(\lambda-\frac{\pi}{4}) Z}\\
    &=\left(e^{i(\phi_d+\frac{\pi}{4})Z}e^{i\theta_dX}\right)\left[\prod_{k=d-1}^{0}
    \mathrm{ctrl}_0(U)
    \left(e^{i(\phi_k+\frac{\pi}{2})Z}e^{i\theta_kX}\right)  \right]  e^{-i(\lambda-\frac{\pi}{4}) Z} \,.
\end{align}
In the second line, we used the fact that $e^{i\frac{\pi}{4} Z}$ commutes with $\mathrm{ctrl}_0(U)$. The intermediate factors $e^{i\frac{\pi}{4} Z}$ can then be combined with the $e^{i\frac{\pi}{4} Z}$ appearing in the subsequent stage of the product, yielding $e^{i\frac{\pi}{2} Z}$.
Therefore, a GQSP protocol that uses the standard processing operators of Eq.~\eqref{eq: GQSP processing operators} and parameterized by $(\Theta,\Phi,\lambda)$ can be transformed into one that uses processing operators of the form $e^{i\phi Z}e^{i\theta X}$ and parameterized by the angles $(\Theta,\Phi',\lambda')$, with $\Phi'=(\phi_0+\frac{\pi}{2},\phi_1+\frac{\pi}{2},\ldots,\phi_d+\frac{\pi}{4})$ and $\lambda'=\lambda-\pi/4$.
The next step is to transform $e^{i\phi Z}e^{i\theta X}$ into an $XY$-rotation. To do so, observe that an $XY$-rotation can be written as an $X$-rotation conjugated by $Z$-rotations, that is,
$R_{\rm xy}(\theta,\phi)=e^{i\frac{\phi}{2} Z}e^{-i\theta X}e^{-i\frac{\phi}{2}Z}$.
Therefore, we write:
\begin{equation}
(e^{i\phi_d'Z}e^{i\theta_d X})
=
(e^{i\phi_d'Z}e^{i\theta_d X}e^{-i\phi_d'Z})
e^{i\phi_d'Z}
=
R_{\rm xy}(-\theta_d,2\phi_d')
e^{i\phi_d'Z} \;  .
\end{equation}
We can then absorb the remaining factor $e^{i\phi_d'Z}$ in the processing operator of the next iterate of the GQSP, resulting in a new phase angle $e^{i(\phi_{d-1}'+\phi_d')Z}$. Repeating this procedure recursively for all remaining layers, we obtain
\begin{equation}
\tilde{\phi_k}
=
2\sum_{\ell=k}^{d}\phi_\ell'
=
\frac{\pi}{2}+(d-k)\pi
+
2\sum_{\ell=k}^{d}\phi_\ell \; ,
\end{equation}
and the residual $Z$-rotation at the end of this process can be absorbed into $\lambda$, giving
$\tilde{\lambda}=\lambda-\pi/2-d\pi/2-\sum_{\ell=0}^{d}\phi_{\ell}$. Finally, taking $\tilde{\theta}=-\theta$ concludes the proof.
\qed

\section{The Nonlinear Fourier Transform}\label{App: The Nonlinear Fourier Transform}

In this section, we introduce the Nonlinear Fourier Transform (NLFT) \cite{tao2012nonlinear} and review its connection to the standard Fourier series.
This perspective reveals an instructive analogy: UHET plays, with respect to GQSP, a role similar to that of the classical Fourier transform relative to its nonlinear counterpart. This makes sense, given that \cite{laneve2025generalized} has shown an equivalence between the NLFT and GQSP, while UHET is based on a Fourier series simulation.

First, recall that the Fourier series of a function,
\begin{equation}
    f(x)=\sum_{k=-\infty}^{\infty} c_k e^{-i k \pi x} \,,
\end{equation}
can be interpreted as the action of a superoperator
\begin{equation}
\mathcal{F}(\{c_k\}_{k \in \mathbb{Z}})
\;\longrightarrow\;
f(x)=\sum_{k=-\infty}^{\infty} c_k e^{-i k \pi x}.
\end{equation}
In this view, $\mathcal{F}$ takes as input an infinite sequence 
$\{c_k\}_{k\in\mathbb{Z}}$ of complex numbers and outputs the
corresponding Fourier series. Conversely, one may define the inverse
superoperator $\mathcal{F}^{-1}$ acting on a function and returning the
sequence of its Fourier coefficients.

The nonlinear Fourier transform follows a similar conceptual pattern.
It can be viewed as a superoperator $\mathcal{F}_{\mathrm{NL}}$ acting on a
sequence $\{c_k\}_{k\in\mathbb{Z}}$, but instead of producing a scalar
function it outputs an infinite product of SU(2) matrices
parametrised by the coefficients $\{c_k\}_{k\in\mathbb{Z}}$.

\begin{equation}
    \mathcal{F}_{\mathrm{NL}}(\{c_k\}_{k\in\mathbb{Z}})=\prod_{k\in \mathbb{Z}} \frac{1}{\sqrt{1+|c_k|^2}}\begin{bmatrix}
        1 & c_kz^k \\
        c_kz^{-k} & 1
    \end{bmatrix}=\begin{bmatrix}
        a(z) & \cdot \\
        b(z) & \cdot
    \end{bmatrix}
\end{equation}
Where $|a(z)|^2+|b(z)|^2=1$ for all $z\in \mathbb{T}$. Note that the RHS is an SU(2)-valued function, that is, a function taking a complex number $z$ on the complex unit circle ($\mathbb{T}$) and outputting an SU(2) matrix depending on this parameter $z$. Interestingly, if we restrict the values of $k$ to lie in $[-d,...,d]$, $a(z)$ and $b(z)$ become Laurent polynomials. The comparison to GQSP is then immediate. To further see this connection, consider the following:

\begin{align}
    \mathcal{F}_{\mathrm{NL}}(\{c_k\}_{k\in\mathbb{Z}})&=\prod_{k\in \mathbb{Z}} \frac{1}{\sqrt{1+|c_k|^2}}\begin{bmatrix}
        1 & c_kz^k \\
        c_kz^{-k} & 1
    \end{bmatrix} \nonumber\\
    &=
    \prod_{k\in \mathbb{Z}} 
    \begin{bmatrix}
        1 & 0\\
        0 & z^{-k}
    \end{bmatrix}
    \begin{bmatrix}
        \frac{1}{\sqrt{1+|c_k|^2}} & \frac{c_k}{\sqrt{1+|c_k|^2}} \\
        \frac{c_k}{\sqrt{1+|c_k|^2}} & \frac{1}{\sqrt{1+|c_k|^2}}
    \end{bmatrix}
    \begin{bmatrix}
        1 & 0\\
        0 & z^k
    \end{bmatrix}.
\end{align}
Taking the substitution $c_k=i\tan(\theta_k)e^{-i\phi_k}$ yields:
\begin{align}
    &=\prod_{k\in \mathbb{Z}} 
    \begin{bmatrix}
        1 & 0\\
        0 & z^{-k}
    \end{bmatrix}
    \begin{bmatrix}
        \frac{1}{\sqrt{1+\tan(\theta_k)^2}} & \frac{i\tan(\theta_k)e^{-i\phi_k}}{\sqrt{1+\tan(\theta_k)^2}} \\
        \frac{i\tan(\theta_k)e^{-i\phi_k}}{\sqrt{1+\tan(\theta_k)^2}} & \frac{1}{\sqrt{1+\tan(\theta_k)^2}}
    \end{bmatrix}
    \begin{bmatrix}
        1 & 0\\
        0 & z^k
    \end{bmatrix} \nonumber\\
    &=\prod_{k\in \mathbb{Z}} 
    \begin{bmatrix}
        1 & 0\\
        0 & z^{-k}
    \end{bmatrix}
    \begin{bmatrix}
        \cos(\theta_k) & i\sin(\theta_k)e^{-i\phi_k}  \\
        i\sin(\theta_k)e^{-i\phi_k} & \cos(\theta_k)
    \end{bmatrix}
    \begin{bmatrix}
        1 & 0\\
        0 & z^k
    \end{bmatrix}\\
    &=\prod_{k\in \mathbb{Z}} 
    \begin{bmatrix}
        1 & 0\\
        0 & z^{-k}
    \end{bmatrix}
    R_{\rm xy}(\theta_k,\phi_k)
    \begin{bmatrix}
        1 & 0\\
        0 & z^k
    \end{bmatrix}=\begin{bmatrix}
        a(z) & \cdot \\
        b(z) & \cdot
    \end{bmatrix}
    \label{eq: NLFT}
\end{align}
Interestingly, this has the very same structure as UHET. Indeed, recall that UHET consists of approximating the Fourier series using a small angle approximation of $\theta_k\rightarrow\theta_k/N$. Since we have $c_k=-i\tan(\theta_k/N)e^{-i\phi_k}=-i\frac{\theta_k}{N}e^{-i\phi_k}+\mathcal{O}(1/N^3)$, the small angle approximation is equivalent to the small coefficients $c_k\rightarrow c_k/N$ approximation. Taking this approximation $|c_k|\rightarrow|c_k|/N$  to the nonlinear Fourier transform yields:
\begin{align}
    \mathcal{F}_{\mathrm{NL}}(\{c_k/N\}_{k\in\mathbb{Z}})
    &=\prod_{k\in \mathbb{Z}} \frac{1}{\sqrt{1+|c_k/N|^2}}\begin{bmatrix}
        1 & \frac{c_kz^{k}}{N} \\
        \frac{c_kz^{-k}}{N} & 1
    \end{bmatrix}\\
    &=
    \prod_{k\in \mathbb{Z}} \frac{1}{\sqrt{1+|c_k/N|^2}}\left(\begin{bmatrix}
        1 & 0 \\
        0 & 1
    \end{bmatrix}+
    \frac{1}{N}
    \begin{bmatrix}
        0 & c_kz^{k} \\
        c_kz^{-k} & 0
    \end{bmatrix}\right)\\
    &=\begin{bmatrix}
    1 & 0\\
    0 & 1
    \end{bmatrix}+\frac{1}{N}\sum_{k\in \mathbb{Z}}\begin{bmatrix}
        0 & c_kz^{k} \\
        c_kz^{-k} & 0
    \end{bmatrix}+\mathcal{O}(1/N^2)\\
    &=\mathbb{I}+\frac{1}{N}\begin{bmatrix}
        0 & \mathcal{F}(\{c_k\}_{k \in \mathbb{Z}}) \\
        \mathcal{F}(\{c_k\}_{k\in \mathbb{Z}}) & 0
    \end{bmatrix}+\mathcal{O}(1/N^2)
    \; ,
\end{align}
where from the second to third step, we collected all the terms depending on $1/N$; note that \\ $1/\sqrt{1+|c_k/N|^2}\approx 1+\mathcal{O}(1/N^2)$. We immediately see that this linearisation of NLFT yields the standard linear Fourier transform in the off-diagonal blocks. It is in that sense that we can say UHET is to GQSP what the Fourier transform is to the NLFT. 

\section{Proofs for UHSVT}

\subsection{Proof of Theorem \ref{thm: UHSVT}}
\label{app:proof_det_UHSVT}

\textbf{Proof of Theorem \ref{thm: UHSVT}:}

Using the triangle inequality, we have
\begin{align}
\left\|
\mathtt{UHSVT}^{\rm det}_{\tilde f_{{\rm odd},K},t,N}(H[A])-e^{-iH[f^{\rm SV}(A)]t}\right\|
\nonumber\
\leq
\left\|
\mathtt{UHSVT}^{\rm det}_{\tilde f_{{\rm odd},K},t,N}(H[A])
-
e^{-iH[\tilde f_{\rm odd,K}^{\rm SV}(A)]t}
\right\|\\
+
\left\|
e^{-iH[\tilde f_{\rm odd,K}^{\rm SV}(A)]t}
-
e^{-iH[f^{\rm SV}(A)]t}
\right\|.
\label{eq: UHSVT, error triangle inequality}
\end{align}
For the first term, by construction, the unitary
$\mathtt{UHSVT}^{\rm det}_{\tilde f_{{\rm odd},K},t,N}(H[A])$
implements
$e^{-iH[\tilde f_{\rm odd,K}^{\rm SV}(A)]t}$
using the first-order Trotter--Suzuki formula. From standard Trotter error analysis, choosing the number of
Trotter steps $N \in \mathcal{O}(\beta^2 t^2 /\epsilon)$, with $\beta=\sum_{k=1}^{K}|b_k|$ therefore guarantees that
\begin{equation}
\left\|
\mathtt{UHSVT}^{\rm det}_{\tilde f_{\rm odd},K,t,N}(H[A])
-
e^{-iH[\tilde f_{\rm odd,K}^{\rm SV}(A)]t}
\right\|
\leq \frac{\epsilon}{2}.
\end{equation}

For the second term, Lemma~7 of \cite{kang2025quantum} gives
\begin{align}
\left\|
e^{-iH[\tilde f_{\rm odd,K}^{\rm SV}(A)]t}
-
e^{-iH[f^{\rm SV}(A)]t}
\right\|
\nonumber\
&\leq
t
\left\|
H[\tilde f_{\rm odd,K}^{\rm SV}(A)]
-
H[f^{\rm SV}(A)]
\right\|
\nonumber \\
&\leq
t
\left\|
\tilde f_{\rm odd,K}(x)-f(x)
\right\|_{[0,1]}
\nonumber\ \\
&\leq
t
\left\|
\tilde f_{\rm odd,K}(x)-\tilde f_{\rm odd}(x)
\right\|_{[-2,2]}
\nonumber \\
&\leq
t\epsilon_K.
\label{eq: UHSVT Fourier truncation error}
\end{align}
The third line follows from the fact that $\tilde f_{\rm odd}$ is the odd periodic
extension of $f$ and therefore coincides with $f$ on $[0,1]$. Since
$f$ is $2$-smooth and satisfies $f(0)=0$, Lemma~\ref{lem:
odd-perdiocially 2-smooth decay} implies that its odd periodic extension
$\tilde f_{\rm odd}$ is periodically $2$-smooth. Consequently, Lemma~\ref{lem:
odd-perdiocially 2-smooth decay} further implies that the Fourier truncation
error satisfies
\begin{equation}
\epsilon_K
=
\left\|
\tilde f_{\rm odd,K}- \tilde f_{\rm odd}
\right\|_{[-2,2]}
\leq
\frac{C}{K}
\end{equation}
for some constant $C>0$. Thus, choosing
\begin{equation}
K
\geq
\frac{2Ct}{\epsilon}
\end{equation}
ensures
\begin{equation}
t\epsilon_K\leq\frac{\epsilon}{2}.
\end{equation}
Combining the two error contributions in
Eq.~\eqref{eq: UHSVT, error triangle inequality}, we obtain
\begin{equation}
\left\|
\mathtt{UHSVT}^{\rm det}_{\tilde f_{{\rm odd},K},t,N}(H[A])
-
e^{-iH[f^{\rm SV}(A)]t}
\right\|
\leq
\frac{\epsilon}{2}+\frac{\epsilon}{2}
=
\epsilon \,.
\end{equation}
Finally, we invoke Lemma~\ref{lem:
odd-perdiocially 2-smooth decay} again,  which ensures that the sine Fourier coefficients $|b_k|$ satisfy $ |b_k| \in \mathcal{O}(|k|^{-2}) $, meaning that $\beta$ converges to a finite constant. Combining this with the above, shows that $N \in\mathcal{O}(t^2/\epsilon)$ and $K \in \mathcal{O}(t/\epsilon)$, concluding the proof.
\qed

\subsection{Proof of Algorithm \ref{alg:SineFourierSeries}}
\label{app:proof_rand_UHSVT}

In the following, we will make use of the diamond norm, which is defined as:
\begin{equation}
    ||\mathcal{C}||_{\diamond}:=\max_{\ket{\psi}\in \mathcal{H\otimes \mathcal{H}^{'}} , ||\ket{\psi}||=1} || \mathcal{C}\otimes \mathcal{I}_{\mathcal{H^{'}}}(\ket{\psi}\bra{\psi})||_1
\end{equation}
for any arbitrary Hilbert space $\mathcal{H}^{'}$.

\begin{theorem}[Scaling of Algorithm \ref{alg:SineFourierSeries}] 
Algorithm \ref{alg:SineFourierSeries} outputs a channel $\mathcal{G_{\rm approx}}$ which approximates the channel $\mathcal{G}(\rho)=e^{-iH[f_{SV}(A)]t}\rho e^{iH[f_{SV}(A)]t}$, where $H[f_{\rm SV}(A)]$ is a Hamiltonian block-encoding in the canonical form (see Def. \ref{def:standardform}) of the matrix $f_{\rm SV}(A)$, with error $\epsilon$ in the diamond norm:
\begin{equation}
    || \mathcal{G}- \mathcal{G_{\rm approx}}||_{\diamond} \leq \epsilon \; ,
\end{equation}
with a time complexity of $\mathcal{O}(t/\epsilon)$ and total evolution time $\mathcal{O}(\frac{t^2}{\epsilon}\log(\frac{t}{\epsilon}))$.
\label{thm: proofscalingalgorithm1}
\end{theorem}

\begin{proof}
To prove our claim, we define intermediate channels, $\mathcal{G}^{(K)}(\rho)=e^{-iH[\tilde{f}^{(K)}_{\rm SV}(A)]t}\rho e^{iH[\tilde{f}^{(K)}_{\rm SV}(A)]t}$, associated to the unitary $e^{-iH[\tilde{f}^K_{\rm SV}(A)]t}$ which is the order $K$ truncation of the sine series of $\tilde{f}_{\rm odd}(x)$, the odd periodic extension of $f$, applied to the singular values of $A$. We therefore have:
\begin{equation}
    ||\mathcal{G}-\mathcal{G}_{\rm approx} ||_{\diamond}\leq \underbrace{||\mathcal{G}-\mathcal{G}^{(K)}_{\rm} ||_{\diamond}}_{\textcolor{wineRed}{(1)}}+\underbrace{||\mathcal{G}^{(K)}-\mathcal{G}_{\rm approx} ||_{\diamond}}_{\textcolor{wineRed}{(2)}} \,.
\label{eq:finalbound}
\end{equation}
\begin{enumerate}
    \item $||\mathcal{G}-\mathcal{G}^{(K)}_{\rm} ||_{\diamond}$: this is the diamond norm of two unitary channels, associated to $e^{-iH[f_{\rm SV}(A)]t}$ and $e^{-iH[f^{(K)}_{\rm SV}(A)]t}$ respectively. The diamond norm between two unitary channels $\mathcal{U}(\cdot)=U(\cdot)U^\dagger$ and $\mathcal{V}(\cdot)=V(\cdot)V^\dagger$  is bounded by the spectral norm $||\cdot||_{\infty}$ (the largest singular value) with:
    \begin{equation}
        ||\mathcal{U}-\mathcal{V}||_{\diamond}\leq 2||U-V||_{\infty} \,.
    \end{equation}
    Therefore,
    \begin{align}
        ||\mathcal{G}-\mathcal{G}^{(K)}||_{\diamond} &\leq 2||e^{-iH[f_{SV}(A)]t}-e^{-iH[\tilde{f}^{(K)}_{\rm SV }(A)]t}||_{\infty} \nonumber\\
        &\leq 2|| f_{SV}(A)-\tilde{f}_{\rm SV}^{(K)}(A)|| t \nonumber\\
        &=2\|f(x)-f^{(K)}(x)\|_{[0,1]} t\\
        &\leq \epsilon/2 \nonumber \; .
    \end{align}
    The second line follows from Lemma 7 in \cite{kang2025quantum}. The third line follows by the assumption that the maximum singular value of A is bounded by $1$ and because $\tilde{f}$  restricted to the domain $[0,1]$ is simply $f$. Finally, the last line follows because we specifically chose $K$ so that the truncation error is bounded by $\epsilon/4t$.
    \item For $||\mathcal{G}^{(K)}_{\rm}-\mathcal{G}_{\rm approx}||_{\diamond}$, recall that $\mathcal{G}_{\rm approx}$ is a qDRIFT  channel, and since we chose $N=N(\beta,t,\epsilon/2)$ according to Eq. \ref{eq:qDRIFT}, we can invoke Lemma \ref{lem:qDRIFT}:
    \begin{equation}
        ||\mathcal{G}^{(K)}_{\rm}-\mathcal{G}_{\rm approx}||_{\diamond}\leq \epsilon/2
    \end{equation}
\end{enumerate}
Combining the two above results in Eq. \eqref{eq:finalbound} yields the desired error.

The time complexity of the procedure is simply proportional to the number of qDRIFT iterations needed to reach accuracy $\epsilon$, which by Lemma \ref{lem:qDRIFT} is $\mathcal{O}(\beta^2t^2/\epsilon)$. Since $\tilde{f}_{\rm odd}$ is an odd-periodically $2$-smooth function, by Lemma \ref{lem:  odd-perdiocially 2-smooth decay}, $\beta\in \mathcal{O}(1)$. 

The average total evolution time can be computed as:
\begin{equation}
    N\times \mathcal{O}\left(\sum_{k=-K}^K|b_k|k\right) \,,
\end{equation}
as stated in the main text, 
which scales like $N\times \mathcal{O}(\log(K))$ since $|b_k|\in \mathcal{O}(k^{-2})$ by Lemma \ref{lem:  odd-perdiocially 2-smooth decay}. Moreover, since we chose $K\in \mathcal{O}(t/\epsilon)$, this suffices to show the claimed scaling.

\end{proof}

\section{Proofs for Hamiltonian QSVT}
\label{App: Proof for Hamiltonian QSVT}
In this section, we prove several lemmas and theorems concerning the Hamiltonian Quantum Singular Value Transformation (H-QSVT), culminating in proofs of theorems and algorithms presented in the main text. In Subsection~\ref{App_subsec: Regions in the Complex Plane}, we develop the mathematical tools required for the subsequent analysis. In Subsection \ref{App_subsec: Approximating functions with Chebyshev polynomials}, we develop the general theory of function approximation with Chebyshev polynomials. In Subsection \ref{App_subsec: Some polynomial approximations} we construct a polynomial approximation to some functions that will become important later, establishing several key results, including the construction of a definite-parity polynomial approximation of $\arccos$. 
In Subsection~\ref{App_subsec: Hamiltonian QSVT in the odd case}, we prove Theorem~\ref{thm: HQSVT with lower bound} for the case of odd polynomials. In contrast to \cite{kang2025quantum}, where access to $X$-rotations on the qubitised subspace is assumed, our proof does not rely on this assumption. Then, in Subsection~\ref{App_subsec: Hamiltonian QSVT in the general case}, we extend the construction to arbitrary polynomials and thus prove Theorem \ref{thm: HQSVT with lower bound} in the general case. 
Finally, in Subsection \ref{App:proof_alg2}, we prove the scaling of Algorithm \ref{alg:QSVT-based algorithm}.
While many of the arguments closely follow those of \cite{kang2025quantum}, we place particular emphasis on the dependence of the complexity on the parameters $\delta$ and $\Lambda$, which respectively denote lower and upper bounds on the singular values of $A$.

\subsection{Regions in the Complex Plane}
\label{App_subsec: Regions in the Complex Plane}

Finding good polynomial approximations to a function typically requires analyzing 
its behaviour over regions of the complex plane. For functions defined on $[-1,1]$, 
the natural approximation basis is that of Chebyshev polynomials, and the relevant 
regions are \emph{Bernstein ellipses}: ellipses with foci at $\pm 1$ and centre at 
the origin. A Bernstein ellipse is parametrised by three quantities: the semi-major 
axis $a$, the semi-minor axis $b$, and the elliptical radius $\rho$, which are 
related by
\begin{align}
    a &= \frac{\rho + \rho^{-1}}{2}, & b = \frac{\rho - \rho^{-1}}{2}, && \rho = a + \sqrt{a^2 - 1}, \\
    a &= \sqrt{b^2 + 1}, & b = \sqrt{a^2 - 1}, & & \rho= b + \sqrt{1 + b^2}.
    \label{eq:bernstein_rel}
\end{align}
Note that specifying any one of $a$, $b$, or $\rho$ uniquely determines the other two, so the ellipse is effectively a one-parameter family. Thus, as an abuse of notation, we will denote a Bernstein Ellipse by $\textbf{Ellipse}(\rho=\cdot)$, $\textbf{Ellipse}(a=\cdot)$ and $\textbf{Ellipse}(b=\cdot)$. Here are some simple properties of such ellipses:
\begin{lemma}
    For any $0 < \delta \leq 1$,
    \begin{align}
        \mathbf{Ellipse}\!\left(a = 1 + \frac{\delta^2}{4}\right) 
        &\subseteq \mathbf{Ellipse}\!\left(\rho = 1 + \delta\right) 
        \subseteq \mathbf{Ellipse}\!\left(a = 1 + \frac{\delta^2}{2}\right), \\
        \mathbf{Ellipse}\!\left(b = \frac{3\delta}{4}\right) 
        &\subseteq \mathbf{Ellipse}\!\left(\rho = 1 + \delta\right) 
        \subseteq \mathbf{Ellipse}\!\left(b = \delta\right).
    \label{eq:EllipseContainment}
    \end{align}
\end{lemma}
We will also consider $\mathbf{Stadium}([c_0, c_1], b)$, defined as the set of all points in the complex plane whose distance to the real interval $[c_0, c_1]$ is at most $b$. We have the following containments (see Figure \ref{fig:EllipseStadium}):
\begin{align}
    &\textbf{Ellipse}(a,b)\subseteq \textbf{Stadium}([-a+b,a-b], b)\\
    &\textbf{Ellipse}(\rho=1+\alpha)\subseteq \textbf{Stadium}\left(\left[-\frac{1}{1-\alpha},\frac{1}{1-\alpha}\right], \frac{2\alpha+\alpha^2}{2(1+\alpha)}\right) \,,
    \label{eq:EllipseStadiumContainment}
\end{align}
where the second containment is a direct application of the first one using $\rho=1+\alpha$.

\begin{figure}[htb]
    \centering
    \includegraphics[width=0.5\linewidth]{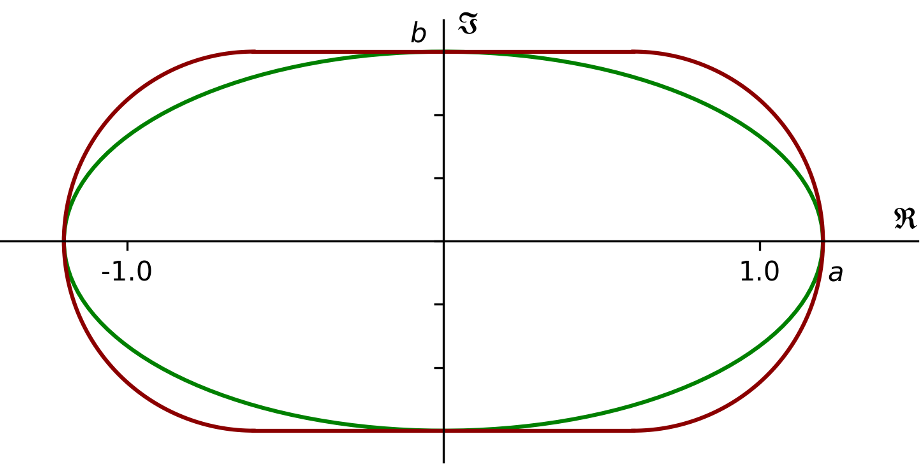}
    \caption{Containment of a Bernstein Ellipse (green) inside a Stadium (red).}
    \label{fig:EllipseStadium}
\end{figure}

\subsection{Approximating functions with Chebyshev polynomials}
\label{App_subsec: Approximating functions with Chebyshev polynomials}
In this section, we give the basics of approximation theory using Chebyshev polynomials. 

We begin by recalling that functions defined on $[-1,1]$ that are analytically continuable to the interior of an ellipse in the complex plane with elliptic radius $\rho$ (denoted $\textbf{Ellipse}(\rho)$) can be efficiently approximated in the Chebyshev polynomial basis. Recall that the Chebyshev polynomials are defined as:
\begin{equation}
    T_k(x) = \cos(k\arccos(x)),
\end{equation}
where $T_k$ has degree $k$. An analytic function is characterised with the following definition:

\begin{definition}[Analytic function]
\label{def: analytic functions}
Let $\mathbf{D}\subset\mathbb{C}$ be an open set. A function
$f:\mathbf{D}\rightarrow\mathbb{C}$
is said to be \emph{analytic} if, for every
$x_0\in\mathbf{D}$, there exists a power series
\begin{equation}
f(x)=\sum_{n=0}^{\infty}a_n(x-x_0)^n
\end{equation}
that converges to $f(x)$ in a neighbourhood of $x_0$, where
$a_n\in\mathbb{C}$.
\end{definition}

A standard theorem in approximation theory ensures that any analytic function $f$ can be well approximated in the Chebyshev basis:
\begin{equation}
    f(x) \approx f_K(x) = \sum_{k=0}^K a_k T_k(x) \; , \quad \text{with } a_k = \int_{-1}^{1} \frac{f(x)T_k(x)}{\sqrt{1-x^2}}\,dx \; .
    \label{eq:Chebcoeff}
\end{equation}

\begin{lemma}[variation on \cite{trefethen2019approximation}]
\label{thm: chebyshev decay properties of analytic functions}
Let $f$ be analytic on $[-1,1]$ and analytically continuable to the interior of $\textbf{Ellipse}(\rho = 1+\alpha)$ for some $\alpha > 0$, where it satisfies $|f(x)| \leq M$. Then its Chebyshev coefficients satisfy $|a_0| \leq M$ and $|a_k| \leq 2Me^{-\alpha k}$ for $k \geq 1$. Consequently, for each $K \geq 0$, the Chebyshev truncation satisfies:
\begin{equation}
    \|f - f_K\|_{[-1,1]} \in \mathcal{O}\!\left(\frac{Me^{-K\alpha}}{\alpha}\right).
\end{equation}
Thus, taking $K = \mathcal{O}\!\left(\alpha^{-1}\log\!\left(\frac{M}{\epsilon}\right)\right)$ yields an $\epsilon$-close approximation to $f$.
\label{thm:Chebtruncerroranalytic}
\end{lemma}

\begin{proof}
By Lemmas 8.1 and 8.2 in \cite{trefethen2019approximation}, we have $|a_k| \leq 2M\rho^{-k} = 2M(1+\alpha)^{-k} \leq 2Me^{-\alpha k}$. It then follows that $\|f - f_K\| \leq \sum_{k=K+1}^\infty |a_k| \leq \frac{M}{\alpha}e^{-\alpha K}$, where the sum is bounded by an integral.
\end{proof}

Now, if instead of analyticity one only assumes smoothness of $f$, the following theorem is more appropriate:

\begin{lemma}[\cite{Martyn_2025}]
    Let $f$ be a $J$-smooth function on $[-1,1]$ (see Definition~\ref{def: J-smooth}). Then its Chebyshev coefficients satisfy $a_k \in o\!\left(k^{-J}\right)$ and its Chebyshev truncation satisfies:
    \begin{equation}
        \|f - f_K\|_{[-1,1]} \in \mathcal{O}(K^{-J+1}).
    \end{equation}
Thus, taking $K \in \mathcal{O}(\epsilon^{-1/(J-1)})$ yields an $\epsilon$-close approximation to $f$. Similarly, if $f \in C^\infty([-1,1])$, then the Chebyshev coefficients decay super-polynomially as $a_k \in \mathcal{O}\!\left(e^{-qk^r}\right)$ for some $q, r > 0$.
\label{thm:Chebtruncerrorsmooth}
\end{lemma}

\begin{proof}
Starting from Eq.~\eqref{eq:Chebcoeff} and applying the change of variables $x = \cos(\theta)$ yields:
    \begin{equation}
        a_k = \int_0^{\pi} f(\cos(\theta))\cos(k\theta)\,d\theta = \int_0^{2\pi} f(\cos(\theta))e^{-ik\theta}\,d\theta \; ,
    \end{equation}
where the second equality uses the fact that $f(\cos(\theta))$ is an even function of $\theta$. Hence $a_k$ is a Fourier coefficient of the $J$-smooth periodic function $f(\cos(\theta))$, and the result follows by invoking Theorem~\ref{th:decaypropertiesfourierseries}.
\end{proof}
Recall that there exist functions that are smooth but not analytic. Bump functions are examples of functions with this property.
Now, to approximate a function on $[\eta,1-\eta]$, a good starting point would be to rescale the interval $x\in[\eta,1-\eta]$ to the canonical domain $y\in[-1,1]$ using a linear transformation and then use Lemmas \ref{thm:Chebtruncerroranalytic} and \ref{thm:Chebtruncerrorsmooth}. However, a difficulty arises because Chebyshev polynomials typically grow rapidly outside the interval $[-1,1]$. Consequently, after rescaling back to the original interval $[\eta,1-\eta]$, the resulting polynomial approximation may violate the boundedness condition ($|\tilde{p}(x)|\leq 1$) required in the usual QSP setting. Controlling this growth outside the approximation domain is the main purpose of Corollary 34 in \cite{kang2025quantum}, which is a slight variation of Theorem 21 in \cite{tang2302cs}:  
\begin{lemma}[Dominated polynomial approximation via Bernstein ellipse \cite{kang2025quantum}]
Let $f$ be analytic over a complex region enclosing $\mathbf{Ellipse}(\rho = 1 + \alpha)$ for some
$\alpha > 0$,  where it is bounded by $M$ and such that $f_{[-1,1]}$ is real. For any $\epsilon> 0$ and $0 < \eta\leq 1 < b$,
there exists a real polynomial $q$ satisfying
\begin{alignat}{3}
    \|q - f\|_{[-1+\eta,\, 1-\eta]} &\leq \epsilon\\
    \|q\|_{[-1,1]} &\leq \|f(x)\| + \epsilon, \\
    \|q\|_{[-b,-1] \cup [1,b]} &\leq \epsilon
\end{alignat}
with the asymptotic degree
\begin{equation}
    \mathbf{O}\!\left(
        \frac{b}{\alpha^2} \log\!\left(\frac{b\,M}{\alpha^2 \epsilon}\right)
        + \frac{b}{\eta} \log\!\left(\frac{M}{\epsilon}\right)
    \right).
\end{equation}
\label{lem: boundedapproximation}
\end{lemma}

\subsection{Some polynomial approximations}
\label{App_subsec: Some polynomial approximations}
We begin by using Lemma \ref{lem: boundedapproximation} to construct a definite parity polynomial approximation to the function $\arccos(x)$ on the domain $[\eta,1-\eta]$:
\begin{lemma}[Definite parity polynomial approximation to $\arccos(x)$] 
Let $\epsilon_{\arccos},\eta \in \mathbb{R}$ such that $0\leq\eta\leq 1/2$. Then, there exists a degree $d_{\arccos}$ polynomial $p_{\arccos}(x)$ such that:
\begin{enumerate}[label=\textcolor{wineRed}{(\arabic*)}]
    \item $\|p_{\rm arccos}(x)-\arccos(x)\|_{[\eta,1-\eta]}\leq\epsilon_{\arccos}$,
    \item $p_{\arccos}$ has a definite parity, 
    \item $\|p_{\arccos}\|_{[-1,1]}\leq \pi/2$ \; .
\end{enumerate}
Furthermore, $d_{\arccos}$ has an asymptotic scaling of $\mathcal{O}[\frac{1}{\eta}\log(\frac{1}{\eta\epsilon_{\arccos}})]$. 
\label{lem: arccosapprox}
\end{lemma}
\begin{proof}
We use Lemma $\ref{lem: boundedapproximation}$. To fit our domain to the conditions of this lemma, we first rescale using the linear function: 
\begin{equation}
    L(x)=\frac{1-\eta}{1/2-\eta}(x-1/2) \; .
    \label{eq: rescaling}
\end{equation}
This transformation maps the interval $[\eta,1-\eta]$ onto $[-1+\eta,1-\eta]$.
Furthermore, we know that $\arccos(x)$ is analytic in $\mathbb{D}=\left\{z\in\mathbb{\gamma}:|z|<1\right\}$. After rescaling our function with the transformation \eqref{eq: rescaling}, $\mathbb{D}$ is mapped to:
\begin{equation}
    L(\mathbb{D})=\{z\in \mathbb{C}:|L^{-1}(z)|< 1\} \; .
\end{equation}
Geometrically, $L(\mathbb{D})$ is the interior of a disk in the complex plane centered at $-(1-\eta)/(1-2\eta)$ with radius $2(1-\eta)/(1-2\eta)$. Notice that $1+\eta\leq (1-\eta)/(1-2\eta)$ for $\eta\leq 1/2$. It follows that the Bernstein ellipse centred at the origin with semi-major axis $a=1+\eta$ is entirely contained in $L(\mathbb{D})$. Using this observation together with the first containment in~\eqref{eq:EllipseContainment}, we obtain:
\begin{equation}
    \mathbf{Ellipse}(\rho=1+\sqrt{2\eta})
    \subseteq
    \mathbf{Ellipse}(a=1+\eta)
    \subseteq
    L(\mathbb{D}) .
\end{equation}
This containment is illustrated in Fig.~\ref{fig:Ellipses}.

\begin{figure}[htb]
    \centering
    \includegraphics[width=0.6\linewidth]{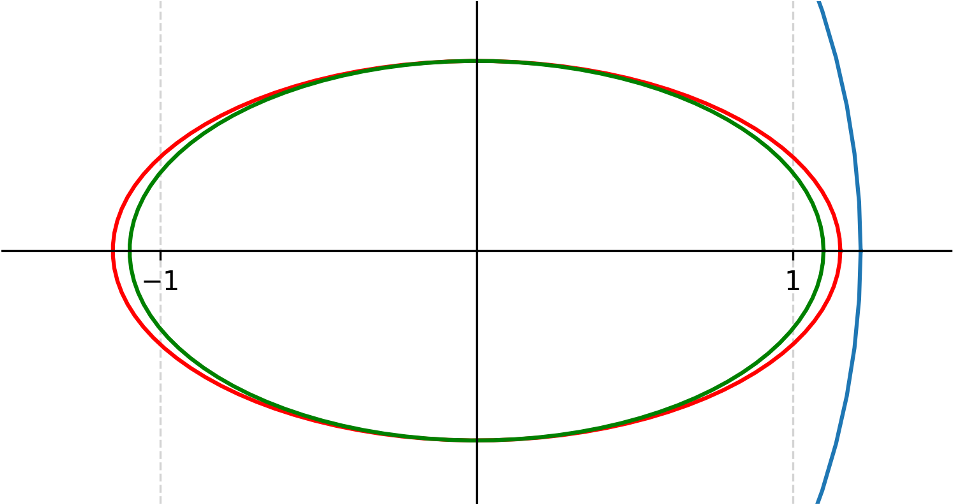}
    \caption{Bernstein ellipses containment in the complex plane. The blue curve is $L^{-1}(\mathbb{D})$, a circle centered at $-(1-\eta)/(1-2\eta)$ with radius $2(1-\eta)/(1-2\eta)$. In red, is an ellipse with semi-major axis $1+\eta$ and in green an ellipse with elliptical radius $\rho=1+\sqrt{2\eta}$.}
    \label{fig:Ellipses}
\end{figure}

Consequently, $\arccos(L^{-1}(z))$ is analytic throughout the interior of $\mathbf{Ellipse}(\rho=1+\sqrt{2\eta})$. In the notation of Lemma \ref{lem: boundedapproximation}, this means that we may choose $\alpha=\sqrt{2\eta}$.

Finally, by taking $b=3\frac{1-\eta}{1-2\eta}$, target accuracy $\epsilon_{\arccos}/2$, and rescaling the resulting approximation back to the interval $[\eta,1-\eta]$, Lemma~\ref{lem: boundedapproximation} guarantees the existence of a polynomial $q(x)$ such that:
\begin{alignat}{3}
    \|\arccos(x)-q(x) \|_{[\eta,\, 1-\eta]} &\leq \epsilon_{\arccos}/2, \\
    \|q\|_{[0,1]} &\leq \pi/2 + \epsilon_{\arccos}/2,  \\
    \|q\|_{[-1,0] \cup [1,2]} &\leq \epsilon_{\arccos}/2,
\end{alignat}
with an asymptotic  degree:
\begin{equation}
    \mathcal{O}\left[\frac{b}{2\eta}\log \left(\frac{bM}{2\eta \epsilon_{\arccos}/2}\right)\right]\subseteq \mathcal{O}\left[\frac{1}{\eta}\log \left(\frac{1}{\eta \epsilon_{\arccos}}\right)\right] \,,
\end{equation}
\\
where we use the fact that $M=\max_{\rm \textbf{Ellipse}(1+\sqrt{2\eta})} (\arccos(L^{-1}z))\leq \pi/2$ and $b\in \mathcal{O}(1)$ for small $\eta$. 
\\
One problem is that $q(x)$ does not have a definite parity; this is where the choice of $b$ comes in handy. By defining the even polynomial $p_{\rm arccos}(x)=q(x)+q(-x)$ (or $q(x)-q(-x)$ if we want an odd polynomial), we obtain:
\begin{equation}
    \|\arccos(x) - p_{\rm arccos}(x)\|_{[\eta,1-\eta]} \leq \epsilon_{\arccos},
\end{equation}
\begin{equation}
    \|p_{\rm arccos}(x)\| \leq \pi/2+\epsilon_{\arccos}, \quad \forall x\in [-1,1] \; .
\end{equation}

The last step consists of normalising $\tilde{p}(x)$ by $1/(1+\epsilon)$ to ensure the normalisation condition, this does not affect the overall error because:
\begin{equation}
    \left|p_{\rm arccos}(x)-\frac{p_{\rm arccos}(x)}{1+\epsilon_{\arccos}}\right|\leq |p_{\rm arccos}(x)|\left|\frac{\epsilon_{\arccos}}{1+\epsilon_{\arccos}}\right|\leq \left(\frac{\pi}{2}+\epsilon_{\arccos}\right)\epsilon_{\arccos} \in \mathcal{O}(\epsilon_{\arccos}) \; .
\end{equation}
By a final rescaling of $\epsilon_{\arccos}$, this concludes the proof of Lemma \ref{lem: arccosapprox}.
\end{proof}
We also want to construct a polynomial approximation to $\frac{1}{\sqrt{1-x^2}}$:
\begin{lemma}[Polynomial approximation to $\frac{1}{\sqrt{1-x^2}}$, variation on \cite{kang2025quantum}] Let $\epsilon_{\rm inv-sqrt},\eta \in \mathbb{R}^+$ with $0\leq \eta\leq 1/2$. Then there exists an even polynomial  $p_{\rm inv-sqrt}$  such that:
\begin{align}
    \left\|p_{\rm inv-sqrt}(x)-\frac{1}{\sqrt{1-x^2}}\right\|_{[-1+\eta,1-\eta]}&\leq\epsilon_{\rm inv-sqrt} \; ,\\
    \|p_{\rm inv-sqrt}\|_{[-1,1]}&\leq \frac{1}{\sqrt{1-x^2}} \,,
\end{align}
where $p_{\rm inv-sqrt}$ is of asymptotic degree $d_{\rm inv-sqrt}\in\mathcal{O}\left(\frac{1}{\eta}\log\left(\frac{1}{\epsilon_{\rm inv-sqrt}\eta}\right)\right)$.
\label{lem: invsqrtapprox}
\end{lemma}
\begin{proof}
 We first rescale the domain $[-1+\eta, 1-\eta]$ using the linear transformation $L(x)=x/(1-\eta)$. Since the function $1/\sqrt{1-x^2}$ is analytic in a circle of radius $1$ in the complex plane, the function $1/\sqrt{1-L(x)^2}$ will be analytic in a circle of radius $1/(1-\eta)$. Since $1+\eta \leq 1/((1-\eta)$, it means that the function is analytic for $\textbf{Ellipse}(\rho=1+\sqrt{2\eta}) \subseteq \textbf{Ellipse}(a=1+\eta) $.  Applying Theorem 21 from \cite{tang2302cs} and choosing $b=1/(1-\eta)$ produces a polynomial $q(x)$ of asymptotic degree $\mathcal{O}(1/\eta\log(1/\epsilon \eta))$  such that $p_{\rm inv-sqrt}(x)=q(L^{-1}(x))$ satisfies the desired properties. 
\end{proof}

We will also make use of the following Lemma, constructing a polynomial approximation to the trigonometric functions $\sin$ and $\cos$:
\begin{lemma}[Polynomial approximation to trigonometric functions \cite{gilyen2019quantum}, Lemma 57] Let $\epsilon_{\rm trig},\alpha \in \mathbb{R}^+$. Then there exists an even polynomial $p_{\cos,\alpha}$ and an odd polynomial $p_{\sin,\alpha}$ such that:
\begin{equation}
    \|p_{\cos,\alpha}(x)-\cos(\alpha x)\|_{[-1,1]}\leq\epsilon_{\rm trig} \quad \text{and}  \quad \|p_{\sin,\alpha}(x)-\sin(\alpha x)\|_{[-1,1]}\leq \epsilon_{\rm trig} \;,
\end{equation}
where $p_{\sin, \alpha}$ and $p_{\cos, \alpha}$ are both of asymptotic degree $d_{\rm trig}\in \mathcal{O}(\alpha+\log(1/\epsilon_{\rm trig}))$.
\label{lem: trigapprox}
\end{lemma}

\subsection{Hamiltonian QSVT in the odd case (Proof of Theorem \ref{thm: HQSVT with lower bound}, odd case)}
\label{App_subsec: Hamiltonian QSVT in the odd case}
To construct $e^{-iH[p_{\rm SV}(A)]}$, we could try to apply Theorem~\ref{thm: ChebQSP} and seek polynomials $\tilde{p}$ and $\tilde{q}$ approximating $\cos(p(\gamma \arccos(x)))$ and $\frac{\sin(p(\gamma \arccos(x)))}{\sqrt{1-x^2}}$, respectively. However, the normalisation condition \ref{cond:normalization chebyshev} requires that $\tilde{p}$ and $\tilde{q}$ satisfy $|\tilde{p}(x)|^2 + |\tilde{q}(x)|^2(1-x^2) = 1,$ which is, in general, difficult to enforce directly. To avoid this issue, we instead invoke a relaxed variant of Theorem~\ref{thm: ChebQSP}:

\begin{theorem}[Theorem 5 in \cite{gilyen2019quantum}]
Let $\tilde{P}, \tilde{Q} \in \mathbb{R}[x]$. There exist $P, Q \in \mathbb{C}[x]$ satisfying the degree condition~\ref{cond:max-degree chebyshev}, the normalization condition~\ref{cond:normalization chebyshev}, and the parity condition~\ref{cond:parity} of the Chebyshev QSP (Theorem~\ref{thm: ChebQSP}), with $\tilde{P} = \mathfrak{Re}[P]$ and $\tilde{Q} = \mathfrak{Re}[Q]$, if and only if $\tilde{P}$ and $\tilde{Q}$ satisfy the degree and parity conditions and
\begin{equation}
    |\tilde{P}(x)|^2 + |\tilde{Q}(x)|^2(1-x^2) \leq 1.
\end{equation}
\label{thm: ChebQSPrelaxed}
\end{theorem}

Therefore, it suffices to find real polynomials $\tilde{p}$ and $\tilde{q}$ that approximate the desired functions and that approximately satisfy the normalization condition, that is, $|\tilde{p}(x)|^2 + |\tilde{q}(x)|^2(1-x^2) \leq 1$, a condition that is significantly easier to achieve. After normalisation of $\tilde{p}$ and $\tilde{q}$, Theorem~\ref{thm: ChebQSPrelaxed} then guarantees the existence of complex polynomials $P$ and $Q$ that can be implemented via QSP. Furthermore, since $\tilde{p}$ and $\tilde{q}$ nearly satisfy the normalisation condition  \ref{cond:normalization chebyshev}, their imaginary parts are of order $\mathcal{O}(\epsilon)$ and therefore contribute negligibly to the overall error of the QSP protocol. The following lemma provides explicit polynomials satisfying these requirements:
\begin{lemma}[Dominated polynomial approximation odd case] 
Let $\epsilon,\eta,\gamma \in \mathbb{R}^+$ such that $0\leq\eta\leq 1/2$. Let $p(x)$ be a degree $d$ odd polynomial. Then, there exists degree $\Tilde{d}$ polynomials $\tilde{p}(x)$ and $\tilde{q}(x)$ both in $\mathbb{R}[x]$ such that:
\begin{enumerate}[label=\textcolor{wineRed}{(\arabic*)}]
    \item $\|\tilde{p}(x)-\cos(p(\gamma\arccos(x)))\|_{[\eta,1-\eta]}\leq\epsilon \quad \text{and} \quad \left\|\tilde{q}(x)-\frac{\sin(p(\gamma\arccos(x))}{\sqrt{1-x^2}}\right\|_{[\eta,1-\eta]}\leq\epsilon \; ,$\label{cond: approxcondition}
    \item $\tilde{p}(x)$ is even and $\tilde{q}(x)$ is odd \label{cond: paritypqcondition},
    \item $ |\tilde{p}(x)|^2+|\tilde{q}(x)|^2(1-x^2)\leq1 \quad \forall x\in[-1,1] \; ,$ 
    \label{cond: subnormalisation}
    \item $ |\tilde{p}(x)|^2+|\tilde{q}(x)|^2(1-x^2)\geq 1-\epsilon \quad \forall x\in[\eta,1-\eta]\; .$ 
    \label{cond: dominatedapprox}
\end{enumerate}
Furthermore, $\Tilde{d}$ has an asymptotic scaling of $\mathcal{O}\left(\frac{d}{\eta}\log\left(\frac{1}{\eta\epsilon}\right)\log\left(\frac{d\gamma}{\eta\epsilon}\right)\right)$. Condition \ref{cond: dominatedapprox} is what we refer to the dominated approximation condition. 
\label{lem: dominatedapproxodd}
\end{lemma}
\begin{proof}
 To construct our polynomials, we will make use of Lemma \ref{lem: arccosapprox} to build an odd polynomial $p_{\rm arccos}\approx \arccos(x)$. We will define:
\begin{equation}
    p_{\rm polyarc}(x)=p(\gamma p_{\arccos}(x)) \; ,
\end{equation}
and $\epsilon_{\rm polyarc}$, the error associated with the target function $p(\gamma \arccos(x))$. We will then invoke multiple lemmas to construct a polynomial approximation to $p_{\sin}\approx \sin(x),p_{\cos}\approx\cos(x)$, $p_{\rm inv-sqrt}\approx \frac{1}{\sqrt{1-x^2}}$. Our final polynomials will be:
\begin{equation}
    \tilde{p}(x) = p_{\cos,\alpha}\left(\frac{1}{\alpha}p_{\rm polyarc}(x)\right) \quad \text{and} \quad \tilde{q}(x)=p_{\sin,\alpha}\left(\frac{1}{\alpha}p_{\rm polyarc}(x)\right)p_{\rm inv-sqrt}(x),
    \label{eq: dominnatedpolynomialapprox}
\end{equation}
where $\alpha = ||p||_{\max}$, is there to ensure normalization. We will show conditions $\ref{cond: approxcondition}$-$\ref{cond: dominatedapprox}$ in order:

\vspace{2ex}

\noindent \textbf{\ref{cond: approxcondition} Approximation condition}

\vspace{6pt}

With the polynomials of Equation \eqref{eq: dominnatedpolynomialapprox}, we want that they approximate well the desired functions:
\begin{align}
\nonumber
    \|\tilde{p}(x)-\cos(p(\gamma\arccos(x)))\|_{[\eta,1-\eta]}&=\left\|p_{\cos,\alpha}\left(\frac{1}{\alpha}p_{\rm polyarc}(x)\right)-\cos(p(\gamma\arccos(x)))\right\|_{[\eta,1-\eta]}\\ \nonumber
    &
\leq \left\|p_{\cos,\alpha}\left(\frac{1}{\alpha}p_{\rm polyarc}(x)\right)-\cos(p_{\rm polyarc}(x)) \right\|_{[\eta,1-\eta]}\\ \nonumber
&\qquad \qquad+\left\|\cos\left(p_{\rm polyarc}(x)\right)-\cos(p_{\rm polyarc}(x)) \right\|_{[\eta,1-\eta]}
    \\
    &\leq \epsilon_{\rm trig}+\epsilon_{\rm polyarc}
\end{align}
and 
\begin{align}
\nonumber
    \left\|\tilde{q}(x)-\frac{\sin(p(\gamma\arccos(x)))}{\sqrt{1-x^2}}\right\|_{[\eta,1-\eta]}&=\left\|p_{\sin,\alpha}\left(\frac{1}{\alpha}p_{\rm polyarc}(x)\right)p_{\rm inv-sqrt}(x)-\frac{\sin(p(\gamma\arccos(x)))}{\sqrt{1-x^2}}\right\|_{[\eta,1-\eta]}\\
    \nonumber
    &\leq \left\|p_{\sin,\alpha}\left(\frac{1}{\alpha}p_{\rm polyarc}(x)\right)p_{\rm inv-sqrt}(x)-\frac{p_{\sin,\alpha}\left(\frac{1}{\alpha}p_{\rm polyarc}(x)\right)}{\sqrt{1-x^2}}\right\|_{[\eta,1-\eta]}\\ \nonumber
\quad &\qquad+\left\|\frac{p_{\sin,\alpha}\left(\frac{1}{\alpha}p_{\rm polyarc}(x)\right)}{\sqrt{1-x^2}}-\frac{\sin(p(\gamma\arccos(x)))}{\sqrt{1-x^2}}\right\|_{[\eta,1-\eta]}\\
\nonumber
    & \leq (1+\epsilon_{\rm trig})\epsilon_{\rm inv-sqrt}+\left\|\frac{1}{\sqrt{1-x^2}}\right\|_{[\eta,1-\eta]}(\epsilon_{\rm trig}+\epsilon_{\rm polyarc})\\
    &\leq (1+\epsilon_{\rm trig})\epsilon_{\rm inv-sqrt}+\frac{1}{\eta^{1/2}}(\epsilon_{\rm trig}+\epsilon_{\rm polyarc}) \; .
\end{align}
Choosing $\epsilon_{\rm inv-sqrt}\in \mathcal{O}(\epsilon)$ and $\epsilon_{\rm trig}, \epsilon_{\rm polyarc}\in \mathcal{O}(\eta^{1/2}\epsilon)$ yields the desired bound. 
For $\epsilon_{\rm polyarc}$ we have:
\begin{align}
    \left|\left|p(\gamma p_{\arccos}(\cdot))-p(\gamma\arccos(\cdot))\right|\right|_{[\eta,1-\eta]}&\leq ||p'||_{[\delta,\Lambda]}||\gamma p_{\arccos}(\cdot)-\gamma\arccos(\cdot)||_{[\eta,1-\eta]}\\
    &\leq d^2\gamma||p||_{[\delta,\Lambda]}\epsilon_{\arccos}
    \label{eq: errorpolyarcos}
\end{align}
Here we used the fact that $||f(g(*))-f(\tilde{g}(*))||_{D}\leq ||f'||_{g(D)}||g-\tilde{g}||_{D}$, as well as the Markov brothers' inequality:
\begin{equation}
    ||p^{(j)}||_D\leq \frac{d^{2j}}{(2j-1)!!}||p||_{D} \,,
    \label{eq: Markovinequality}
\end{equation}
for $p$ a polynomial and $D$ its domain. Thus, choosing $ \epsilon_{\arccos}=\eta^{1/2}\epsilon/(d^2\gamma||p||_{[\delta,\Lambda]})$ yields the desired error.

\vspace{2ex}

\noindent \textbf{\ref{cond: paritypqcondition} Parity condition}

\vspace{6pt}

We know that $\tilde{p}$ is even because $p_{\cos,\alpha}$ is even, and we know that $\tilde{q}$ is odd because $p_{\sin,\alpha}$ is odd and $p_{\rm polyarc}$ is odd because both $p$ and $p_{\arccos}$ are odd. Finally, multiplying by an even polynomial $p_{\rm inv-sqrt}$ does not affect oddness. 

For the sub-normalization condition \ref{cond: subnormalisation}, we have:
\begin{align}
    |\tilde{p}(x)|^2+|\tilde{q}(x)|^2(1-x^2)&= |p_{\cos}(\frac{1}{\alpha}p_{\rm polyarc}(x))|^2+|p_{\sin}(\frac{1}{\alpha}p_{\rm polyarc}(x))|^2|p_{\rm inv-sqrt}(x)|^2(1-x^2)  \nonumber \\
    & \leq |p_{\cos}(p_{\rm polyarc}(x))|^2+|p_{\sin}(p_{\rm polyarc}(x))|^2 \nonumber\\
    &\leq |\cos(\frac{1}{\alpha}p_{\rm polyarc}(x))|^2+|\sin(\frac{1}{\alpha}p_{\rm polyarc}(x))|^2+\mathcal{O}(\epsilon_{\rm trig}) \label{eq: dominateapproxproof4}\\
    &\leq 1+2\sqrt{2}\epsilon_{\rm trig}+2\epsilon_{\rm trig}^2\leq 1+2(1+\sqrt{2})\epsilon_{\rm trig} \nonumber\\
    &\leq 1+\mathcal{O}(\epsilon) \nonumber \; ,
\end{align}
which becomes $1+\epsilon$ after rescaling. The second line follows because $|p_{\rm inv-sqrt}(x)|\leq 1/\sqrt{1-x^2}$ for $x\in[-1,1]$. The third line follows because $\frac{1}{\alpha}p_{\rm polyarc}$ is bounded by $1$ and $|p_{\cos}(x)|\leq |\cos(x)|+\epsilon$ and similarly for $\sin$. 
\vspace{2ex}

\noindent \textbf{\ref{cond: subnormalisation} Subnormalization condition}

\vspace{6pt}

We need $|\tilde{p}(x)|^2+|\tilde{q}(x)|^2(1-x^2)\leq 1$: we have to rescale the polynomials by $1/\sqrt{1+\epsilon}$:
\begin{equation}
    \tilde{p}(x)\leftarrow \frac{\tilde{p}(x)}{\sqrt{1+\epsilon}}, \quad \tilde{q}(x)\leftarrow \frac{\tilde{q}(x)}{\sqrt{1+\epsilon}}\; .
\end{equation}
From calculation \eqref{eq: dominateapproxproof4}, it is easy to see that for $x\in [-1+\eta,1-\eta]$, $|\tilde{p}(x)|\leq \sqrt{1+\epsilon}$ and $|\tilde{q}(x)|\leq \frac{\sqrt{1+\epsilon}}{\sqrt{1-x^2}}\in \mathcal{O}(\frac{\sqrt{1+\epsilon}}{\eta^{1/2}})$. Thus, rescaling the polynomials does not significantly increase the error since:
\begin{equation}
    \left|\frac{\tilde{p}(x)}{\sqrt{1+\epsilon}}-\tilde{p}(x)\right|\leq \left|\frac{1}{\sqrt{1+\epsilon}}-1\right||\tilde{p}(x)|\leq \sqrt{1-\epsilon}-1\leq \epsilon/2 \; ,
\end{equation}

\begin{equation}
    \left\|\frac{\tilde{q}(x)}{\sqrt{1+\epsilon}}-\tilde{q}(x)\right\|\leq \left\|\frac{1}{\sqrt{1+\epsilon}}-1\right\|\|\tilde{q}(x)\|\leq (\sqrt{1-\epsilon}-1)(1+\epsilon)\eta^{-1/2}\leq \frac{\epsilon}{2\eta^{1/2}} \; .
\end{equation}
Rescaling $\epsilon$ to $\mathcal{O}(\epsilon \eta^{1/2})$  thus leads to the same error and the same asymptotic scaling of the polynomial degree. 

\vspace{2ex}

\noindent \textbf{\ref{cond: dominatedapprox} Dominated approximation condition}

\vspace{6pt}

Here, we use a similar calculation to \eqref{eq: dominateapproxproof4}, but this time we use the fact that for $x \in [\eta,1-\eta]$, $|p_{\rm inv-sqrt}(x)|\geq \frac{1}{\sqrt{1-x^2}}-\epsilon$,  $|p_{\rm \cos,\alpha}(x)|\geq |\cos(\alpha x)|-\epsilon$ and  $|p_{\rm \sin,\alpha}(x)|\geq |\sin(\alpha x)|-\epsilon$:
\begin{align}
    |\tilde{p}(x)|^2+|\tilde{q}(x)|^2(1-x^2)&= |p_{\cos}(\frac{1}{\alpha}p_{\rm polyarc}(x))|^2+|p_{\sin}(\frac{1}{\alpha}p_{\rm polyarc}(x))|^2|p_{\rm inv-sqrt}(x)|^2(1-x^2) \nonumber\\
    & \geq ||\cos(p_{\rm polyarc}(x))|-\epsilon|^2+||\sin(p_{\rm polyarc}(x))|-\epsilon|^2\left(\frac{1}{\sqrt{1-x^2}}-\epsilon\right)^2(1-x^2) \nonumber \\
    &\geq |\cos(p_{\rm polyarc}(x))|^2-\epsilon+\left(|\sin(p_{\rm polyarc}(x))|^2-\epsilon\right)\left(1-\epsilon\sqrt{1-x^2}\right) \nonumber  \\
    &\geq |\cos(p_{\rm polyarc}(x))|^2+|\sin(p_{\rm polyarc}(x))|^2-\mathcal{O}(\epsilon)=1-\mathcal{O}(\epsilon) \; .
\end{align}
Our approximation thus satisfies the condition \ref{cond: dominatedapprox} up to a final rescaling.
The degree of $\tilde{p}$ and $\tilde{q}$, is $d_{\rm trig}*d*d_{\arccos}$ and $d_{\rm trig}*d*d_{\arccos}+d_{\rm inv-sqrt}$ respectively, behaving like; 
\begin{align}
    deg(\tilde{p})&\in \mathcal{O}\left[\left(\alpha+\log\left(\frac{1}{\epsilon_{\rm trig}}\right)\right)\left(\frac{d}{\eta}\log\left(\frac{1}{\epsilon_{\rm polyarc}\eta}\right)\right)\right]\subseteq \mathcal{O}\left[\frac{d}{\eta}\log\left(\frac{1}{\epsilon\eta}\right)\log\left(\frac{\gamma d}{\epsilon\eta}\right)\right] \; ,\\
    deg(\tilde{q})&\in \mathcal{O}\left[\left(\alpha+\log\left(\frac{1}{\epsilon_{\rm trig}}\right)\right)\left(\frac{d}{\eta}\log\left(\frac{1}{\epsilon_{\rm polyarc}\eta}\right)\right)+\frac{1}{\eta}\log\left(\frac{1}{\epsilon_{\rm inv-sqrt}}\eta\right)\right] \nonumber\\
    &\subseteq \mathcal{O}\left[\frac{d}{\eta}\log\left(\frac{1}{\epsilon\eta}\right)\log\left(\frac{d\gamma}{\epsilon\eta}\right)+\frac{1}{\eta}\log\left(\frac{1}{\epsilon\eta}\right)\right] \,,
\end{align}
which are both in $\mathcal{O}\left[\frac{d}{\eta}\log\left(\frac{1}{\epsilon\eta}\right)\log(\frac{d\gamma}{\epsilon\eta})\right]$.
\end{proof}

Finally, as discussed earlier in the main text, for the approximation of $\arccos$ to work, we 
require $x \in [\eta, 1-\eta]$; therefore, we  have to rescale the matrix $A$ by a constant $\gamma$ so that $\cos(A/\gamma)$ remains in the interval $[\eta,1-\eta]$. 
Choosing such a constant is the topic of the following lemma:

\begin{lemma}Let $\delta, \Lambda \in \mathbb{R}^{>0}$ such that $\delta < \Lambda$. Then there exists a constant $\gamma$ such that for all $\Sigma\in[\delta,\Lambda]$ we have:
\label{lem: rescalingcos}
\begin{equation}
    \eta \leq \cos(\Sigma/\gamma) \leq 1-\eta
\end{equation}
for $\eta\in\mathcal{O}(\delta^2/\Lambda^2)$.

\end{lemma}
\begin{proof}
    Choosing $\gamma=\frac{\Lambda}{\frac{\pi}{2}+(\frac{\pi\delta}{2\Lambda})^2}$ leads to the following:
\begin{align}
     \frac{\pi}{2}\frac{\delta}{\Lambda}+\frac{\pi^2}{4}\left(\frac{\delta}{\Lambda}\right)^3& \leq \Sigma/\gamma \leq\frac{\pi}{2}+\left(\frac{\pi\delta}{2\Lambda}\right)^2\\
      \frac{\pi}{2}\frac{\delta}{\Lambda}& \leq \Sigma/\gamma \leq\frac{\pi}{2}+\left(\frac{\pi\delta}{2\Lambda}\right)^2.
\end{align}
Now, from $0$ to $\pi$, $\cos(x)$ is a decreasing function, which means that:
\begin{align}
    \cos\left(\frac{\pi}{2}+\left(\frac{\pi\delta}{2\Lambda} \right)^2 \right)&\leq \cos(\Sigma/\gamma) \leq  \cos\left( \frac{\pi\delta}{2\Lambda}\right)
    \\
    \frac{11}{24}\left(\frac{\pi\delta}{2\Lambda} \right)^2&\leq \cos(\Sigma/\gamma) \leq  1-\frac{11}{24}\left( \frac{\pi\delta}{2\Lambda}\right)^2 \; .
\end{align}
Thus, taking $\eta=\frac{11}{24}\left(\frac{\pi\delta}{2\Lambda} \right)^2\in \mathcal{O}(\delta^2/\Lambda^2)$ yields the desired property.
\end{proof}

With Lemmas \ref{lem: dominatedapproxodd} and \ref{lem: rescalingcos} in the pocket, we have polynomials $\tilde{p}(x)$ and $\tilde{q}(x)$ satisfying the dominated approximation condition \ref{cond: dominatedapprox}. We are now ready to prove Theorem \ref{thm: HQSVT with lower bound}.

\vspace{4ex}

\noindent \textbf{Proof of Theorem \ref{thm: HQSVT with lower bound} in the odd case.}

We first employ Lemma \ref{lem: dominatedapproxodd} to build both $\tilde{p}$ and $\tilde{q}$ satisfying the four conditions of the dominated polynomial approximation of $\cos(p(\gamma\arccos(x)))$ and $\frac{\sin(p(\gamma \arccos(x)))}{\sqrt{1-x^2}}$ respectively. 
Invoking Theorem \ref{thm: ChebQSPrelaxed} we obtain complex polynomials $P(x)$ and $Q(x)$ such that $\tilde{q}(x)=\mathfrak{Re}(Q(x))$ and $\tilde{p}(x)=\mathfrak{Re}(P(x))$ both achievable by QSP and thus satisfying the normalization condition:
\begin{align}
    &|P(x)|^2+|Q(x)|^2(1-x^2)= 1\nonumber\\
    &\quad \mathfrak{Re}\left[|P(x)|^2+|Q(x)|^2(1-x^2)\right]+\mathfrak{Im}\left[|P(x)|^2+|Q(x)|^2(1-x^2)\right]= 1 \nonumber \\
    &\quad \mathfrak{Im}\left[|P(x)|^2+|Q(x)|^2(1-x^2)\right]= 1-\underbrace{\mathfrak{Re}\left[|P(x)|^2+|Q(x)|^2(1-x^2)\right]}_{\ref{cond: dominatedapprox}:|\tilde{p}(x)|^2+|\tilde{q}(x)|^2(1-x^2)\geq 1-\epsilon, } \nonumber \\
    &\quad \mathfrak{Im}\left[|P(x)|^2+|Q(x)|^2(1-x^2)\right]\leq \epsilon \quad \forall x\in[\eta,1-\eta]  \; .
\end{align}
This leads to an overall error $\epsilon_{\rm tot}$:
\begin{align}
    \epsilon_{\rm tot}&=\left\|e^{-i\theta_{\tilde{d}}(2\Pi_{0} - I)}
        \prod_{k=\tilde{d}-1}^{0}
        e^{-iH[A]/\gamma}e^{-i\phi_k(2\Pi_{0} - I)}-e^{-iH[p_{\rm SV}(A)]}\right\|_{[\delta,\Lambda]}\nonumber\\
    &=\left\|\begin{bmatrix}
        P(\cos(\frac{x}{\gamma})) & -iQ(\cos(\frac{x}{\gamma}))\sin(\frac{x}{\gamma})\nonumber\\
        -iQ^*(\cos(\frac{x}{\gamma}))\sin(\frac{x}{\gamma})& P^*(\cos(\frac{x}{\gamma})) 
    \end{bmatrix}-\begin{bmatrix}
        \cos(p(x)) & -i\sin(p(x))\nonumber\\
        -i\sin(p(x)) & \cos(p(x))
    \end{bmatrix}\right\|_{[\delta,\Lambda]}
    \nonumber\\
    &\leq \left\|\begin{bmatrix}
        \tilde{p}(x)-\cos(p(x)) & -i\tilde{q}(x)\sqrt{1-x^2}+i\sin(p(x))\\
        -i\tilde{q}(x)\sqrt{1-x^2}+i\sin(p(x))& \tilde{p}(x)-\cos(p(x)) 
\end{bmatrix}\right\|_{[\delta,\Lambda]} \nonumber\\
    & \qquad \qquad \qquad \qquad \qquad \qquad+\left\|\begin{bmatrix}
        \mathfrak{Im}[P(x)] & i\mathfrak{Im}[Q(x)]\sqrt{1-x^2}\\
        -i\mathfrak{Im}[Q(x)]\sqrt{1-x^2}& -\mathfrak{Im}[P(x)] 
\end{bmatrix}\right\|_{[\delta,\Lambda]}\nonumber\\
    &\leq \sqrt{|\tilde{p}(x)-\cos(p(x))|^2+|\tilde{q}(x)\sqrt{1-x^2}-\sin(p(x))|^2}+ \sqrt{\mathfrak{Im}[|P(x)|^2+|Q(x)|^2(1-x^2)]}\nonumber\\
    &\leq \mathcal{O}(\sqrt{\epsilon}) \; ,
\end{align}
where in the last line, we used the fact that $\|\vec{\alpha}\cdot \vec{\sigma}-\vec{\beta}\cdot \vec{\sigma}\|\leq \|\vec{\alpha}-\vec{\beta}\|$, for $\vec{\sigma}=\{I,X,Y,Z\}$.

Finally, from Lemma \ref{lem: rescalingcos}, we choose $\gamma\in\mathcal{O}(\Lambda)$ and $\eta\in \mathcal{O}(\frac{\Lambda^2}{\delta^2})$ such that both $\tilde{q}$ and $\tilde{p}$ are of asymptotic degree:
\begin{equation}
    \mathcal{O}\left[d\frac{\Lambda^2}{\delta^2}\log\left(\frac{\Lambda}{\epsilon\delta}\right)\log\left(\frac{d\Lambda}{\epsilon \delta}\right)\right] \; .
\end{equation}
Up to a final rescaling of $\epsilon$, this concludes the proof.\qed

\vspace{1ex}

\begin{subsection}{Hamiltonian QSVT in the general case (Proof of Theorem \ref{thm: HQSVT with lower bound}, general case)}
\label{App_subsec: Hamiltonian QSVT in the general case}
In the previous subsection, we constructed an approximation of $p(\gamma\arccos(x))$ by first approximating $\arccos(x)$ with a polynomial $p_{\arccos}$ and then composing it with $p$. For even polynomials $p$, notice that both $ \cos(p(\gamma\arccos(x))) $ and
$ \frac{\sin(p(\gamma\arccos(x)))}{\sqrt{1-x^2}} $
are even functions. This is somewhat problematic in the standard QSVT setting, where $P$ and $Q$ must have different parity. Indeed, since our construction implements $p(\gamma\arccos(x))$ through polynomial composition, the resulting approximations inherit the parity of $p$.

However, because we only require accuracy on the restricted interval $[\eta,1-\eta]$, there is no need to explicitly construct an approximation of $\arccos(x)$ and then compose it with $p$. Instead, one may directly seek a polynomial approximation $p_{\rm polyarc}$ to the composite function $p(\gamma\arccos(x))$ itself. This technique also comes from \cite{kang2025quantum}, where they use it to improve the scaling of the approximation to $p(\arcsin(x))$ in terms of $\epsilon$. This approach avoids the parity issues associated with the intermediate approximation and may lead to lower-degree polynomial approximants on the interval of interest. We emphasize that this technique is not limited to even polynomials. Since the approximation is performed directly on the composite function, it can be used to implement arbitrary polynomials, including odd polynomials, even polynomials, and polynomials with no definite parity. 

Just like in \cite{kang2025quantum}, for odd polynomials, we show that this approach improves the scaling with respect to $\epsilon$ compared to the odd case of Theorem \ref{thm: HQSVT with lower bound}. However, this gain comes at a cost: the resulting approximation exhibits a less favourable dependence on the parameters $\delta$ and $\Lambda$ than the construction obtained by separately approximating $\arccos(x)$ and composing with $p$.

\begin{lemma}[Polynomial approximation to composite functions] 
Let $\epsilon_{\arccos},\eta \in \mathbb{R}$ such that $0\leq\eta\leq 1/2$.Let $p(x)\in\mathbb{R}[x]$ be a degree $d$ polynomial. Then, there exists a degree $\Tilde{d}$ polynomial $p_{\rm polyarc}(x)$ such that:
\begin{enumerate}[label=\textcolor{wineRed}{(\arabic*)}]
    \item $|p_{\rm arccos}(x)-p(\gamma\arccos(x))|\leq\epsilon_{\rm polyarc}, \quad \forall x\in [\eta,1-\eta]$,
    \item $p_{\arccos}(x)$ has a definite parity, 
    \item $|p_{\arccos}(x)|\leq |p(\gamma\arccos(x))|+\epsilon_{\rm polyarc}, \quad \forall x\in[-1,1] $ \; .
\end{enumerate}
Furthermore, $d_{\rm polyarc}$ has an asymptotic scaling of $\mathcal{O}(d\gamma^{1/2}\eta^{-5/4}+\eta^{-1}\log(1/\eta\epsilon_{\rm polyarc}))$. 
\label{lem: parccosapprox}
\end{lemma}
\begin{proof}

We follow the same steps as the proof of Lemma \ref{lem: arccosapprox}. This yields an even polynomial $p_{\rm polyarc}(x)$, such that:   
\begin{alignat}{3}
    |p_{\rm polyarc}(x) - p(\gamma\arccos(x))| &\leq \epsilon_{\rm polyarc}, \qquad & &\forall x \in [\eta,\, 1-\eta], \\
    |p_{\rm polyarc}(x)| &\leq |p(\gamma\arccos(x))| + \epsilon_{\rm polyarc}, \qquad & &\forall x \in [-1,1] \; ,
\end{alignat}
with an asymptotic  degree:
\begin{equation}
    \mathcal{O}\left(\frac{1}{\eta}\log \left(\frac{M}{\eta \epsilon_{\rm polyarc}}\right)\right).
    \label{eq:scalingpolynomialarccos}
\end{equation}
But now, $M=\max_{\textbf{Ellipse}(1+\sqrt{2\eta})}(p(\gamma\arccos(L^{-1}(z))))$, the maximum value of the composite function $p(\gamma\arccos(L^{-1}(z)))$ in $\textbf{Ellipse}(1+\sqrt{2\eta})$. To find an upper bound on $M$ we use similar lines from \cite{kang2025quantum} in the proof of Propostion 13. Using property \eqref{eq:EllipseStadiumContainment} we have:
\begin{equation}
    \textbf{Ellipse}(\rho=1+\sqrt{2\eta})\subseteq  \textbf{Stadium}\left(\left[\frac{-1}{1+\sqrt{2\eta}},\frac{1}{1+\sqrt{2\eta}}\right],\frac{\eta+\sqrt{2\eta}}{1+\sqrt{2\eta}}\right)\equiv \textbf{S}(\eta) \; .
\end{equation}
Finding the maximum value on $\textbf{S}(\eta)$ will thus upper bound the maximum value on $\textbf{Ellipse}(\rho=1+\sqrt{2\eta})$. We will make use of the following lemma:

\begin{lemma} Let $\eta\in [0,1/2]$. Then
    \begin{align}
    \arccos(y^{-1}(\textbf{S}(\eta))) 
    &\subseteq \textbf{Stadium}\left([0,\pi/2], \eta^{-1/2}\right) \; .
\end{align}
\end{lemma}

\begin{proof}
    First, notice that $\left[\frac{-1}{1+\sqrt{2\eta}},\frac{1}{1+\sqrt{2\eta}}\right]\subseteq [-1,1]$. Applying the function $g=\arccos(L^{-1}(\cdot))$  thus yields:
\begin{equation}
    g\left(\left[\frac{-1}{1+\sqrt{2\eta}},\frac{1}{1+\sqrt{2\eta}}\right]\right)\subseteq g([-1,1])\subseteq [0,\pi/2] \,.
\end{equation}
For the rest of the points $w \in \textbf{S}(\eta) $, we have to show that, after applying the transformation, the nearest point $x$ in the real interval $[0,\pi/2]$ is closer than $\eta^{-1/2}$: 
\begin{align}
    |\arccos(y^{-1}(w))-\arccos(y^{-1}(x))|&\leq \max_{z\in\textbf{S}(\eta)}|(\arccos'(y^{-1}(z)))|*|w-x|\nonumber\\
    &\leq \max_{z\in\textbf{S}(\eta)}\left|\frac{(y^{-1}(z))'}{\sqrt{1-(y^{-1}(z))^2}}\right|\left(\frac{\eta+\sqrt{2\eta}}{1+\sqrt{2\eta}}\right)\nonumber\\
    &= \max_{z\in\textbf{S}(\eta)}\left(\frac{1}{\sqrt{\left|1-(\frac{1}{2}+\frac{1/2-\eta}{1-\eta}z)^2\right|}}\frac{1/2-\eta}{1-\eta}\right)\left(\frac{\eta+\sqrt{2\eta}}{1+\sqrt{2\eta}}\right) \; .
\end{align}
This is maximized whenever $(\frac{1}{2}+\frac{1/2-\eta}{1-\eta}z)^2$ approaches $1$, which occurs for $z=\frac{1+\sqrt{2\eta}+\eta}{1+\sqrt{2\eta}}$, the maximum value of $z$ on the real line. We therefore have:
\begin{align}
    |\arccos(y^{-1}(w))-\arccos(y^{-1}(x))|&\leq  \left(\frac{1}{\sqrt{\left|1-(\frac{1}{2}+\frac{1/2-\eta}{1-\eta}\frac{1+\sqrt{2\eta}+\eta}{1+\sqrt{2\eta}})^2\right|}}\frac{1/2-\eta}{1-\eta}\right)\left(\frac{\eta+\sqrt{2\eta}}{1+\sqrt{2\eta}}\right)\\
    &\leq \frac{\eta^{-1/2}}{2} \quad \text{for }\eta\in[0,1/2].
\end{align}
\end{proof}
To bound a polynomial in this stadium, we use a slight variation of Proposition 37 from \cite{kang2025quantum}: 
\begin{lemma}[\cite{kang2025quantum}]
    Let $p$ be a real polynomial of degree $d$. For any $l\geq  0$ we have:
    \begin{equation}
        ||p(\gamma x)||_{\max, \textbf{ Stadium}([0,\pi/2],l)}\leq e^{\sqrt{2l\gamma}d}||p||_{\max, [0,\gamma\pi/2]} \; .
    \end{equation}
\end{lemma}
\begin{proof}
   Using the same proof technique as Proposition~37 of \cite{kang2025quantum}, we obtain an analogous bound after accounting for the rescaling $x \mapsto \gamma x$. Indeed, by the chain rule, $\frac{d^j}{dx^j}p(\gamma x)=\gamma^j p^{(j)}(\gamma x)$, and therefore$\left\|(p(\gamma\cdot))^{(j)}\right\|=\gamma^j \left\|p^{(j)}(\gamma\cdot)\right\|$.
Compared to the proof of \cite{kang2025quantum}'s Proposition~37, this introduces an additional multiplicative factor of $\gamma^j$ in each term of the relevant summation. Proceeding exactly as in \cite{kang2025quantum}, these factors can be absorbed into the final estimate, yielding the bound
\[
\exp\!\bigl(\sqrt{2\ell \gamma}\, d\bigr)\,
\|p\|_{\max,[0,\pi/2]}.
\]
Thus, the only effect of the rescaling is to replace the factor
$\exp(\sqrt{2\ell}\, d)$ appearing in Proposition~37 of \cite{kang2025quantum} by $\exp\!\bigl(\sqrt{2\ell \gamma}\, d\bigr)$.

\end{proof}
To summarise, we have:
\begin{align}
    M=||p(\gamma\arccos(L^{-1}(x)||_{\max,\textbf{Ellipse}(\rho=1+\sqrt{2\eta})}&\leq ||p(\gamma\arccos(L^{-1}(x))||_{\max,\textbf{S}(\eta)} \nonumber\\
    &\leq ||p(\gamma x)||_{\max, \textbf{Stadium}([0,\gamma\pi/2],\eta^{-1/2})}\nonumber\\
    &\leq e^{\eta^{-1/4}\gamma^{1/2}d}||p||_{\max,[0,\gamma\pi/2]} \; . 
\end{align}
Plugging this expression into \eqref{eq:scalingpolynomialarccos}  yields:
\begin{align}
    \tilde{d}\in \mathcal{O}\left[\frac{1}{\eta}\log\left(\frac{e^{\eta^{-1/4}\gamma^{1/2}d}||p||_{\max,[0,\gamma\pi/2]}}{\eta \epsilon}\right)\right]=\mathcal{O}\left[\frac{\gamma^{1/2}d}{\eta^{5/4}}+\frac{1}{\eta}\log\left(\frac{||p||_{\max,[0,\gamma\pi/2]}}{\eta \epsilon}\right)\right] \,,
\end{align}
where we used the fact that $\alpha\leq 1$ by assumption. 
\end{proof}

Finally, with Lemma \ref{lem: parccosapprox} it is now possible to do Hamiltonian QSVT for arbitrary polynomials:
\begin{lemma}[Dominated polynomial approximation, general case] 
Let $\epsilon,\eta,\gamma \in \mathbb{R}^+$ be such that $0\leq\eta\leq 1/2$. Let $p(x)\in \mathbb{R}[x]$ be a degree $d$ polynomial. Then, there exist degree $\Tilde{d}$ polynomials $\tilde{p}(x)$ and $\tilde{q}(x)$ such that:
\begin{enumerate}[label=\textcolor{wineRed}{(\arabic*)}]
    \item $|\tilde{p}(x)-\cos(p(\gamma\arccos(x)))|\leq\epsilon \quad {\rm and} \quad \left|\tilde{q}(x)-\frac{\sin(p(\gamma\arccos(x))}{\sqrt{1-x^2}}\right|\leq\epsilon\quad \forall x\in [\eta,1-\eta]$ ,
    \item $\tilde{p}(x)$ is even and $\tilde{q}(x)$ is odd.
    \item $ |\tilde{p}(x)|^2+|\tilde{q}(x)|^2(1-x^2)\leq1 \quad \forall x\in[-1,1]\;$ 
    \item $ |\tilde{p}(x)|^2+|\tilde{q}(x)|^2(1-x^2)\geq 1-\epsilon\quad \forall x\in[\eta,1-\eta] \; .$
\end{enumerate}
Furthermore, $\Tilde{d}$ has an asymptotic scaling of $\mathcal{O}\left(\frac{d\gamma^{1/2}}{\eta^{5/4}}\log\left(\frac{1}{\eta\epsilon}\right)+\frac{1}{\eta}\log^2\left(\frac{1}{\eta\epsilon}\right)\right)$. 
\label{lem: dominatedapproxgeneral}
\end{lemma}
\begin{proof}
The proof is exactly the same as in Lemma \ref{lem: dominatedapproxodd}, but instead we construct the polynomials $\tilde{p}$ and $\tilde{q}$ from equation \eqref{eq: dominnatedpolynomialapprox} with $p_{\rm polyarc}$ given by Lemma \ref{lem: parccosapprox}. The scaling of the degree of theses polynomial will be:
\begin{align}
    \rm deg(\tilde{p})&\in \mathcal{O}\left[\left(\alpha+\log\left(\frac{1}{\epsilon_{\rm trig}}\right)\right)\left(\frac{d\gamma^{1/2}}{\eta^{5/4}}+\log\left(\frac{1}{\epsilon_{\rm polyarc}\eta}\right)\right)\right] \nonumber\\
    &\subseteq \mathcal{O}\left[\frac{d\gamma^{1/2}}{\eta^{5/4}}\log\left(\frac{1}{\epsilon\eta}\right)+\frac{1}{\eta}\log^2\left(\frac{1}{\epsilon\eta}\right)\right] \; ,\\
    \rm deg(\tilde{q})&\in \mathcal{O}\left[\left(\alpha+\log\left(\frac{1}{\epsilon_{\rm trig}}\right)\right)\left(\frac{d\gamma^{1/2}}{\eta^{5/4} }+\log\left(\frac{1}{\epsilon_{\rm polyarc}\eta}\right)\right)+\frac{1}{\eta}\log\left(\frac{1}{\epsilon_{\rm inv-sqrt}\eta}\right)\right] \nonumber\\
    &\subseteq \mathcal{O}\left[\frac{d\gamma^{1/2}}{\eta^{5/4}}\log\left(\frac{1}{\epsilon\eta}\right)+\frac{1}{\eta}\log^2\left(\frac{1}{\epsilon\eta}\right)\right] \; .
    \label{eq: scalingdominatedpolyarccos}
\end{align}
\end{proof}
We are now in a position to prove the rest of Theorem \ref{thm: HQSVT with lower bound}, for implementing H-QSVT for even polynomials and polynomials with no definite parity:
\vspace{4ex}

\noindent \textbf{Proof of Theorem \ref{thm: HQSVT with lower bound} in the general case.}

    Using the dominated approximation polynomials $\tilde{p}$ and $\tilde{q}$ given by Lemma \ref{lem: dominatedapproxgeneral}, we use Chebychev QSP to achieve $P,Q$ given by Theorem \ref{thm: ChebQSPrelaxed} such that $\tilde{p}=\mathfrak{Re}[P]$ and $\tilde{q}=\mathfrak{Re}[Q]$. From Lemma \ref{lem: rescalingcos} we have $\eta\in \mathcal{O}(\frac{\delta^2}{\Lambda^2})\in \mathcal{O}(\delta^2)$ and $\gamma\in \mathcal{O}(1)$. Plugging these parameters in Eq. \eqref{eq: scalingdominatedpolyarccos} gives use the desired scaling. The error can be proven smaller than $\epsilon$ using the same analysis as in the odd case.
    \qed
\end{subsection}

\subsection{Proof of Algorithm \ref{alg:QSVT-based algorithm}}
\label{App:proof_alg2}

\begin{theorem}[Scaling of Algorithm \ref{alg:QSVT-based algorithm}] Algorithm \ref{alg:QSVT-based algorithm} outputs a unitary $U_{\rm approx}$ which approximates the unitary $e^{-iH[f^{\rm SV}(A)]}$, where $H[f^{\rm SV}(A)]$ is an Hamiltonian block-encoding in the canonical form (see Def. \ref{def:standardform}) of the matrix $f_{\rm SV}(A)$ with error $\epsilon$ in the operator norm:
\begin{equation}
    \| U_{\rm approx}-e^{-iH[f^{\rm SV}(A)]}\|_{[\delta,\pi/2]} \leq \mathcal{O}(\epsilon )\,,
\end{equation}
and time complexity $\mathcal{O}[\frac{1}{\delta^{5/2}}\frac{1}{\epsilon^{1/(J-1)}}\log(\frac{1}{\epsilon \delta})]$ if $f$ is $J$-smooth with $J\geq 2$ or $\mathcal{O}[\frac{1}{\delta^{5/2}}\log(\frac{1}{\epsilon \delta})\log(\frac{1}{\epsilon})]$ if $f$ is analytic.
\label{thm: proofscalingQSVTbasedalgorithm}
\end{theorem}
\begin{proof}
We have
    \begin{align}
        \| U_{\rm approx}-e^{-iH[f^{\rm SV}(A)]}\|_{[\delta,\pi/2]} &\leq \| U_{\rm approx}-e^{-iH[p_{\rm SV}(A)]}\|_{[\delta,\pi/2]}+\|e^{-iH[p_{\rm SV}(A)]}-e^{-iH[f^{\rm SV}(A)]} \|_{[\delta,\pi/2]} \\
        &\leq \epsilon+\|H[p^{\rm SV}(A)]-H[f^{\rm SV}(A)]\|_{[\delta,\pi/2]} \\
        & \leq \epsilon+\|p-f\|_{[\delta,\pi/2]}\in \mathcal{O}(\epsilon)\; .
    \end{align}
From the first to the second line, we applied Theorem~\ref{thm: HQSVT with lower bound} to bound the first term and Lemma~7 from \cite{kang2025quantum} to bound the second term. From the second to the third line, we used the fact that the eigenvalues of the Hamiltonian block-encoding $H[A]$ are $\pm\xi$, with $\xi$ the singular values of $A$, which allows us to bound the second term by the infinity norm of the corresponding functions. Finally, by construction, $p$ approximates $f$ to within $\epsilon$, which is achievable because $f$ is $J$-smooth (see Theorem~\ref{thm:Chebtruncerrorsmooth}) or analytic (see Theorem~\ref{thm:Chebtruncerroranalytic}). 

For the time complexity, $U_{\rm approx}$ requires $\tilde{d}$ calls to $e^{-iH[A]/\gamma}$, where $\tilde{d}$ is the time complexity of the H-QSVT protocol. Since $f$ has no definite parity in general, Theorem~\ref{thm: HQSVT with lower bound} gives $\tilde{d} \in \mathcal{O}\!\left(\frac{d}{\delta^{5/2}}\log\!\left(\frac{1}{\epsilon\delta}\right)\right)$.  
From Theorem~\ref{thm:Chebtruncerrorsmooth}, since $f$ is $J$-smooth, we have $d \in \mathcal{O}(\epsilon^{-1/(J-1)})$ for $J < \infty$ and $d \in \mathcal{O}(\log(1/\epsilon))$ for analytic functions. The total time complexity is $\gamma^{-1}\tilde{d}$, and since $\gamma \in \mathcal{O}(1)$, the overall complexity reduces to:
\begin{align}
    &\mathcal{O}\!\left[\frac{1}{\delta^{5/2}\,\epsilon^{1/(J-1)}}\log\!\left(\frac{1}{\epsilon\delta}\right)\right], \quad \text{for $J$-smooth functions with } J \geq 2, \\
    &\mathcal{O}\!\left[\frac{1}{\delta^{5/2}}\log\!\left(\frac{1}{\epsilon}\right)\log\!\left(\frac{1}{\epsilon\delta}\right)\right], \quad \text{for analytic functions}.
\end{align}
\end{proof}

\end{document}